\documentclass[journal,twocolumn]{IEEEtran}

\IEEEoverridecommandlockouts

\usepackage[cmex10]{amsmath} 
\usepackage{amsfonts,acronym,amssymb,amsthm,dsfont}
\usepackage{graphicx}
\usepackage{url}
\usepackage{cite}
\usepackage{hyperref}
\usepackage{braket}
\usepackage{physics}
\usepackage{tabularx}
\usepackage{flushend}
\usepackage{balance}
\usepackage{enumerate}
\usepackage{lipsum}
\usepackage{diagbox}
\usepackage[dvipsnames]{xcolor}
\usepackage{mathrsfs}
\usepackage{bm}
\usepackage{stackengine}
\usepackage{float}
\usepackage{dblfloatfix}

\makeatletter
\def\blfootnote{\xdef\@thefnmark{}\@footnotetext}
\makeatother

\allowdisplaybreaks
\allowbreak

\newcommand{\calA}{\mathcal{A}}

\newcommand{\calB}{\mathcal{B}}

\newcommand{\calC}{\mathcal{C}}
\newcommand{\bbC}{\mathbb{C}}

\newcommand{\calD}{\mathcal{D}}
\newcommand{\bbD}{\mathbb{D}}

\newcommand{\calE}{\mathcal{E}}
\newcommand{\bbE}{\mathbb{E}}

\newcommand{\calH}{\mathcal{H}}
\newcommand{\bbH}{\mathbb{H}}

\newcommand{\bbI}{\mathbb{I}}

\newcommand{\calJ}{\mathcal{J}}

\newcommand{\calL}{\mathcal{L}}

\newcommand{\calM}{\mathcal{M}}

\newcommand{\bbN}{\mathbb{N}}

\newcommand{\calP}{\mathcal{P}}
\newcommand{\bbP}{\mathbb{P}}

\newcommand{\calQ}{\mathcal{Q}}

\newcommand{\calR}{\mathcal{R}}
\newcommand{\bbR}{\mathbb{R}}

\newcommand{\calS}{\mathcal{S}}

\newcommand{\calT}{\mathcal{T}}

\newcommand{\calU}{\mathcal{U}}

\newcommand{\calV}{\mathcal{V}}
\newcommand{\bbV}{\mathbb{V}}

\newcommand{\calX}{\mathcal{X}}

\newcommand{\calY}{\mathcal{Y}}

\newcommand{\calZ}{\mathcal{Z}}

\newcommand{\brk}[1]{\ensuremath{\big[{#1}\big]}}

\newcommand{\indep}{\perp \!\!\! \perp}

\def\on{{\otimes n}}

\DeclareMathOperator{\tra}{Tr}

\newcommand{\den}[2]{\ensuremath{\ket{#1}{\hspace{-1.6mm}}\bra{#2}}}

\newtheorem{theorem}{Theorem}
\newtheorem{corollary}{Corollary}
\newtheorem{definition}{Definition}

\newtheorem{remark}{Remark}

\newtheorem{lemma}{Lemma}[]

\newcommand{\indic}[1]{\ensuremath{\mathds{1}}}
\newcommand{\card}[1]{\ensuremath{\left\lvert{#1}\right\rvert}}   

\newcommand{\sbra}[2]{\ensuremath{\left[{#1}{\,:\,}{#2}\right]}}%
\newcommand{\sbr}[1]{\ensuremath{\left[{#1}\right]}}%
\newcommand{\pr}[1]{\ensuremath{\left({#1}\right)}}%
\newcommand{\br}[1]{\ensuremath{\left\{{#1}\right\}}}%

\newcommand{\Squad}{\hspace{0.5em}}
\acrodef{ACDIS}[ACDIS]{Adaptive Communication Decision and Information Systems}
\acrodef{AEP}{Asymptotic Equipartition Property}
\acrodef{AoA}{Angle of Arrival}
\acrodef{AWGN}{Additive White Gaussian Noise}
\acrodef{AVC}[AVC]{Arbitrarily Varying Channel}
\acrodef{BER}{Bit-Error-Rate}
\acrodef{BEC}{Binary Erasure Channel}
\acrodef{BPSK}{Binary Phase-Shift Keying}
\acrodef{BSC}{Binary Symmetric Channel}
\acrodef{RV}{Random Variable}
\acrodefplural{RV}{Random Variables}
\acrodef{MP}{Multiple Parent}
\acrodef{CR}{Common Randomness}
\acrodef{JTE}{Joint Typical Encoder}
\acrodef{BICM}[BICM]{Bit-Interleaved Coded-Modulation}
\acrodef{CDF}[CDF]{Cumulative Distribution Function}
\acrodef{CGF}[CGF]{Cumulant Generating Function}
\acrodef{CLT}[CLT]{Central Limit Theorem}
\acrodef{CSI}[CSI]{Channel State Information}
\acrodef{DMC}[DMC]{Discrete Memoryless Channel}
\acrodef{DMS}[DMS]{Discrete Memoryless Source}
\acrodef{ERM}[ERM]{Empirical Risk Minimization}
\acrodef{FER}[FER]{Frame Error Rate}
\acrodef{ICA}[ICA]{Independent Component Analysis}
\acrodef{iid}[i.i.d.]{independent and identically distributed}
\acrodef{IoT}[IoT]{Internet of Things}
\acrodef{KKT}[KKT]{Karush-Kuhn Tucker}
\acrodef{LASSO}[LASSO]{Least Absolute Shrinkage and Selection Operator}
\acrodef{LPD}[LPD]{Low Probability of Detection}
\acrodef{LDPC}[LDPC]{Low-Density Parity-Check}
\acrodef{LLMS}[LLMS]{Linear Least Mean Square}
\acrodef{LMS}[LMS]{Least Mean Square}
\acrodef{MAC}[MAC]{multiple-access channel}
\acrodef{MAWTC}[MAC-WTC]{multiple-access wiretap channel}
\acrodef{QMAC}[QMAC]{Quantum Multiple-Access Channel}
\acrodef{MGF}[MGF]{Moment Generating Function}
\acrodef{MLC}[MLC]{Multi-Level Coding}
\acrodef{MLE}[MLE]{Maximum Likelihood Estimate}
\acrodef{MIMO}[MIMO]{Multiple-Input Multiple-Output}
\acrodef{MISO}{Multiple-Input Single-Output}
\acrodef{MSD}[MSD]{Multi-Stage Decoding}
\acrodef{MMSE}[MMSE]{Minimum Mean-Square Error}
\acrodef{PAC}[PAC]{Probably Approximately Correct}
\acrodef{PCA}[PCA]{Principal Component Analysis}
\acrodef{PDF}[PDF]{Probability Density Function}
\acrodefplural{PDF}{Probability Density Functions}
\acrodef{PMF}[PMF]{Probability Mass Function}
\acrodefplural{PMF}{Probability Mass Functions}
\acrodef{PPM}[PPM]{Pulse Position Modulation}
\acrodef{PSD}{Power Spectral Density}
\acrodef{PSK}{Phase Shift Keying}
\acrodef{QKD}{Quantum Key Distribution}
\acrodef{ROC}{Receiver Operating Characteristic}
\acrodef{CVQKD}{Continuous-Variable \ac{QKD}}
\acrodef{QPSK}{Quadrature Phase-Shift Keying}
\acrodef{RV}{random variable}
\acrodef{SIMO}{Single-Input Multiple-Output}
\acrodef{SNR}{Signal-to-Noise Ratio}
\acrodef{SVM}[SVM]{Support Vector Machine}
\acrodef{POVM}{Positive Operator-Valued Measure}
\acrodefplural{POVM}{Positive Operator-Valued Measures}
\acrodef{wrt}[w.r.t.]{with respect to}
\acrodef{WSS}{Wide Sense Stationary}
\acrodef{RHS}{Right Hand Side}
\acrodef{LHS}{Left Hand Side}
\acrodef{CPTP}{Completely Positive and Trace Preserving}
\acrodef{ADSI}[ADSI]{Action-Dependent State Information}

\allowdisplaybreaks
\begin{document}

\title{Quantum Advantage in Multiple Access Wiretap Channels with Entangled Transmitters}

\author{
\IEEEauthorblockN{Hassan ZivariFard and Xiaodong Wang}\\
\thanks{The authors are with the Department of Electrical Engineering, Columbia University, New York, NY 10027. This work is supported in part by the U.S. Office of Naval Research (ONR) under grant N000142112155. E-mails: \{hz2863, xw2008\}@columbia.edu. Part of this work will be presented at the 2026 IEEE Information Theory Workshop.
}
}
\maketitle
\date{}

\begin{abstract}
We investigate secure communication over a classical \ac{MAWTC}, specifically exploring the benefit of shared entanglement between the transmitters. Under the strict semantic security criterion, we derive an achievable rate region and a regularized expression for the secrecy capacity. We further establish a single-letter upper bound on the secure capacity of the \ac{MAWTC} with entangled transmitters. Our results demonstrate that entanglement strictly enlarges the secrecy capacity of \ac{MAWTC} compared to sharing only classical correlated randomness, a finding we illustrate using a pseudo-telepathy game example. Finally, we establish new strong soft-covering lemmas for the output statistics of \acp{MAC} with entangled transmitters. Our results generalize existing results for non-entangled systems.
\end{abstract}

\section{Introduction}
\label{sec:Intro}
Cooperation and security are two central pillars of modern communication systems: cooperation can boost achievable transmission rates, while security constraints often restrict them~\cite{Loock20,Bassoli21,UziCribbing,MAC_GMS,Uzi_Relay,Hayashi06}. 
As the fundamental model for simultaneous communication among multiple senders, the \ac{MAC} plays a central role in the study of communication networks. Likewise, the wiretap channel serves as the canonical model for secure communication \cite{Wyner,BCC:IT78}. The role of quantum resources in communication over \acp{MAC} has been explored across several modeling frameworks. Winter~\cite{QMAC} established a regularized characterization of the classical capacity region for quantum \acp{MAC}. Boche and N\"otzel \cite{Boche14} studied a cooperative classical–quantum \ac{MAC} in which the encoders can conference, i.e., exchange messages with one another at a fixed rate. Hsieh et al. \cite{Hsieh08} and Shi et al. \cite{Shi21} considered a quantum \ac{MAC} in which both transmitters independently share entanglement resources with the receiver. Chou \cite{QMAC_Security} investigated secure communication over a quantum \ac{MAC}. The authors of the present paper studied covert communication over a quantum \ac{MAC} with general message sets \cite{MAC_GMS}.

Earlier work has shown that enriching multi-user systems with auxiliary resources, such as encoder cooperation, can boost achievable communication rates \cite{WillemsConferencing}. This naturally leads to a broader question: can shared entanglement between the transmitters, even when the channel and all other components remain fully classical, offer meaningful advantages?

In \cite{Quek17}, Quek and Shor illustrated the benefits of entangled transmitters by examining rate separations. They considered a specific interference channel example based on the CHSH game \cite{Clauser69}, with a primary focus on super-quantum non-local correlations arising from the PR-box model introduced by Popescu and Rohrlich \cite{Popescu94}. Leditzky et al.~\cite{Leditzky_20,Seshadri23} demonstrated through examples that entanglement shared between transmitters can strictly increase the achievable sum rate of a classical \ac{MAC} (see also \cite{Notzel201,Notzel202,Yun23}). The examples constructed in \cite{Leditzky_20} are based on pseudo-telepathy games \cite{Brassard05}, in which quantum strategies achieve guaranteed success and surpass the performance of classical strategies. Further insights and extensions are provided in \cite{Doolittle22}. Fawzi and Ferm\'e \cite{Fawzi22} further demonstrated a separation in the more fundamental setting of the binary adder channel, assuming the transmitters have access to non-signaling correlations \cite{Quek17}. Recently, Pereg et al.~\cite{Uzi_Entangled_MAC} generalized the findings of \cite{Leditzky_20} to arbitrary \ac{MAC} scenarios with entangled transmitters. They derived a general achievable rate region, provided a regularized characterization of the capacity region, and also investigated \ac{MAC} models with both conferencing and entangled transmitters. Very recently, Hawellek et al.~\cite{Interference_Entangled} studied the effect of entanglement shared among the transmitters in the interference channel, deriving general inner and outer bounds on the capacity region and presenting another example that highlights the benefits of shared entanglement between the transmitters. See also the related development in~\cite{Aghaee25}.

The classical \ac{MAWTC} has been studied extensively in the literature \cite{GMAWCJamming,YassaeeMAWC,Wiese2012,WieseBoche,MultiCase_Paper,ICC2014,IETpaper,ForenPaper,Frey18,XU24}. Tekin and Yener examined the Gaussian \ac{MAWTC} in~\cite{GMAWCJamming}, where they derived an achievable rate region under the weak secrecy criterion. Yassaee and Aref~\cite{YassaeeMAWC} investigated the discrete memoryless \ac{MAWTC} under the strong secrecy criterion. Their analysis also introduced output-statistics techniques for the \ac{MAC} and established two soft-covering lemmas, which form essential components of their achievability scheme for the \ac{MAWTC}. Wise and Boche~\cite{Wiese2012} examined the discrete memoryless \ac{MAWTC} with a common message under the strong security metric. ZivariFard et al.~\cite{MultiCase_Paper} studied \ac{MAWTC} under the tradeoff between the rate of random numbers needed to realize the stochastic encoders and the rates of confidential messages under the weak security criterion. Frey et al.~\cite{Frey18} further advanced the study of \ac{MAWTC} by adopting the semantic security criterion, developing an achievability scheme that relies on a strong soft-covering lemma for the \ac{MAC}.

In this work, we investigate secure communication over classical \acp{MAWTC} when the transmitters share entanglement resources. In our model, two senders wish to convey confidential messages to a common receiver while ensuring semantic security against an external eavesdropper. The key feature of our setup is that the transmitters are allowed to utilize pre-shared entanglement as a resource, as depicted in Fig.~\ref{fig:System_Model}.  
We present general inner and outer bounds for this problem, and we present an example showing that shared entanglement between the transmitters can strictly enlarge the capacity of the \ac{MAWTC} when the transmitters only share classical \ac{CR}. Our achievable rate region recovers, as a special case, several known results, including secure communication over classical \acp{MAWTC} under the strong and semantic security constraints \cite{YassaeeMAWC, Frey18} and communication over classical \acp{MAC} when the transmitters share entanglement resources \cite{Uzi_Entangled_MAC}. We further obtain a regularized expression that characterizes the secrecy capacity of \ac{MAWTC}. We also study output statistics of classical \acp{MAC} with entangled transmitters. We present two strong soft-covering lemmas for output statistics of classical \acp{MAC} with entangled transmitters, which ensure that the variational distance between the output distribution and approximated output distribution is doubly-exponentially decreasing, enabling
semantic security through the union bound. 
To the best of our knowledge, our results are the first to show that shared entanglement can offer tangible benefits in secure multi-user communication settings.

Our achievability scheme integrates random coding, coded time-sharing \cite{ElGamalKim}, wiretap coding \cite{Wyner,BCC:IT78}, and channel resolvability techniques \cite{Hayashi06,Bloch2013,effectivesecrec,Cuff15,Cuff16,ZivSemantic2016,Ziv_WTC_with_NonCausaul_CSI}. Compared with the achievability scheme in \cite{Uzi_Entangled_MAC}, a key challenge in our setting is ensuring semantic security against the eavesdropper, which we address by developing new strong soft-covering lemmas tailored to \acp{MAC} with entangled transmitters. Relative to the schemes in \cite{YassaeeMAWC,Frey18}, an additional challenge lies in leveraging the shared entanglement simultaneously for reliable communication, secrecy, and channel resolvability. Furthermore, compared with \cite{Uzi_Entangled_MAC,YassaeeMAWC,Frey18}, another challenge is to demonstrate, via an explicit example, that shared entanglement can strictly enlarge the set of semantically secure achievable rates.

The remainder of the paper is organized as follows. In Section~\ref{sec:Problem_Statement}, we introduce the problem formulation for the classical \ac{MAWTC} with entangled transmitters. Section~\ref{sec:Main_MACWTC} presents our main results on the secrecy capacity of the \ac{MAWTC} in the presence of shared entanglement. In Section~\ref{sec:Output_Statistics_MAC}, we establish two strong soft-covering lemmas for the output statistics of a classical \ac{MAC} with entangled transmitters. Concluding remarks are provided in Section~\ref{sec:Conclusion}. All proofs are deferred to the appendices.

\textit{Notation:}
Let $\bbN$ be the set of natural numbers and $\bbR$ be the set of real numbers. For any $x, y\in\bbR$, define $\sbra{x}{y} \triangleq [\lfloor x\rfloor, \lceil y \rceil]\cap \bbN$ and $[x] \triangleq \sbra{1}{x}$. Script letters such as $\calX$ denote finite sets. Uppercase letters, e.g., $X$, represent classical \acp{RV}, while lowercase letters, e.g., $x$, denote their realizations. For $n\in\bbN$, we write $x^n = \pr{x(i)}_{i\in[n]}$ for a length-$n$ sequence over $\calX$. Also, we write $X_i^j$ for a sequence $\pr{X(i), X(i+1),\cdots, X(j)}$ over $\calX$. The distribution of a discrete \ac{RV} $X$ is denoted by the \ac{PMF} $p_X(x)$ on the finite alphabet $\calX$. Let $p_X$ and $q_X$ be \acp{PMF} on a finite set $\calX$. 
The total variation distance between $p_X$ and $q_X$ is defined as $\bbV(p_X,q_X) \triangleq \frac{1}{2}\sum_{x\in\calX} \abs{p_X(x)-q_X(x)}$. 
The Kullback--Leibler (KL) divergence between $p_X$ and $q_X$ is given by $\bbD(p_X\lVert q_X) \triangleq \sum_{x\in\calX} p_X(x)\log\frac{p_X(x)}{q_X(x)}$. 
For any $\alpha>0$ with $\alpha\neq1$, the R\'enyi divergence of order $\alpha$ of $p_X$ from $q_X$ is defined as $\bbD_\alpha(p_X\lVert q_X) \triangleq \frac{1}{\alpha-1}
\log\sum_{x\in\calX} p_X(x)^\alpha q_X(x)^{1-\alpha}$. 
The support of a probability distribution $p$ is denoted by $\mathrm{supp}(p)$. The notation $p_X^\on(x^n)$ refers to the \ac{iid} product distribution $\prod_{i=1}^n p_X(x(i))$. $\indic{1}_{\{\cdot\}}$ denotes the indicator function and $\exp\pr{\cdot}$ denotes the exponential function. The set of $\epsilon$-strongly typical sequences of length $n$ \ac{wrt} $P_X$ is defined as,
    \begin{align}
        &\calT_\epsilon^{(n)}\triangleq\nonumber\\
        &\left\{x^n\in\calX^n:\left|\frac{\textup{Nu}\left(\bar{x}|x^n\right)}{n}-P_X(\bar{x})\right|\le\epsilon P_X(\bar{x}),\forall\bar{x}\in\calX\right\},\nonumber
    \end{align}where $\textup{Nu}\left(\bar{x}|x^n\right)=\sum_{i=1}^n\indic{1}_{\{x_i=\bar{x}\}}$.
We denote the set of positive semi-definite operators on a finite-dimensional Hilbert space $\calH$ by $\calP(\calH)$ and denote the set of quantum states by $\calD(\calH)\triangleq\{\rho\in\calP(\calH):\tra[\rho]=1\}$.  
We also denote the space of the bounded linear operators on $\calH$ with $\calL(\calH)$. The identity operator on some Hilbert space $\calH$ is denoted by $\mathds{I}$. We use $\rho$ to denote a quantum density matrix and use $\ket{\rho}$ to denote the associated quantum~state. 
\begin{figure*}[b!]
    \hrulefill
\begin{align}
    \mathsf{e}(C_n)&=\frac{1}{\abs{\calM_1}\abs{\calS_1}\abs{\calM_2}\abs{\calS_2}}\sum_{(m_1,s_1)\in\calM_1\times\calS_1}\sum_{(m_2,s_2)\in\calM_2\times\calS_2}\sum_{\pr{x_1^n,x_2^n}\in\calX_1^n\times\calX_2^n}f\pr{x_1^n,x_2^n\lvert m_1,s_1,m_2,s_2}\nonumber\\
    &\qquad\times\sum_{\substack{y^n\in\calY^n:g(y^n)\ne\pr{m_1,m_2}}}W_{Y|X_1X_2}^\on(y^n|x_1^n,x_2^n).\label{eq:probaility_error}
\end{align}
\end{figure*}
\section{Problem Statement for \texorpdfstring{\ac{MAWTC}}{MAC-WTC}}
\label{sec:Problem_Statement}
\begin{figure*}[t!]
\centering
\includegraphics[width=12.0cm]{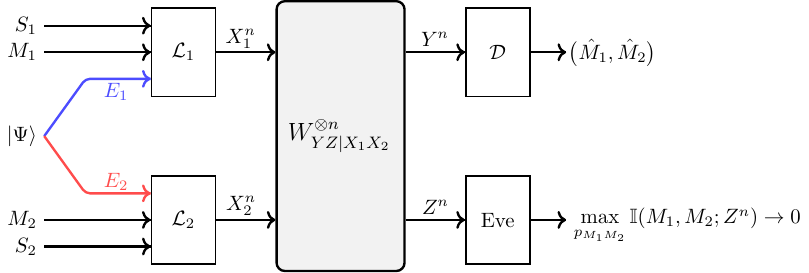}
\caption{Multiple Access Wiretap Channel with Entangled Transmitters. Each Encoder~$i$, for $i\in\{1,2\}$, encodes its classical message $M_i$ into the channel input $X_i^n$ while ensuring that the message remains confidential from the eavesdropper. Encoder~$i$ takes as input the message $M_i$, local randomness $S_i$, its share of the quantum entangled state $E_i$, and a set of \acp{POVM} $\calL_i$, and produces the channel input $X_i^n$.
}
\label{fig:System_Model}
\vspace{-0.15in}
\end{figure*}

Formally, a discrete memoryless \ac{MAWTC} $\pr{\calX_1,\calX_2,W_{YZ\lvert X_1X_2},\calY,\calZ}$ consists of two finite input alphabets $\calX_1$ and $\calX_2$, two finite output alphabets $\calY$ and $\calZ$, and transition probabilities $W_{YZ\lvert X_1X_2}$.
Fig.~\ref{fig:System_Model} depicts a communication system in which two transmitters communicate securely over a classical channel $W_{YZ\lvert X_1X_2}$ with the aid of shared entanglement. The entanglement resources are denoted by $E_1$ and $E_2$, corresponding to the first and second transmitters, respectively, and their joint Hilbert space is given by $\calH_{E_1E_2}\triangleq\calH_{E_1}\otimes\calH_{E_2}$. We define codes as follows.
\begin{definition}[Secrecy Code Components for \ac{MAWTC} with Entangled Transmitters]
\label{defi:Code_Components_Secure_Comm}
A sequence of $(2^{nR_1},2^{nR_2},n)$ codes $\br{C_n}_{n\in\bbN}$ for the \ac{MAWTC} $W_{YZ\lvert X_1X_2}$ with entangled transmitters consists of the following:
\begin{itemize}
    \item message sets $\calM_1\triangleq\left[\left\lfloor2^{nR_1}\right\rfloor\right]$ and $\calM_2\triangleq\left[\left\lfloor2^{nR_2}\right\rfloor\right]$;
    \item a bipartite state $\Psi_{E_1E_2}\in\calD(\calH_{E_1E_2})$, which is shared between the two transmitters;
    \item two local randomness sources, $\pr{\calS_1,p_{S_1}}$ and $\pr{\calS_2,p_{S_2}}$; 
    \item two collections of encoding \acp{POVM} 
    \begin{align*}
        \calL_1^{(m_1,s_1)}&=\br{L_{x_1^n}^{(m_1,s_1)}}_{x_1^n\in\calX_1^n},\nonumber\\
        \calL_2^{(m_2,s_2)}&=\br{L_{x_2^n}^{(m_2,s_2)}}_{x_2^n\in\calX_2^n},
    \end{align*}acting on $E_1$ and $E_2$, respectively, with one pair of \acp{POVM} specified for each tuple $(m_1,s_1,m_2,s_2)\in\calM_1\times\calS_1\times\calM_2\times\calS_2$;
    \item a decoding function $g:\calY^n\to\calM_1\times\calM_2\cup\br{\mathfrak{e}}$, which assigns each channel observation $y^n$ either to a message pair $(m_1,m_2)\in \calM_1\times\calM_2$ or to an error symbol $\mathfrak{e}$.
\end{itemize}
\end{definition}
The transmitters share an entangled state $\Psi_{E_1E_2}\in\calD(\calH_{E_1E_2})$. For $i\in\br{1,2}$, given local randomness $s_i\in\calS_i$ and a message $m_i\in\calM_i$, the Transmitter~$i$ applies the encoding measurement $\calL_i^{(m_i,s_i)}$ to its share of the entanglement resource $E_i$, and sends the resulting outcome $x_i^n\in\calX_i^n$ across the channel over $n$ uses. Therefore, the joint input distribution is
\begin{align*}
    &f\pr{x_1^n,x_2^n\lvert m_1,s_1,m_2,s_2}\nonumber\\
    &\quad=\tra\sbr{\pr{L_{x_1^n}^{(m_1,s_1)}\otimes L_{x_2^n}^{(m_2,s_2)}}\Psi_{E_1E_2}^\on}.
\end{align*}
Bob observes the channel output $y^n$ and forms an estimate of the transmitted message pair as $(\hat{m}_1,\hat{m}_2) \triangleq g(y^n)$. 
The code is known by all the terminals, and the objective is to design a reliable and semantically secure code. Note that if the joint state $\Psi_{E_1E_2}$ is a product state, then Definition~\ref{defi:Code_Components_Secure_Comm} reduces to secrecy code components for a classical \ac{MAC} $W_{YZ|X_1X_2}$ without the use of quantum resources. 
From Definition~\ref{defi:Code_Components_Secure_Comm} the average probability of error is defined as \eqref{eq:probaility_error} at bottom of the page.

\begin{definition}[Secrecy Code]
\label{defi:Code_Secure_Comm}
    A rate pair $(R_1,R_2)$ is said to be achievable for the \ac{MAWTC} $W_{YZ\lvert X_1X_2}$ if there exists a sequence of $\pr{2^{nR_1},2^{nR_2},n}$ codes $\br{C_n}_{n\in\bbN}$ such that
    \begin{subequations}\label{eq:Reliability_Security}
    \begin{align}
        \lim_{n\to\infty}\mathsf{e}(C_n)&=0,\label{eq:Reliability}\\
        \lim_{n\to\infty}\max_{P_{M_1,M_2}}\bbI_{C_n}(M_1,M_2;Z^n)&=0.\label{eq:Security}
    \end{align}
    \end{subequations}
\end{definition}The secrecy capacity region of the \ac{MAWTC} $W_{YZ|X_1X_2}$ with entangled transmitters, denoted by $\calC_{\textup{ET-MAC-WTC}}\pr{W_{YZ|X_1X_2}}$, is defined as the closure of all achievable rate pairs $(R_1,R_2)$.
\begin{remark}[Message Distribution]
    In the definition of the average error probability in \eqref{eq:probaility_error}, we assume that the messages intended for the legitimate receiver are uniformly distributed. This choice is standard in capacity analysis and corresponds to the worst-case scenario for reliability. On the other hand, the security constraint in \eqref{eq:Security} is defined in a distribution-independent manner, ensuring that the secrecy guarantee holds for \emph{any} message distribution. Hence, while the reliability analysis assumes uniform messages, this does not restrict the generality of the security criterion.
\end{remark}
In this paper, we show that the secrecy capacity region of the \ac{MAWTC} $W_{YZ|X_1X_2}$ with entangled transmitters, i.e., $\calC_{\textup{ET-MAC-WTC}}\pr{W_{YZ|X_1X_2}}$, strictly improves upon the secrecy capacity region of the \ac{MAWTC} $W_{YZ|X_1X_2}$ with classical private \ac{CR}, which is defined in the following.
\begin{definition}[Secrecy Code Components for \ac{MAWTC} with \ac{CR}]
\label{defi:Code_Components_Classical_CR}
A sequence of $(2^{nR_1},2^{nR_2},n)$ codes $\br{C_n}_{n\in\bbN}$ for the \ac{MAWTC} $W_{YZ\lvert X_1X_2}$ with \ac{CR} shared between the transmitters consists of the following:
\begin{itemize}
    \item message sets $\calM_1\triangleq\left[\left\lfloor2^{nR_1}\right\rfloor\right]$ and $\calM_2\triangleq\left[\left\lfloor2^{nR_2}\right\rfloor\right]$;
    \item a source of unlimited common randomness $(\calS,p_S)$ shared between the transmitters and kept secret from both the legitimate receiver and the eavesdropper;
    \item two local randomness sources, $\pr{\calS_1,p_{S_1}}$ and $\pr{\calS_2,p_{S_2}}$; 
    \item two encoding functions $g_1:\calM_1\times\calS\times\calS_1\to\calX_1^n$ and $g_2:\calM_2\times\calS\times\calS_2\to\calX_2^n$, corresponding to the first and second transmitters, respectively;
    \item a decoding function $g:\calY^n\to\calM_1\times\calM_2\cup\br{\mathfrak{e}}$, which assigns each channel observation $y^n$ either to a message pair $(m_1,m_2)\in \calM_1\times\calM_2$ or to an error symbol $\mathfrak{e}$.
\end{itemize}
\end{definition}
The probability of error and the secrecy code for the problem defined above can be specified in the same manner as in \eqref{eq:probaility_error} and Definition~\ref{defi:Code_Secure_Comm}, respectively. The secrecy capacity region of the \ac{MAWTC} $W_{YZ|X_1X_2}$ with \ac{CR} shared among the transmitters, denoted by $\calC_{\textup{CR-MAC-WTC}}\!\pr{W_{YZ|X_1X_2}}$, is defined as the closure of all achievable rate pairs $(R_1,R_2)$.
\section{Main Results for \texorpdfstring{\ac{MAWTC}}{MAC-WTC} with Entangled Transmitters}
\label{sec:Main_MACWTC}
In this paper, we establish an achievable rate region for the \ac{MAWTC} with entangled transmitters, derive a regularized expression for its secrecy capacity, and present a single-letter upper bound on $\calC_{\textup{ET-MAC-WTC}}\pr{W_{YZ|X_1X_2}}$. We then present an example showing that shared entanglement can strictly enlarge the secrecy capacity region relative to the setting in which the transmitters share only classical \ac{CR}.
\subsection{Semantic Security for \texorpdfstring{\ac{MAWTC}}{MAC-WTC} with Entangled Transmitters}
\begin{theorem}
\label{thm:Achievable}
Let
\begin{subequations}\label{eq:Regions}
\begin{align}
\calB_1&\triangleq 
\nonumber\\
&\left\{\hspace{-2mm} \begin{array}{rl}
  (R_1,R_2):\;
    R_1&\hspace{-2.5mm}<\bbI(U_1;Y\lvert U_0,U_2)-\bbI(U_1;Z\lvert U_0),\\
    R_2&\hspace{-2.5mm}<\bbI(U_2;Y\lvert U_0,U_1)-\bbI(U_2;Z\lvert U_0),\\
    R_1+R_2&\hspace{-2.5mm}<\bbI(U_1,U_2;Y\lvert U_0)-\bbI(U_1,U_2;Z\lvert U_0),
	\end{array}
\hspace{-2mm}\right\},
\label{eq:Region_1}\\
\calB_2&\triangleq 
\Big\{ (R_1,0):\;
    R_1<\bbI(U_1;Y\lvert U_0,U_2)-\bbI(U_1;Z\lvert U_0,U_2)\Big\},
\label{eq:Region_2}\\
\calB_3&\triangleq 
\Big\{ (0,R_2):\;
    R_2<\bbI(U_2;Y\lvert U_0,U_1)-\bbI(U_2;Z\lvert U_0,U_1)\Big\}.
\label{eq:Region_3}
\end{align}
Then we have
\begin{align}
  \bigcup\limits_{p_{U_0U_1U_2X_1X_2YZ}\in\mathfrak{P}}\hspace{-4mm}\br{\calB_1\cup\calB_2\cup\calB_3}\subseteq\calC_{\textup{ET-MAC-WTC}}\pr{W_{YZ|X_1X_2}},\label{eq:thm_Achievable}
\end{align} where $\mathfrak{P}$ is the set of \acp{PMF} that factor as
\begin{align}
    &p_{U_0U_1U_2X_1X_2YZ}(u_0,u_1,u_2,x_1,x_2,y,z )=\nonumber\\
    &\quad p_{U_0}(u_0)p_{U_1\lvert U_0}(u_1\lvert u_0)p_{U_2\lvert U_0}(u_2\lvert u_0)\nonumber\\
    &\quad\times\tra\br{\pr{L_1(x_1\lvert u_0,u_1)\otimes L_2(x_2\lvert u_0,u_2)}\Psi_{E_1E_2}}\nonumber\\
    &\quad\times W_{YZ\lvert X_1X_2}(y,z\lvert x_1,x_2).\label{eq:Joint_Dist_Single_Letter_1}
\end{align}
\end{subequations}
\end{theorem}Theorem~\ref{thm:Achievable} is proved in Appendix~\ref{proof:thm:Achievable}. The positive mutual information quantities in Theorem~\ref{thm:Achievable} represent the communication rates achievable over the \ac{MAC}, while the negative mutual information quantities capture the penalty rates that the transmitters must incur to ensure secure communication. The region $\calB_1$ corresponds to the regime in which both transmitters aim to communicate their messages confidentially to the legitimate receiver, whereas the regions $\calB_i$, for $i\in\{2,3\}$, correspond to regimes in which only Transmitter~$i-1$ seeks to transmit its message confidentially. The proof of $\calB_1$ in Theorem~\ref{thm:Achievable} relies on coded time-sharing, Wyner's wiretap coding, for each transmitter, and a strong soft-covering lemma for \acp{MAC} to approximate the output 
distribution corresponding to both transmitters sending \ac{iid} codewords, thereby ensuring 
semantic security via the union bound. 
Similarly, the proof of $\calB_i$, for $i\in\{2,3\}$, in Theorem~\ref{thm:Achievable} 
is based on coded time-sharing and Wyner's wiretap coding, but employs a different strong 
soft-covering lemma for \acp{MAC} to approximate the output distribution when only Transmitter~$i-1$ 
sends an \ac{iid} codeword. The main distinction between the achievability proof of Theorem~\ref{thm:Achievable} and the achievability proofs in 
\cite{YassaeeMAWC,Wiese2012,Frey18} for classical \ac{MAWTC}s without shared entanglement lies in the generation of the channel inputs. In our scheme, we first generate a time-sharing sequence represented by the auxiliary \ac{RV} $U_0$. Conditioned on $U_0$, Transmitter~$i$, for $i\in\{1,2\}$, superimposes a wiretap codebook to encode the message $M_i$. The channel inputs $X_1$ and $X_2$ are then generated via encoding \acp{POVM} that depend on the classical auxiliary \acp{RV} $(U_0,U_1)$ and $(U_0,U_2)$, respectively. Because these \acp{POVM} act on entangled states shared between the transmitters, they can induce correlations between $X_1$ and $X_2$, which are deliberately exploited to harness the advantages provided by entanglement. Whereas in \cite{YassaeeMAWC,Wiese2012,Frey18}, the channel input of Transmitter~$i$ is generated via channel prefixing conditioned on $(U_0,U_i)$, this mechanism cannot induce correlation between the channel inputs when conditioned on the classical auxiliary \acp{RV} $(U_0,U_i)$.

The achievable rate region stated in Theorem~\ref{thm:Achievable} closely parallels the regions established in \cite{YassaeeMAWC,Wiese2012,Frey18}. The key difference is that our achievable rate region includes an additional optimization over both the encoding \acp{POVM} and the shared entangled quantum states. As a result, any potential performance gains arising from entanglement are not immediately apparent from the rate expressions alone. To clearly demonstrate the benefit of shared entanglement, a concrete example is therefore required, which we present in Section~\ref{sec:Example}.

\begin{remark}[Semantic Security for \ac{MAWTC} without Entanglement]
When the encoders do not share entanglement prior to communication, i.e., when $\Psi_{E_1E_2}=\Psi_{E_1}\otimes\Psi_{E_2}$, the region $\calB_1$ in \eqref{eq:Region_1} reduces to \cite[Theorem~10]{Frey18}. In this case, the auxiliary systems $\Psi_{E_1}$ and $\Psi_{E_2}$ carry no correlations, and the corresponding \ac{POVM} $\calL_1 \otimes \calL_2$ reduces to trivial local measurements that can be absorbed into the encoding operations.
\end{remark}
\begin{remark}[Semantic Security for \ac{MAWTC} without Entanglement]
When the encoders do not share entanglement prior to communication, that is, when 
$\Psi_{E_1E_2}=\Psi_{E_1}\otimes\Psi_{E_2}$ and $U_0=\emptyset$, the achievable region in 
Theorem~\ref{thm:Achievable} reduces to the region in~\cite[Theorem~1]{YassaeeMAWC}, but under a stronger security constraint.
\end{remark}
\begin{remark}[Communication Over a \ac{MAC} with Entanglement]
    By setting  $Z=\emptyset$, the achievable rate region described in Theorem~\ref{thm:Achievable} reduces to the achievable rate region derived for communication over the classical \ac{MAC} with entangled transmitters, presented in \cite[Theorem~2]{Uzi_Entangled_MAC}.
\end{remark}

\begin{lemma}[Cardinality Bounds and Pure Entangled States]
\label{lemma:Pure_Shared_States}
The union of the inner bound in~\eqref{eq:thm_Achievable} can be achieved using auxiliary variables $U_0, U_1, U_2$ satisfying $\abs{\calU_0}\le 3$ and $\abs{\calU_i} \le 3\abs{\calX_1}\abs{\calX_2}$ for $i\in\{1,2\}$, together with pure shared states $\rho_{E_1E_2} \triangleq \den{\Psi_{E_1E_2}}{\Psi_{E_1E_2}}$. That is, restricting the union to such auxiliary variables and pure states does not reduce the achievable region.
\end{lemma}
The proof of Lemma~\ref{lemma:Pure_Shared_States} relies on the Fenchel--Eggleston--Carath\'eodory theorem~\cite{Eggleston} together with a perturbation argument~\cite{Perurbation,Perturbation_Quantum}. The full details are provided in Appendix~\ref{proof:lemma:Pure_Shared_States}.
\begin{figure*}[b!]
    \hrulefill
    \setcounter{equation}{4}
\begin{align}
\bigcup\limits_{\substack{p_{U_0}p_{U_1U_2\lvert U_0},\,\calL_1\otimes\calL_2,\,\Psi_{E_1E_2}}}\left\{ \begin{array}{l}
(R_1,R_2)\in\bbR_+^2: \\ 
 R_1 \le \min\{\bbI(U_1;Y\lvert U_0)-\bbI(U_1;Z\lvert U_0),\\
 \qquad\bbI(U_1;Y\lvert U_0,U_2)-\bbI(U_1;Z\lvert U_0,U_2)\} \\ 
 R_2 \le \min\{\bbI(U_2;Y\lvert U_0)-\bbI(U_2;Z\lvert U_0),\\
 \qquad\bbI(U_2;Y\lvert U_0,U_1)-\bbI(U_2;Z\lvert U_0,U_1)\} \\ 
 R_1 + R_2 \le \bbI(U_1,U_2;Y\lvert U_0)-\bbI(U_1,U_2;Z\lvert U_0) \\ 
 \end{array} \right\},\label{eq:Upper_Region_ET}
\end{align}
    \setcounter{equation}{3}
\end{figure*}

We now establish a regularized expression for the secrecy capacity region of the \ac{MAWTC} with entangled transmitters.

\begin{theorem}
\label{thm:NLetter_Capacity}
Let
\begin{align*}
\calB_1^{(n)}&\triangleq\nonumber\\
&\hspace{-4mm}\left\{ \begin{array}{rl}
  (R_1,R_2):\;
    R_1&\hspace{-2.5mm}<\bbI(U_1^n;Y^n\lvert U_2^n)-\bbI(U_1^n;Z^n),\\
    R_2&\hspace{-2.5mm}<\bbI(U_2^n;Y^n\lvert U_1^n)-\bbI(U_2^n;Z^n),\\
    R_1+R_2&\hspace{-2.5mm}<\bbI(U_1^n,U_2^n;Y^n)-\bbI(U_1^n,U_2^n;Z^n),
	\end{array}
\hspace{-2mm}\right\},
\\
\calB_2^{(n)}&\triangleq 
\Big\{ (R_1,0):\;
    R_1<\bbI(U_1^n;Y^n\lvert U_2^n)-\bbI(U_1^n;Z^n\lvert U_2^n)\Big\},\\
\calB_3^{(n)}&\triangleq 
\Big\{ (0,R_2):\;
    R_2<\bbI(U_2^n;Y^n\lvert U_1^n)-\bbI(U_2^n;Z^n\lvert U_1^n)\Big\}.
\end{align*}
Then we have
\begin{align}
  &\frac{1}{n}\bigcup_{n=1}^\infty\,\,\,\bigcup\limits_{p_{U_1^n}p_{U_2^n},\Psi_{E_1E_2}^\on,\calL_1\otimes\calL_2}\br{\calB_1^{(n)}\cup\calB_2^{(n)}\cup\calB_3^{(n)}}\nonumber\\
  &=\calC_{\textup{ET-MAC-WTC}}\pr{W_{YZ|X_1X_2}},\label{eq:Capacity_Region_nLetter}
\end{align}where the cardinality of the auxiliary \acp{RV} are restricted to $\abs{\calU_i} \le \abs{\calX_1}\abs{\calX_2}$ for $i\in\{1,2\}$.
\end{theorem}The achievability proof of Theorem~\ref{thm:NLetter_Capacity} follows directly from the achievability proof of Theorem~\ref{thm:Achievable} by setting $U_0=\emptyset$. The converse proof of Theorem~\ref{thm:NLetter_Capacity} is provided in Appendix~\ref{proof:thm:NLetter_Capacity}. Similar to Theorem~\ref{thm:Achievable}, the auxiliary random sequence $U_i^n$, for $i\in\{1,2\}$, represents the coded version of the message~$M_i$.

We now present a general single-letter upper bound on the secrecy capacity.
\begin{theorem}
\label{thm:Single_Letter_Outter}
The secrecy capacity region $\calC_{\textup{ET-MAC-WTC}}\pr{W_{YZ|X_1X_2}}$ is included in \eqref{eq:Upper_Region_ET} at the bottom of the page,  
where the cardinality of the auxiliary \acp{RV} are restricted to $\abs{\calU_0}\le 5$ and $\abs{\calU_i} \le 5\pr{\abs{\calX_1}\abs{\calX_2}+1}$ for $i\in\{1,2\}$.
\end{theorem}The proof of Theorem~\ref{thm:Single_Letter_Outter} adapts the techniques developed in \cite{BCC:IT78,Uzi_Entangled_MAC} to the problem considered in this paper and is provided in Appendix~\ref{proof:thm:Single_Letter_Outter}. Similar to Theorem~\ref{thm:Achievable}, the auxiliary \ac{RV} $U_0$ can be interpreted as a time-sharing \ac{RV}, while the auxiliary \acp{RV} $U_i$, for $i\in\{1,2\}$, represent the messages $M_i$. The proofs of the corresponding cardinality bounds follow arguments similar to those used in Lemma~\ref{lemma:Pure_Shared_States} and are omitted for brevity.

\begin{remark}[Unlimited Entanglement Resources]
    In this section, we analyze the capacity region under the assumption of unlimited entanglement between the transmitters. We emphasize that, in all results presented thus far, the dimensions of the entangled systems $E_1$ and $E_2$ are unbounded. Since the unions in \eqref{eq:thm_Achievable} and \eqref{eq:Capacity_Region_nLetter} range over the closure of the set of finite-dimensional quantum correlations, the regions also include correlations that arise as $\dim(\calH_{E_t}) \to \infty$. This broader class of correlations has been investigated in~\cite{Slofstra20}.
\end{remark}

\subsection{Entanglement Resources with Limited Rate}
\label{sec:Limited_Entanglement}
Now, we consider a setting where the entanglement
resources are limited. A code with rate-limited entanglement resources is defined such that the encoders have access to Hilbert spaces that grow as a function of the channel uses in a rate-limited fashion. The precise definition is similar to Definition~\ref{defi:Code_Components_Secure_Comm}, with the only modification in the second bullet point, which now reads:
\begin{itemize}
    \item a bipartite quantum state $\Psi_{E_1E_2} \in \calD(\calH_{E_1E_2})$ shared between the two transmitters, where the local Hilbert space dimensions satisfy $\dim(\calH_{E_t}) \le 2^{n\theta_E}$ for $t \in\sbr{1,2}$.
\end{itemize}

The communication scheme follows the same structure as in Section~\ref{sec:Problem_Statement}. 
A rate pair $(R_1,R_2)$ is said to be achievable with entanglement rate $\theta_E$ if, for every 
$\epsilon>0$ and all sufficiently large block lengths $n$, there exists a 
$\pr{2^{nR_1},2^{nR_2},n}$ code in which the transmitters share entanglement at rate $\theta_E$, 
and both the average error probability, $\lim_{n\to\infty}\mathsf{e}(C_n)=0$,
and the semantic security condition,
$\lim_{n\to\infty}\max_{P_{M_1,M_2}} 
\bbI_{C_n}(M_1,M_2;Z^n)=0$, are satisfied.
The capacity region $\calC_{\textup{ET-MAC-WTC}}\!\pr{W_{YZ|X_1X_2},\theta_E}$ with 
rate-limited entanglement between the transmitters is defined as the set of all achievable 
rate pairs $(R_1,R_2)$ under a shared entanglement rate of $\theta_E$.
\begin{corollary}[Semantically Secure Achievable Rates with Rate-Limited Entanglement]
The secrecy capacity of \ac{MAWTC} satisfies
    \begin{align*}
  &\bigcup\limits_{p_{U_0}p_{U_1|U_0}p_{U_2|U_0},\Psi_{E_1E_2},\calL_1\otimes\calL_2}\br{\calB_1\cup\calB_2\cup\calB_3}\nonumber\\
  &\qquad\subseteq\calC_{\textup{ET-MAC-WTC}}\pr{W_{YZ|X_1X_2},\theta_E},
\end{align*}where $\calB_1,\calB_2$, and $\calB_3$ are defined in \eqref{eq:Region_1}, \eqref{eq:Region_2}, and \eqref{eq:Region_3}, respectively, and the states $\Psi_{E_1E_2}$ are restricted to be pure states such that $\bbH(E_1)_\Psi=\bbH(E_2)_\Psi\le\theta_E$.
\end{corollary}The proof proceeds along the same lines as that of Theorem~\ref{thm:Achievable} and is therefore omitted for brevity. We also note that the proof of the entangled states can be restricted to the pure states in Lemma~\ref{lemma:Pure_Shared_States}, and the cardinality bounds established in Lemma~\ref{lemma:Pure_Shared_States} apply in this setting as well.

\subsection{Example}
\label{sec:Example}
In this section, we present an example of a classical \ac{MAWTC} for which sharing entanglement between the transmitters strictly increases the achievable sum rate. 
\subsubsection{Mermin-Peres Magic Square Game}
In this example, the channel is defined in terms of a pseudo-telepathy game, i.e., a non-local game in which quantum strategies outperform classical strategies and guarantee winning with certainty. The example is based on the Mermin-Peres magic square game \cite{Mermin90,Peres90,Aravind03,Brassard05}, also known as the magic square game, which is a cooperative game with two players and a referee. In the magic square game, the central object is a $3 \times 3$ grid whose cells are to be filled with binary values. The referee selects one of the nine cells uniformly at random and reveals only partial information to the players: Player~1 is told the row index $r$, and Player~2 is told the column index $c$. Each player must then output three binary values, Player~1 for the entire row $r$, and Player~2 for the entire column $c$. The players win the round if the following conditions are simultaneously satisfied:
\begin{enumerate}
    \item their assignments agree on the overlapping cell $(r,c)$;
    \item row $r$ has even parity; and
    \item column $c$ has odd parity.
\end{enumerate}
Since no communication is allowed after the game begins, the players can coordinate only by agreeing on a joint strategy beforehand. Their goal is to maximize the probability of satisfying all three constraints for every possible choice of $(r,c)$.

Under any classical strategy, perfect success is impossible because the even-parity constraints on the rows are incompatible with the odd-parity constraints imposed on the columns. A representative deterministic assignment achieving the best possible performance is shown in Table~\ref{table:Magic_Deterministic}. For each pair $(r,c)$ selected by the referee, the players output the corresponding row and column from this table, which guarantees consistency on all overlapping entries except when $(r,c) = (3,3)$. Since this is the only losing case, the maximum achievable winning probability is $\tfrac{8}{9}$ when $(r,c)$ is drawn uniformly at random. As shown in~\cite{Brassard05}, this value remains optimal even when the players are allowed to use randomized classical strategies.

\setlength{\arrayrulewidth}{1.5pt}

\begin{table}[!t]
\centering
\large
\caption{The Magic Square Game: Deterministic Strategy}
\label{table:Magic_Deterministic}
\begin{tabular}{|c|c|c|}
\hline
{\small \,\,1\,\,} & {\small \,\,0\,\,} & {\small \,\,1\,\,} \\
\hline
{\small \,\,0\,\,} & {\small \,\,0\,\,} & {\small \,\,0\,\,} \\
\hline
{\small \,\,0\,\,} & {\small \,\,1\,\,} & {\small \,\,?\,\,} \\
\hline
\end{tabular}
\end{table}

\setcounter{equation}{5}

If the players share an entangled resource, they can win the magic square game with probability~1. Before the game begins, they prepare the state
\begin{align}
    \ket{E'_1E''_1E'_2E''_2}
    &= \frac{1}{2}\!\left(
        \ket{00}\!\otimes\!\ket{11}
        + \ket{11}\!\otimes\!\ket{00}\right.\nonumber\\
        &\left.\qquad- \ket{01}\!\otimes\!\ket{10}
        - \ket{10}\!\otimes\!\ket{01}\right),\label{eq:Bipartite_State}
\end{align}
and systems $E'_i, E''_i$ are given to the player~$i$.

Upon receiving the row index $r$ and column index $c$, the players act according to the quantum strategy summarized in Table~\ref{table:Magic_Quantum}. In this table, each cell is associated with a specific observable, where $X$, $Y$, and $Z$ denote the Pauli operators. Player~1 performs simultaneous measurements of the three observables appearing in row~$r$ on the joint system $(E'_1, E''_1)$, while Player~2 measures the three observables in column~$c$ on $(E'_2, E''_2)$. Each measurement outcome is then placed in the corresponding cell of the row or column they report.

This entangled strategy ensures that the players always agree on the overlapping cell, achieving a perfect winning probability of~1~\cite{Brassard05}.
\begin{table}[!t]
\centering
\renewcommand{\arraystretch}{1.5}
\caption{The Magic Square Game: Quantum Strategy}
\label{table:Magic_Quantum}
\begin{tabular}{|c|c|c|}
\hline
{\small $X\otimes\mathds{I}$} & {\small $\mathds{I}\otimes X$} & {\small $X\otimes X$} \\
\hline
{\small $\mathds{I}\otimes Y$} & {\small $Y\otimes\mathds{I}$} & {\small $Y\otimes Y$} \\
\hline
{\small $-X\otimes Y$} & {\small $-Y\otimes X$} & {\small $Z\otimes Z$} \\
\hline
\end{tabular}
\end{table}
\subsubsection{Channel Model}
We define a \ac{MAWTC} based on the magic square game, where the legitimate receiver observes a noiseless channel whenever the transmitted inputs correspond to a winning instance of the game. For all other input combinations, the channel outputs are generated uniformly at random, independent of the channel inputs. While the channel for the eavesdropper is the opposite, in the sense that the eavesdropper's channel outputs are generated uniformly at random, independent of the channel inputs, whenever the transmitted inputs correspond to a winning instance of the game. For all other input combinations, the eavesdropper observes a noiseless channel. In this way, the cooperation between the transmitters that enables them to win the game also contributes to securing the communication and increasing the achievable rates. The precise definition is provided below.

In the general formulation of the magic square game, a referee selects two questions $q_1 \in \calQ_1$ and $q_2 \in \calQ_2$, uniformly at random, and sends $q_1$ to Player~1 and $q_2$ to Player~2. Player~1 and Player~2 then produce answers $a_1 \in \calA_1$ and $a_2 \in \calA_2$, respectively, according to (possibly randomized) strategies $f_i : \calQ_i \to \calA_i$ for $i \in \{1,2\}$. The game is won if $(q_1, q_2, a_1, a_2) \in \mathfrak{W}$, where $\mathfrak{W}$ denotes the winning set.

We now define a \ac{MAWTC} $W_{YZ|X_1X_2}$ associated with this game by setting
\begin{align*}
    \calX_i &= \calQ_i \times \calA_i, \qquad i \in \br{1,2},\\
    \calY = \calZ &= \calQ_1 \times \calQ_2,
\end{align*}
such that,
\begin{subequations}\label{eq:Game_Input_Output_Relations}
\begin{align}
 &W_{Y|X_1X_2}\pr{\pr{\hat{q}_1,\hat{q}_2}|q_1,a_1,q_2,a_2}\nonumber\\
 &\quad=\begin{cases}
\indic{1}_{\br{\hat{q}_1=q_1}\cap\br{\hat{q}_2=q_2}},\,\,&\text{if}\Squad \pr{q_1,q_2,a_1,a_2}\in\mathfrak{W}\\
\frac{1}{\abs{\calQ_1}\abs{\calQ_2}},&\text{otherwise}
\end{cases},\label{eq:Game_Input_Output_Y}\\
&W_{Z|X_1X_2}\pr{\pr{\check{q}_1,\check{q}_2}|q_1,a_1,q_2,a_2}\nonumber\\
 &\quad=\begin{cases}
\frac{1}{\abs{\calQ_1}\abs{\calQ_2}},&\text{if}\Squad \pr{q_1,q_2,a_1,a_2}\in\mathfrak{W}\\
\indic{1}_{\br{\check{q}_1=q_1}\cap\br{\check{q}_2=q_2}},\,\,&\text{otherwise}
\end{cases}.\label{eq:Game_Input_Output_Z}
\end{align}
\end{subequations}
In other words, if the inputs $X_1 = (q_1, a_1)$ and $X_2 = (q_2, a_2)$ correspond to a winning instance of the game, 
\begin{itemize}
    \item the legitimate receiver receives the question pair exactly, that is, $Y = (q_1, q_2)$ with probability~1;
    \item the eavesdropper's output, $Z$, is drawn uniformly at random, independent of the channel inputs.
\end{itemize}Conversely, if the game is lost:
\begin{itemize}
    \item the output $Y$ is drawn uniformly at random from the question set, independent of the channel inputs;
    \item the eavesdropper receives the question pair in a noiseless manner, that is, $Z = (q_1, q_2)$ with probability~1.
\end{itemize}

Formally, the magic square game is defined by the question sets 
$\calQ_1 \triangleq \{1,2,3\}$ and $\calQ_2 \triangleq \{1,2,3\}$, 
and answer alphabets $\calA_1 = \calA_2 = \{0,1\}^3$. 
The winning set $\mathfrak{W}$ consists of all tuples 
\begin{align*}
\left(q_1, q_2, (a_1(j), a_2(j))_{j \in \{1,2,3\}}\right) 
    \in \calQ_1 \times \calQ_2 \times \calA_1 \times \calA_2,
\end{align*}
that satisfy the following conditions:
\begin{align*}
&a_1(q_2) = a_2(q_1), \nonumber\\
&a_1(1) + a_1(2) + a_1(3) \equiv 0 \pmod{2},\nonumber\\
&a_2(1) + a_2(2) + a_2(3) \equiv 1 \pmod{2}.
\end{align*}

\begin{subequations}
The secrecy capacity region of the \ac{MAWTC} $W_{YZ|X_1X_2}$, with marginals $W_{Y|X_1X_2}$ and $W_{Z|X_1X_2}$, when the transmitters share classical \ac{CR}, i.e., $\calC_{\textup{CR-MAC-WTC}}\!\pr{W_{YZ|X_1X_2}}$, is contained in the standard CR-assisted capacity region of the corresponding \ac{MAC} $W_{Y|X_1X_2}$, denoted by $\calC_{\textup{CR-MAC}}\!\pr{W_{Y|X_1X_2}}$:
\begin{align}
    \calC_{\textup{CR-MAC-WTC}}\!\pr{W_{YZ|X_1X_2}}
    \subseteq
    \calC_{\textup{CR-MAC}}\!\pr{W_{Y|X_1X_2}}.\label{eq:CR_MAC_VS_CR-MAWTC}
\end{align}
This inclusion holds because the \ac{MAWTC} imposes an additional secrecy constraint.

It is also known that the capacity region of a \ac{MAC} without any \ac{CR} shared between the transmitters, denoted by $\calC_{\textup{MAC}}\!\left(W_{Y|X_1X_2}\right)$, is equal to the capacity region of the \ac{MAC} with \ac{CR}; see, e.g., \cite[Remark~2]{Uzi_Entangled_MAC}. In particular,
\begin{align}
    \calC_{\textup{CR-MAC}}\!\pr{W_{Y|X_1X_2}}=\calC_{\textup{MAC}}\!\pr{W_{Y|X_1X_2}}.\label{eq:MAC_VS_CR-MAC}
\end{align}
\end{subequations}

Seshadri~\emph{et al.}~\cite{Seshadri23} showed that for the \ac{MAC} without any classical or quantum resources shared between the transmitters, the sum rate of the \ac{MAC} $W_{Y|X_1X_2}\!\pr{(\hat{q}_1,\hat{q}_2)|q_1,a_1,q_2,a_2}$ described in~\eqref{eq:Game_Input_Output_Y} is bounded by  $R_1+R_2\le3.02$. Therefore, from \eqref{eq:CR_MAC_VS_CR-MAWTC} and \eqref{eq:MAC_VS_CR-MAC}, we have
\begin{align}
     \calC_{\textup{CR-MAC-WTC}}\!\pr{W_{YZ|X_1X_2}}&\subseteq\calC_{\textup{CR-MAC}}\!\pr{W_{Y|X_1X_2}}\nonumber\\
     &\subseteq\br{(R_1,R_2):R_1+R_2\le3.02}.\label{eq:Seshadri}
\end{align}This outer bound is illustrated in Fig.~\ref{fig:Example_Region} by the orange region.
\subsubsection{Main Result}We now present the secrecy capacity region for the \ac{MAWTC} $W_{YZ|X_1X_2}$ defined in~\eqref{eq:Game_Input_Output_Relations}. 
\begin{theorem}
\label{thm:Capacity_Example}
The secrecy capacity region of the \ac{MAWTC} $W_{YZ|X_1X_2}$ defined in \eqref{eq:Game_Input_Output_Relations} is given by
\begin{align}
    \calC_{\textup{ET-MAC-WTC}}\!\pr{W_{YZ|X_1X_2}}= %
\left\{ \begin{array}{rl}
  (R_1,R_2):\;
    R_1&\hspace{-2.5mm}<\log_2(3)\\
    R_2&\hspace{-2.5mm}<\log_2(3)
	\end{array}
\hspace{-1mm}\right\}.\label{eq:Capacity_Example}
\end{align}
\end{theorem}The capacity region $\calC_{\textup{ET-MAC-WTC}}\!\pr{W_{YZ|X_1X_2}}$ in \eqref{eq:Capacity_Example} is depicted in Fig.~\ref{fig:Example_Region}. As shown in the figure, the capacity region in Theorem~\ref{thm:Capacity_Example} violates the outer bound in \eqref{eq:Seshadri} for the \ac{MAWTC} with classical \ac{CR} between the transmitters. In particular, the green region corresponds to rate pairs that are achievable with entanglement but not with classical \ac{CR}. Note that the outer bound in \eqref{eq:Seshadri} is originally derived for \acp{MAC} without security constraints. Since no classical strategy can win the game with probability~1, some information necessarily leaks from the transmitters to the eavesdropper. Consequently, the outer bound in \eqref{eq:Seshadri} is generally loose for the \ac{MAWTC} defined in \eqref{eq:Game_Input_Output_Relations} when classical \ac{CR} is shared between the transmitters.
\begin{figure}[t!]
\centering
\includegraphics[width=7.0cm]{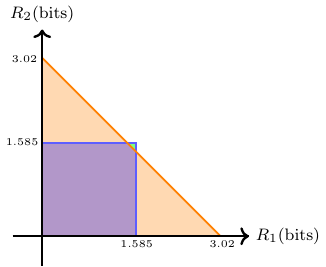}
\caption{Bounds on the secure rates for the magic square \ac{MAWTC} $W_{YZ|X_1X_2}$ defined in \eqref{eq:Game_Input_Output_Relations}. The orange region depicts the outer bound on the capacity region $\calC_{\textup{CR-MAC-WTC}}$ with classical \ac{CR} shared between the transmitters. The blue square represents the capacity region $\calC_{\textup{ET-MAC-WTC}}$ when the transmitters share entanglement. The green region highlights the rate pairs that are achievable with entanglement, but not achievable with classical~\ac{CR}.
}
\label{fig:Example_Region}
\vspace{-0.15in}
\end{figure}
\begin{proof}
Note that, similar to \eqref{eq:CR_MAC_VS_CR-MAWTC}, the secrecy capacity region of the entanglement-assisted \ac{MAWTC}, $\calC_{\textup{ET-MAC-WTC}}\!\pr{W_{YZ|X_1X_2}}$, is contained in the entanglement-assisted \ac{MAC} capacity region $\calC_{\textup{ET-MAC}}\!\pr{W_{Y|X_1X_2}}$:
\begin{align}
    \calC_{\textup{ET-MAC-WTC}}\!\pr{W_{YZ|X_1X_2}}
    \subseteq
    \calC_{\textup{ET-MAC}}\!\pr{W_{Y|X_1X_2}},\label{eq:ET_MAC_VS_ET-MAWTC}
\end{align}
again due to the secrecy requirement. 

The converse proof of Theorem~\ref{thm:Capacity_Example} is immediate, since we have~\eqref{eq:ET_MAC_VS_ET-MAWTC}, and the \ac{RHS} of~\eqref{eq:Capacity_Example} corresponds to the capacity of the noiseless \ac{MAC}  $W_{Y|X_1X_2}\!\pr{(\hat{q}_1,\hat{q}_2)|q_1,a_1,q_2,a_2}=\indic{1}_{\br{\hat{q}_1=q_1}\cap\br{\hat{q}_2=q_2}}$, without any secrecy constraint and regardless of whether entanglement resources are available.

To prove the direct part, consider the region $\calB_1$ in Theorem~\ref{thm:Achievable}.  We choose a bipartite state, the \acp{POVM}, and the associated classical \acp{RV} as follows.  Let the bipartite state $\ket{E_1E_2}$ be as defined in~\eqref{eq:Bipartite_State}, where  $E_i \triangleq E'_iE''_i$ for $i\in\br{1,2}$.  We set $\calU_0=\emptyset$ and choose $p_{U_i}$ to be uniform over  $\calU_i=\br{1,2,3}$ for $i\in\br{1,2}$.  Thus, the \acp{RV} $U_1$ and $U_2$ are distributed according to the referee's questions in the pseudo-telepathy game.

Given $U_i$, the Transmitter~$i$ performs the measurement $\calL_i(U_i)$ as follows. Transmitter~1 measures the observables in row $r=U_1$ of Table~\ref{table:Magic_Quantum}, obtains a random triplet  $V_1 \triangleq (a_1(j))_{j\in\br{1,2,3}}$, and transmits $X_1=(U_1,V_1)$ over the channel. Similarly, Transmitter~2 measures the observables in column $c=U_2$, obtains  $V_2 \triangleq (a_2(j))_{j\in\br{1,2,3}}$, and transmits $X_2=(U_2,V_2)$.

Since $(U_1,U_2,V_1,V_2)$ wins the game with certainty under the chosen entangled strategy, we have  $Y=(U_1,U_2)$ and $Z \indep (U_1,U_2,V_1,V_2)$ with probability~1. Therefore,
\begin{align*}
        \bbI(U_1;Y | U_2) &= \bbH(U_1) = \log_2(3)=1.585,\\
        \bbI(U_2;Y | U_1) &= \bbH(U_2) = \log_2(3)=1.585,\\
        \bbI(U_1,U_2;Y) &= \bbH(U_1,U_2) = 2\log_2(3)=2\times1.585,\\
        \bbI(U_1;Z | U_2) &= 
        \bbI(U_2;Z | U_1) =
        \bbI(U_1,U_2;Z) = 0.
\end{align*}
This construction requires an entanglement rate of $\theta_E = 2$ qubit pairs per channel use.
\end{proof}
\begin{remark}[Interpretations]
Intuitively, in the \ac{MAWTC} defined in \eqref{eq:Game_Input_Output_Relations}, when the game is won, the queries are effectively delivered only to the legitimate receiver; when the game is lost, they are delivered only to the eavesdropper. Since the magic square game admits a perfect quantum strategy, the resulting communication channel from the transmitters to the legitimate receiver becomes noiseless, while the eavesdropper's output is independent of the inputs and therefore uninformative. Consequently, the secrecy capacity coincides with the capacity of the corresponding entanglement-assisted \ac{MAC} without secrecy constraints. 
In contrast, because no classical strategy can win the game with probability~1, the eavesdropper's output cannot be made independent of the users' inputs in the classical setting. 
This leads to a secrecy penalty, and the outer bound in \eqref{eq:Seshadri} shows that the classical \ac{CR}-assisted secrecy region cannot contain the full entanglement-assisted region.
\end{remark}
\section{Output Statistics of \texorpdfstring{\ac{MAC}}{MAC} with Entangled Transmitters}
\label{sec:Output_Statistics_MAC}
\begin{figure}[t!]
\centering
\includegraphics[width=8.0cm]{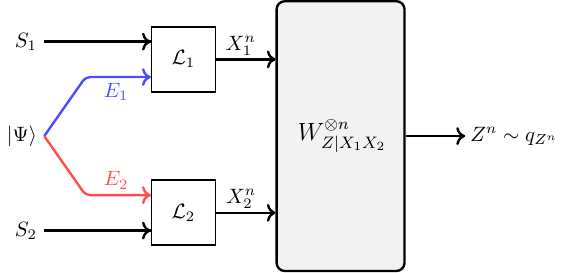}
\caption{Distribution approximation over a \ac{MAC} with entangled transmitters
}
\label{fig:System_Model_Resolvability}
\vspace{-0.15in}
\end{figure}
In this section, we study the output statistics of a \ac{MAC} with entangled transmitters. 
Approximating a desired output distribution through the use of minimal randomness at the channel input has emerged as a fundamental tool in information theory, with applications in source coding, rate–distortion theory, coordination, and secrecy, and is therefore of independent interest. 
Here, we establish two strong soft-covering lemmas that serve as the main technical tools for ensuring semantic security. 
The first lemma is a strong soft-covering result for \acp{MAC} that approximates the \ac{iid} output distribution induced when both transmitters send \ac{iid} codewords. 
The second lemma is a strong soft-covering result for \acp{MAC} that approximates the output distribution when only one transmitter sends an \ac{iid} codeword, while the other transmitter transmits a codeword selected from a codebook; consequently, the target output distribution in this case is non-\ac{iid}.

Consider a two-transmitter discrete memoryless \ac{MAC} $\pr{\calX_1,\calX_2,W_{Z\lvert X_1X_2},\calZ}$ with entangled transmitters, as illustrated in Fig.~\ref{fig:System_Model_Resolvability}, where $\calX_i$, $i\in\br{1,2}$, denotes the input alphabet at Transmitter~$i$, $W_{Z\lvert X_1X_2}$ is the channel law, and $\calZ$ is the output alphabet. 
\begin{definition}[Resolvability Code Components for \ac{MAC} with Entangled Transmitters]
\label{defi:code_Resolvability}
A sequence of $(2^{nR_1},2^{nR_2},n)$ resolvability codes $\br{C_n}_{n\in\bbN}$ for the \ac{MAC} $W_{Z\lvert X_1X_2}$, with entangled transmitters, consists of the following:
\begin{itemize}
    \item local randomness sets $\calS_1\triangleq\left[\left\lfloor2^{nR_1}\right\rfloor\right]$ and $\calS_2\triangleq\left[\left\lfloor2^{nR_2}\right\rfloor\right]$;
    \item a target distribution $q_{Z^n}$, which may in general be non-\ac{iid};
    \item a bipartite state $\Psi_{E_1E_2}\in\calD(\calH_{E_1E_2})$, which is shared between the two transmitters;
    \item two collections of encoding \acp{POVM} 
    \begin{align*}
        \calL_1^{(s_1)}=\pr{L_{x_1^n}^{(s_1)}}_{x_1^n\in\calX_1^n},\quad\text{and}\quad\calL_2^{(s_2)}=\pr{L_{x_2^n}^{(s_2)}}_{x_2^n\in\calX_2^n},
    \end{align*}acting on $E_1$ and $E_2$, respectively, with one pair of \acp{POVM} specified for each pair $(s_1,s_2)\in\calS_1\times\calS_2$.
\end{itemize}
\end{definition}
Traditionally, a soft-covering (resolvability) code is defined as follows.
\begin{definition}[Resolvability Code]
\label{defi:Code_Resolvability}
    A rate pair $(R_1,R_2)$ is said to be achievable for a discrete memoryless \ac{MAC} $W_{Z\lvert X_1 X_2}$ with entangled transmitters if, for a given target distribution $q_{Z^n}$, there exists a sequence of $\pr{2^{nR_1},2^{nR_2},n}$ codes $\br{C_n}_{n\in\bbN}$ such that
    \begin{align}
        \lim_{n\to\infty}\bbV\pr{p_{Z^n}, q_{Z^n}}&=0,\label{eq:Resolvability}
    \end{align}where $p_{Z^n}$ is the distribution induced at the output of the channel by $\br{C_n}_{n\in\bbN}$. 
\end{definition}The resolvability region, denoted by $\calR_{\textup{ET-MAC}}\pr{q_{Z^n}}$, where $q_{Z^n}$ is the target distribution, is defined as the closure of all achievable rate pairs $(R_1,R_2)$.

Here, we adopt a stronger notion, commonly referred to as \emph{strong soft-covering}~\cite{Cuff15,Cuff16,ZivSemantic2016,Ziv_WTC_with_NonCausaul_CSI}. In this setting, with high probability (with respect to the random code construction), the distance between the induced distribution and the target distribution decays exponentially fast in the block length $n$. Moreover, the probability that a randomly generated code fails to achieve this exponential decay is itself doubly–exponentially small.

Thus, instead of the requirement in~\eqref{eq:Resolvability}, we impose the stronger condition
\begin{align*}
    \bbP\br{\bbV\pr{p_{Z^n}, q_{Z^n}}>\exp\pr{-\gamma_1n}}&\le\exp\pr{-\exp\pr{\gamma_2n}},
\end{align*}for some constants $\gamma_1,\gamma_2>0$.

Consider a joint distribution $p_{U_0}p_{U_1\lvert U_0}p_{U_2\lvert U_0}$, entangled states $\Psi_{E_1E_2}$, and \ac{POVM} measurements $\calL_1\otimes\calL_2$ such that 
\begin{align}
    q_{Z\lvert U_0}(z\lvert u_0)&=\sum_{u_1,u_2,x_1,x_2}p_{U_1\lvert U_0}(u_1\lvert u_0)p_{U_2\lvert U_0}(u_2\lvert u_0)\nonumber\\
    &\times\tra\br{\pr{L_1(x_1\lvert u_0,u_1)\otimes L_2(x_2\lvert u_0,u_2)}\Psi_{E_1E_2}}\nonumber\\
    &\times W_{Z\lvert X_1X_2}(z\lvert x_1,x_2).\label{eq:Target_Dist}
\end{align}
    Let $u_0^n$ be a random time-sharing sequence, generated \ac{iid} according to $\prod_{t=1}^np_{U_0}(u_{0,t})$. Also, for $i\in\{1,2\}$, let $C_{i,n}\triangleq\br{U_i^n(s_i)}_{s_i\in\calS_i}$, where $\calS_i\triangleq\brk{2^{nR_i}}$, be a random codebook generated \ac{iid} according to $\prod_{t=1}^np_{U_i\lvert U_0}(u_{i,t}\lvert u_{0,t})$. A realization of $C_{i,n}$ is denoted by  $\calC_{i,n}\triangleq\br{u_i^n(s_i)}_{s_i\in\calS_i}$. 
Now, let $C_n\triangleq\br{U_0^n,C_{1,n},C_{2,n}}$ and $\calC_n\triangleq\br{u_0^n,\calC_{1,n},\calC_{2,n}}$.

Letting $\mathfrak{C}_n$ denote the set of all possible realizations of $C_n$, the codebook construction described above induces a probability measure $\mu$ over this ensemble. For every codebook $\calC_n \in \mathfrak{C}_n$, we have
\begin{align*}
    \mu\pr{\calC_n}&=p_{U_0}^\on(u_0^n)\prod_{s_1\in\calS_1}p_{U_1|U_0}^\on(u_1^n(s_1)|u_0^n)\nonumber\\
    &\quad\times\prod_{s_2\in\calS_2}p_{U_2|U_0}^\on(u_2^n(s_2)|u_0^n).
\end{align*}

Given the random codebooks and the local randomness $s_i$, the Transmitter~$i$, for $i\in\{1,2\}$, performs a measurement $\bigotimes_{t=1}^n\pr{\calL_i(u_{0,t},u_{i,t}(s_i))}$ on the entangled state $E_i^n$ and transmits the outcome of the measurement over the channel. Therefore,
\begin{align}
    &f\pr{x_1^n,x_2^n\lvert s_1,s_2}\nonumber\\
    &=\tra\Big[\Big(L_1^n\big(x_1^n\lvert u_0^n,u_1^n\pr{s_1}\big)\otimes L_2^n\big(x_2^n\lvert u_0^n,u_2^n\pr{s_2}\big)\Big)\Psi_{E_1E_2}^\on\Big]\nonumber\\
    &=\prod_{t=1}^n\tra\!\Big[\Big(L_1\big(x_{1,t}\lvert u_{0,t},u_{1,t}\pr{s_1}\big)\otimes L_2\big(x_{2,t}\lvert u_{0,t},u_{2,t}\pr{s_2}\big)\!\Big)\nonumber\\
    &\qquad\times\Psi_{E_1E_2}\Big],\label{eq:Encoding_Res}
\end{align}where $L_i^n\big(x_i^n\lvert u_0^n,u_i^n\pr{s_i}\big)\triangleq\bigotimes_{t=1}^nL_i\big(x_{i,t}\lvert u_{0,t},u_{i,t}\pr{s_i}\big)$, for $i\in\{1,2\}$. For a fixed codebook $\calC_n$, this encoding scheme induces the following joint distribution,
\begin{align}
    &p_{S_1S_2U_0^nU_1^nU_2^nX_1^nX_2^nZ^n\lvert\calC_n}\pr{s_1,s_2,\tilde{u}_0^n,\tilde{u}_1^n,\tilde{u}_2^n,x_1^n,x_2^n,y^n,z^n}\nonumber\\
    &\,\,\triangleq\frac{1}{2^{n(R_1+R_2)}}\indic{1}_{\br{\tilde{u}_0^n=u_0^n}\cap\br{\tilde{u}_1^n=u_1^n(s_1)}\cap\br{\tilde{u}_2^n=u_2^n(s_2)}}\nonumber\\
    &\qquad\times f\pr{x_1^n,x_2^n\lvert s_1,s_2}W_{Z\lvert X_1X_2}^\on\pr{z^n\lvert x_1^n,x_2^n}.\label{eq:Joint_Dist_Res}
\end{align}Therefore, we have the following marginal distribution:
\begin{align*}
    &p_{Z^n\lvert\calC_n}(z^n)=\frac{1}{2^{n(R_1+R_2)}}\nonumber\\
    &\quad\times\sum_{(s_1,s_2)\in\calS_1\times\calS_2}W_{Z\lvert U_0U_1U_2}^\on\pr{z^n\lvert u_0^n,u_1^n(s_1),u_2^n(s_2)},
\end{align*}where $W_{Z\lvert U_0U_1U_2}^\on$ is the marginal of \eqref{eq:Joint_Dist_Res}.

\begin{lemma}
\label{lemma:Resolvability_1}
For a discrete memoryless \ac{MAC} $W_{Z|X_1,X_2}$ with entangled transmitters and a target distribution $q_{Z|U_0}^\on(\cdot|u_0^n)$, where $q_{Z|U_0}(\cdot|u_0)$ is defined in \eqref{eq:Target_Dist}, if $(R_1,R_2)$ belongs to
    \begin{align}
        \bigcup\limits_{\substack{p_{U_0}p_{U_1\lvert U_0}p_{U_2\lvert U_0}\\ ,\calL_1\otimes\calL_2,\,\Psi_{E_1E_2}}}\left\{ \begin{array}{l}
 (R_1,R_2)\in\bbR_+^2: \\ 
 R_1 > \bbI(U_1;Z\lvert U_0) \\ 
 R_2 > \bbI(U_2;Z\lvert U_0) \\ 
 R_1 + R_2 > \bbI(U_1,U_2;Z\lvert U_0) \\ 
 \end{array} \right\},\label{eq:Resolvability_11}
\end{align}
then for sufficiently large $n$ there exists $\gamma_1,\gamma_2>0$ such that
\begin{align*}
    \bbP_\mu\pr{\bbV\pr{p_{Z^n|\bbC_n},q_{Z|U_0}^\on(\cdot|u_0^n)}>\exp\pr{-\gamma_1n}}\nonumber\\
    \le\exp\pr{-\exp(\gamma_2n)}.
\end{align*}
\end{lemma}
To prove Lemma \ref{lemma:Resolvability_1}, we revisit the proof techniques of \cite{Wiese2012,Frey18}, which differ slightly from those in \cite{Cuff15,Cuff16}, and extend them to \acp{MAC} with entangled transmitters. As detailed above, the proof is based on coded time sharing, random coding, and distribution approximation methods. In our scheme, we first generate a time-sharing sequence represented by the auxiliary \ac{RV} $U_0$. Conditioned on $U_0$, Transmitter~$i$, for $i\in\{1,2\}$, superimposes a codebook, represented by $U_i$, to encode the local randomness $S_i$. The channel inputs $X_1$ and $X_2$ are then generated via encoding \acp{POVM} that depend on the classical auxiliary \acp{RV} $(U_0,U_1)$ and $(U_0,U_2)$, respectively. 
Because these \acp{POVM} act on entangled states shared between the transmitters, they can induce correlations between $X_1$ and $X_2$. The full proof of Lemma \ref{lemma:Resolvability_1} is given in Appendix~\ref{proof:thm:Resolvability_1}. 
\begin{remark}[Strong Soft-Covering Lemma for \texorpdfstring{\acp{MAC}}{MACs} without Entanglement]
When the encoders do not share entanglement prior to communication, i.e., when $\Psi_{E_1E_2}=\Psi_{E_1}\otimes\Psi_{E_2}$, Lemma~\ref{lemma:Resolvability_1} reduces to a soft-covering result for \acp{MAC} without entanglement, recovering the resolvability region in \cite[Theorem~1]{Frey18}. In this case, the auxiliary systems $\Psi_{E_1}$ and $\Psi_{E_2}$ carry no correlations, and the corresponding \ac{POVM} $\calL_1 \otimes \calL_2$ reduces to trivial local measurements that can be absorbed into the encoding operations.
\end{remark}
Now consider an alternative scenario, in which the first transmitter does not select its codeword from the codebook $C_{1,n}$. Instead, conditioned on the common sequence $u_0^n$, it generates a length-$n$ sequence $U_1^n$ i.i.d.\ according to $p_{U_1|U_0}$. It then performs the product measurement $\bigotimes_{t=1}^n\pr{\calL_1(u_{0,t},u_{1,t})}$ on its share of the entangled state $E_1^n$, and the resulting measurement outcome is transmitted through the channel. The second transmitter operates similarly as before: given the local randomness $s_2$, it selects a codeword $U_2^n(s_2)$ from the predetermined codebook $C_{2,n}$. It then performs the measurement $\bigotimes_{t=1}^n\pr{\calL_2(u_{0,t},u_{2,t})}$ on its share of the entangled state $E_2^n$, and transmits the resulting measurement outcome over the channel. Let the distribution induced at the output of the channel by this scenario be denoted by $p_{Z^n\lvert C_{2,n}}$. For a fixed codebook $\calC_{2,n}$, this encoding scheme induces the following joint distribution,
\begin{align}
    &p_{S_2U_0^nU_1^nU_2^nX_1^nX_2^nZ^n\lvert\calC_{2,n}}\pr{s_2,\tilde{u}_0^n,\tilde{u}_1^n,\tilde{u}_2^n,x_1^n,x_2^n,y^n,z^n}\nonumber\\
    &\,\,\triangleq\frac{1}{2^{nR_2}}\indic{1}_{\br{\tilde{u}_0^n=u_0^n}\cap\br{\tilde{u}_2^n=u_2^n(s_2)}}p_{U_1|U_0}^\on(u_1^n|u_0^n) f\pr{x_1^n,x_2^n\lvert s_2}\nonumber\\
    &\qquad\times W_{Z\lvert X_1X_2}^\on\pr{z^n\lvert x_1^n,x_2^n},\label{eq:Joint_Dist_Res_One_Sided}
\end{align}where $f(x_1^n,x_2^n | s_2)$ is defined analogously to \eqref{eq:Encoding_Res}, except that the sequence $u_1^n$ is used in place of the codeword $u_1^n(s_1)$. Therefore, we have the following marginal distribution:
\begin{align*}
    p_{Z^n\lvert\calC_{2,n}}(z^n)=\frac{1}{2^{nR_2}}\sum_{s_2\in\calS_2}W_{Z\lvert U_0U_2}^\on\pr{z^n\lvert u_0^n,u_2^n(s_2)},
\end{align*}where $W_{Z\lvert U_0U_2}^\on$ is the marginal of \eqref{eq:Joint_Dist_Res_One_Sided}.

\begin{lemma}\label{lemma:one_Sided_MAC_Res}
For a discrete memoryless \ac{MAC} $W_{Z|X_1,X_2}$ with entangled transmitters and a target distribution $ p_{Z^n\lvert C_{2,n}}$, if $(R_1,R_2)$ belongs to
\begin{equation*}
    \bigcup\limits_{p_{U_0}p_{U_1\lvert U_0}p_{U_2\lvert U_0},\,\calL_1\otimes\calL_2,\,\Psi_{E_1E_2}} {(\calR_1 \cup \calR_2)} 
\end{equation*}
where, 
\begin{align*}
&\calR_1=\left\{ \begin{array}{l}
 (R_1,R_2)\in\bbR_+^2: \\ 
 R_1 > \bbI(U_1;Z\lvert U_0) \\ 
 R_2 > \bbI(U_2;Z\lvert U_0) \\ 
 R_1 + R_2 > \bbI(U_1,U_2;Z\lvert U_0) \\ 
 \end{array} \right\},\\
&\calR_2=\left\{\!\!\!\begin{array}{l}
 (R_1,R_2)\in\bbR_+^2:\,\,
 R_1 > \bbI(U_1;Z|U_0,U_2) \\ 
 \end{array}\!\!\!\right\},
\end{align*}
then for sufficiently large $n$ there exists $\gamma_1,\gamma_2>0$ such that
\begin{align}
    \bbP_\mu\pr{\bbV\pr{p_{Z^n|C_n},p_{Z^n|C_{2,n}}}>\exp\pr{-\gamma_1n}}\nonumber\\
    \le\exp\pr{-\exp(\gamma_2n)}.\label{eq:Doble_Expo_One_Sided}
\end{align}
\end{lemma}Similar to the proof of Lemma~\ref{lemma:Resolvability_1}, the proof is based on coded time sharing, random coding, and distribution approximation methods. Similarly to before, we first generate a time-sharing sequence represented by the auxiliary \ac{RV} $U_0$. Conditioned on $U_0$, Transmitter~$i$, for $i\in\{1,2\}$, superimposes a codebook, represented by $U_i$, to encode the local randomness $S_i$. The channel inputs $X_1$ and $X_2$ are then generated via encoding \acp{POVM} that depend on the classical auxiliary \acp{RV} $(U_0,U_1)$ and $(U_0,U_2)$, respectively. 
Since these \acp{POVM} act on entangled states shared between the transmitters, they can induce correlations between $X_1$ and $X_2$. A complete proof of Lemma~\ref{lemma:one_Sided_MAC_Res} is provided in Appendix~\ref{proof:lemma:one_Sided_MAC_Res}.
\begin{lemma}[Cardinality Bounds and Pure Entangled States]
\label{lemma:Pure_Shared_States_Resol}
    The union of the inner bound in Lemma~\ref{lemma:Resolvability_1} and Lemma~\ref{lemma:one_Sided_MAC_Res} can be achieved using auxiliary variables $U_0, U_1, U_2$ satisfying $\abs{\calU_0}\le 3$ and $\abs{\calU_i} \le 3\abs{\calX_1}\abs{\calX_2}$ for $i\in\{1,2\}$, together with pure shared states $\rho_{E_1E_2} \triangleq \den{\Psi_{E_1E_2}}{\Psi_{E_1E_2}}$. That is, restricting the union to such auxiliary variables and pure states does not reduce the achievable regions.
\end{lemma}The proof of Lemma~\ref{lemma:Pure_Shared_States_Resol} is similar to that of Lemma~\ref{lemma:Pure_Shared_States} and is omitted for brevity.
\begin{remark}[Strong Soft-Covering Lemma for \texorpdfstring{\acp{MAC}}{MACs} without Entanglement]
By setting $\Psi_{E_1E_2}=\Psi_{E_1}\otimes\Psi_{E_2}$, and taking $U_0=\emptyset$, 
Lemma~\ref{lemma:Resolvability_1} reduces to the soft-covering result for \acp{MAC} without entanglement given in \cite[Theorem~3]{YassaeeMAWC} and \cite[Lemma~1]{ISIT21}, but under a stronger doubly exponential guarantee.
\end{remark}
\begin{remark}[Strong Soft Covering with Rate-Limited Entanglement]
    Similar to Section~\ref{sec:Limited_Entanglement}, the achievable rate regions derived in this section can be extended to the setting in which the transmitters share entanglement resources of limited rate.
\end{remark}
\section{Conclusion}
\label{sec:Conclusion}
In this work, we investigated secure communication over a classical \ac{MAWTC} when the transmitters have access to shared entanglement. Our main contributions are the derivation of inner and outer bounds on the secrecy capacity region for a general \ac{MAWTC} with entangled transmitters, and the demonstration that shared entanglement can strictly enlarge the secrecy capacity compared to the case in which the transmitters share only classical \ac{CR}.

Furthermore, we studied the output statistics of a classical \ac{MAC} with entangled transmitters and established two strong soft-covering lemmas tailored to this setting. These results enabled us to prove semantic security for the \ac{MAWTC} and extend existing soft-covering arguments to scenarios involving entanglement.
\begin{appendices}
\begin{figure*}[b!]
    \hrulefill
\begin{align}
    &\bbP\pr{\pr{\hat{M}_1,\hat{S}_1,\hat{M}_2,\hat{S}_2}\ne\pr{M_1,S_1,M_2,S_2}\big\lvert(M_1,S_1,M_2,S_2)=(m_1,s_1,m_2,s_2)}\nonumber\\
    &=\bbP\pr{\calE_{m_1,s_1,m_2,s_2}^c\cup\bigcup_{\pr{\tilde{m}_1,\tilde{s}_1,\tilde{m}_2,\tilde{s}_2}\ne\pr{m_1,s_1,m_2,s_2}}\calE_{\tilde{m}_1,\tilde{s}_1,\tilde{m}_2,\tilde{s}_2}\Big\lvert\pr{m_1,s_1,m_2,s_2}}\nonumber\\
    &\le\bbP\pr{\calE_{m_1,s_1,m_2,s_2}^c\lvert\pr{m_1,s_1,m_2,s_2}}+\sum_{\substack{\pr{\tilde{m}_1,\tilde{s}_1}=\pr{1,1}\\ \pr{\tilde{m}_1,\tilde{s}_1}\ne \pr{m_1,s_1}}}^{2^{n(R_1+R'_1)}}\bbP\pr{\calE_{\tilde{m}_1,\tilde{s}_1,m_2,s_2}\Big\lvert(m_1,s_1,m_2,s_2)}\nonumber\\
    &\qquad+\sum_{\substack{\pr{\tilde{m}_2,\tilde{s}_2}=\pr{1,1}\\\pr{\tilde{m}_2,\tilde{s}_2}\ne \pr{m_2,s_2}}}^{2^{n(R_2+R'_2)}}\bbP\pr{\calE_{m_1,s_1,\tilde{m}_2,\tilde{s}_2}\Big\lvert(m_1,s_1,m_2,s_2)}+\sum_{\substack{\pr{\tilde{m}_1,\tilde{s}_1}=\pr{1,1}\\\pr{\tilde{m}_1,\tilde{s}_1}\ne \pr{m_1,s_1}}}^{2^{n(R_1+R'_1)}}\sum_{\substack{\pr{\tilde{m}_2,\tilde{s}_2}=\pr{1,1}\\\pr{\tilde{m}_2,\tilde{s}_2}\ne \pr{m_2,s_2}}}^{2^{n(R_2+R'_2)}}\bbP\pr{\calE_{\tilde{m}_1,\tilde{s}_1,\tilde{m}_2,\tilde{s}_2}\Big\lvert(m_1,s_1,m_2,s_2)},\label{eq:Union_Bound}
\end{align}
\end{figure*}
\section{Proof of Theorem~\ref{thm:Achievable}}
\label{proof:thm:Achievable}
We show that for each $\delta_1,\delta_2,\epsilon>0$, there exists a sequence of codes $\pr{2^{n(R_1-\delta_1)},2^{n(R_2-\delta_2)},n,\epsilon}$ codes $\br{\calC_n}_{n\in\bbN}$ for the \ac{MAWTC} $W_{YZ\lvert X_1X_2}$ that satisfies both the reliability and the semantic security constraints. 
Fix $p_{U_0},p_{U_1\lvert U_0},p_{U_2\lvert U_0}$, a bipartite state $\Psi_{E_1E_2}$, a set of \ac{POVM} operators $\calL_i(u_0,u_i)=\br{L_i(x_i\lvert u_0,u_i)}_{i\in\br{1,2}}$, and $\delta_1,\delta_2,\epsilon>0$. We assume that the transmitters share $n$ copies of the bipartite state, i.e., $\Psi_{E_1E_2}^\on$. 
\subsection{Random Codebook Generation}
\label{sec:Codebook_Cons} 
Let $u_0^n$ be a random time-sharing sequence, generated \ac{iid} according to $\prod_{t=1}^np_{U_0}(u_{0,t})$. Also, let $C_{i,n}\triangleq\br{U_i^n(m_i,s_i)}_{(m_i,s_i)\in\calM_i\times\calS_i}$, for $i\in\{1,2\}$, where $\calM_i\triangleq\brk{2^{nR_i}}$ and $\calS_i\triangleq\brk{2^{nR'_i}}$, be a random codebook generated \ac{iid} according to $\prod_{t=1}^np_{U_i\lvert U_0}(u_{i,t}\lvert u_{0,t})$. The indices $(m_i,s_i)$ can be seen as a one-layer binning, where $m_i$ denotes the bin index and $s_i$ denotes the codeword index in the bin number $m_i$. A realization of $C_{i,n}$ is denoted by  $\calC_{i,n}\triangleq\br{u_i^n(m_i,s_i)}_{(m_i,s_i)\in\calM_i\times\calS_i}$. 
Now, let $C_n\triangleq\br{U_0^n,C_{1,n},C_{2,n}}$ and $\calC_n\triangleq\br{u_0^n,\calC_{1,n},\calC_{2,n}}$.

Letting $\mathfrak{C}_n$ denote the set of all possible realizations of $C_n$, the codebook construction described above induces a probability measure $\mu$ over this ensemble. For every codebook $\calC_n \in \mathfrak{C}_n$, we have
\begin{align*}
    \mu\pr{\calC_n}=p_{U_0}^\on(u_0^n)\prod_{(m_1,s_1)\in\calM_1\times\calS_1}p_{U_1|U_0}^\on(u_1^n(m_1,s_1)|u_0^n)\nonumber\\
    \times\prod_{(m_2,s_2)\in\calM_2\times\calS_2}p_{U_2|U_0}^\on(u_2^n(m_2,s_2)|u_0^n).
\end{align*}

\subsection{Encoding}
Given the random codebooks, the message $m_i$, and the local randomness $s_i$, 
the Transmitter~$i$, for $i\in\{1,2\}$, first computes the sequences $\pr{u_0^n,u_i^n(m_i,s_i)}$. It then performs the measurement $\bigotimes_{t=1}^n\calL_i\pr{u_{0,t},u_{i,t}(m_i,s_i)}$ on its share of the entangled state $E_i^n$ and transmits the resulting measurement outcomes over the channel. Therefore,
\begin{align*}
    &f\pr{x_1^n,x_2^n\lvert m_1,s_1,m_2,s_2,\calC_n}\nonumber\\
    &=\tra\Big[\Big(L_1^n\big(x_1^n\lvert u_0^n,u_1^n\pr{m_1,s_1}\big)\otimes L_2^n\big(x_2^n\lvert u_0^n,u_2^n\pr{m_2,s_2}\big)\Big)\nonumber\\
    &\qquad\times\Psi_{E_1E_2}^\on\Big]\nonumber\\
    &=\prod_{t=1}^n\tra\Big[\Big(L_1\big(x_{1,t}\lvert u_{0,t},u_{1,t}\pr{m_1,s_1}\big)\nonumber\\
    &\qquad\otimes L_2\big(x_{2,t}\lvert u_{0,t},u_{2,t}\pr{m_2,s_2}\big)\Big)\Psi_{E_1E_2}\Big],
\end{align*}where $L_i^n\big(x_i^n\lvert u_0^n,u_i^n\pr{m_i,s_i}\big)\triangleq\bigotimes_{t=1}^nL_i\big(x_{i,t}\lvert u_{0,t},u_{i,t}\pr{m_i,s_i}\big)$, for $i\in\{1,2\}$. For a fixed codebook $\calC_n$, this encoding scheme induces the following joint distribution,
\begin{align*}
    &p_{M_1S_1M_2S_2U_0^nU_1^nU_2^nX_1^nX_2^nY^nZ^n\lvert\calC_n}\big(m_1,s_1,m_2,s_2,\tilde{u}_0^n,\tilde{u}_1^n,\tilde{u}_2^n\nonumber\\
    &,x_1^n,x_2^n,y^n,z^n\big)\triangleq\frac{1}{2^{n(R_1+R'_1+R_2+R'_2)}}\nonumber\\
    &\quad\times\indic{1}_{\br{\tilde{u}_0^n=u_0^n}\cap\br{\tilde{u}_1^n=u_1^n(m_1,s_1)}\cap\br{\tilde{u}_2^n=u_2^n(m_2,s_2)}}\nonumber\\
    &\quad\times f\pr{x_1^n,x_2^n\lvert m_1,s_1,m_2,s_2,\calC_n}W_{YZ\lvert X_1X_2}\pr{y^n,z^n\lvert x_1^n,x_2^n}.
\end{align*}

\subsection{Decoding and Error Probability Analysis}
The decoder declares that $(\hat{m}_1,\hat{m}_2)=(m_1,m_2)$ if there is a unique pair $(\hat{m}_1,\hat{m}_2)$ such that $\pr{U_0,U_1^n(\hat{m}_1,s_1),U_2^n(\hat{m}_2,s_2),Y^n}\in\calT_\epsilon^{(n)}$, for some $(s_1,s_2)$. 
Now we define the following error events
\begin{subequations}
\begin{align*}
    &\calE_{m_1,s_1,m_2,s_2}\triangleq\nonumber\\
    &\br{\pr{U_0^n,U_1^n(\hat{m}_1,s_1),U_2^n(\hat{m}_2,s_2),Y^n}\in\calT_\epsilon^{(n)}\pr{p_{U_0U_1U_2Y}}},
\end{align*}where $p_{U_0U_1U_2Y}$ is the marginal of the following distribution,
\begin{align*}
    &p_{U_0U_1U_2X_1X_2YZ}(u_0,u_1,u_2,x_1,x_2,y,z)\nonumber\\
    &=p_{U_0}(u_0)p_{U_1\lvert U_0}(u_1\lvert u_0)p_{U_2\lvert U_0}(u_2\lvert u_0)\nonumber\\
    &\quad\times\tra\br{\pr{L_1(x_1\lvert u_0,u_1)\otimes L_2(x_2\lvert u_0,u_2)}\Psi_{E_1E_2}}\nonumber\\
    &\quad\times W_{YZ\lvert X_1X_2}(y,z\lvert x_1,x_2).
\end{align*}
\end{subequations}Then the probability of error is
\begin{align*}
    &\bbE_\mu\mathsf{e}(C_n)\nonumber\\
    &=\bbE_\mu\bbP\pr{\pr{\hat{M}_1,\hat{M}_2}\ne\pr{M_1,M_2}}\nonumber\\
    &\le\bbE_\mu\bbP\pr{\pr{\hat{M}_1,\hat{S}_1,\hat{M}_2,\hat{S}_2}\ne\pr{M_1,S_1,M_2,S_2}}\nonumber\\
    &=\frac{1}{2^{n(R_1+R'_1+R_2+R'_2)}}\nonumber\\
    &\sum_{m_1=1}^{2^{nR_1}}\sum_{s_1=1}^{2^{nR'_1}}\sum_{m_2=1}^{2^{nR_2}}\sum_{s_2=1}^{2^{nR'_2}}\bbP\left(\pr{\hat{M}_1,\hat{S}_1,\hat{M}_2,\hat{S}_2}\ne\right.\nonumber\\
    &\quad\pr{M_1,S_1,M_2,S_2}\big\lvert(M_1,S_1,M_2,S_2)=(m_1,s_1,m_2,s_2)\Big),
\end{align*}where the inequality follows since the probability of error for decoding $(M_1,M_2)$ is upper bounded by the probability of error for decoding $(M_1,S_1,M_2,S_2)$. Now we have \eqref{eq:Union_Bound} at the bottom of the page, 
where the inequality follows from the Union Bound. 
From the law of large numbers, the first term on the \ac{RHS} of the \eqref{eq:Union_Bound} vanishes to zero when $n$ grows. To bound the second term on the \ac{RHS} of \eqref{eq:Union_Bound}, note that the codewords are generated independently, and that for any $\pr{\tilde{m}_1,\tilde{s}_1} \neq (m_1,s_1)$, the channel output $Y^n$ is independent of $\br{U_1(\tilde{m}_1,\tilde{s}_1)}_{(\tilde{m}_1,\tilde{s}_1)\in\calM_1\times\calS_1}$. Therefore, we have,
\setcounter{equation}{19}
\begin{subequations}\label{eq:Pe_m2_Sum}
\begin{align}
    &\bbP\pr{\calE_{\tilde{m}_1,\tilde{s}_1,m_2,s_2}\Big\lvert(m_1,s_1,m_2,s_2)}\nonumber\\
    &=\sum_{\pr{u_0^n,u_1^n,u_2^n,y^n}\in\calT_\epsilon^{(n)}(p_{U_0U_1U_2Y})}p_{U_0}^\on(u_0^n)p_{U_1\lvert U_0}^\on(u_1^n\lvert u_0^n)\nonumber\\
    &\qquad\times p_{U_2\lvert U_0}^\on(u_2^n\lvert u_0^n)W_{Y\lvert U_0U_2}^\on(y^n\lvert u_0^n,u_2^n)\nonumber\\
    &\mathop\le\limits^{(a)}\hspace{-3mm}\sum_{\pr{u_0^n,u_1^n,u_2^n,y^n}\in\calT_\epsilon^{(n)}(p_{U_0U_1U_2Y})}\hspace{-8mm}2^{-n(1-\epsilon)\bbH(U_0,U_1)}2^{-n(1-\epsilon)\bbH(U_2,Y\lvert U_0)}\nonumber\\
    &=\abs{\calT_\epsilon^{(n)}(p_{U_0U_1U_2Y})}2^{-n(1-\epsilon)\bbH(U_0,U_1)}\times2^{-n(1-\epsilon)\bbH(U_2,Y\lvert U_0)}\nonumber\\
    &\mathop\le\limits^{(b)}2^{n(1+\epsilon)\bbH(U_0,U_1,U_2,Y)}2^{-n(1-\epsilon)\bbH(U_0,U_1)}2^{-n(1-\epsilon)\bbH(U_2,Y\lvert U_0)}\nonumber\\
    &\mathop=\limits^{(c)}2^{-n\pr{\bbI\pr{U_1;Y\lvert U_0,U_2}-3\epsilon_n}},\label{eq:Pe_m1}
\end{align}where
\begin{itemize}
    \item[$(a)$] and $(b)$ follow from \cite[Theorem~1.1 and 1.2]{Kramer_Book};
    \item[$(c)$] follows since $U_1-U_0-U_2$ forms a Markov chain;
\end{itemize}
Similarly, one can show 
    \begin{align}
    \bbP\pr{\calE_{m_1,s_1,\tilde{m}_2,\tilde{s}_2}\Big\lvert(m_1,s_1,m_2,s_2)}&\le2^{-n\pr{\bbI\pr{U_2;Y\lvert U_0,U_1}-3\epsilon_n}},\\
    \bbP\pr{\calE_{\tilde{m}_1,\tilde{s}_1,\tilde{m}_2,\tilde{s}_2}\Big\lvert(m_1,s_1,m_2,s_2)}&\le2^{-n\pr{\bbI\pr{U_1,U_2;Y\lvert U_0}-3\epsilon_n}}.
\end{align}
\end{subequations}Substituting \eqref{eq:Pe_m2_Sum} into the \ac{RHS} of 
\eqref{eq:Union_Bound} shows that if,
\begin{subequations}\label{eq:Reliability_Constraints}
    \begin{align}
        R_1+R'_1&<\bbI(U_1;Y\lvert U_0,U_2),\\
        R_2+R'_2&<\bbI(U_2;Y\lvert U_0,U_1),\\
        R_1+R'_1+R_2+R'_2&<\bbI(U_1,U_2;Y\lvert U_0),
    \end{align}
\end{subequations}then the average error probability satisfies
\begin{align}
    \mathsf{e}(C_n)\xrightarrow[n\to\infty]{}0.\label{eq:Pe_to0}
\end{align}
\subsection{Security Analysis}First, consider \eqref{eq:Bounding_MIM1M2Z_1} at the bottom of the page,
\begin{figure*}[b!]
    \hrulefill
    \setcounter{equation}{22}
\begin{subequations}
\begin{align}
    \bbI\pr{M_1,M_2;Z^n\lvert\calC_n}&\le\bbI\pr{M_1,M_2;U_0^n,Z^n\lvert\calC_n}\nonumber\\
    &=\bbD\pr{p_{M_1,M_2,U_0^n,Z^n|\calC_n}\lVert p_{M_1,M_2|\calC_n}\,p_{U_0^n,Z^n|\calC_n}}\nonumber\\
    &\mathop=\limits^{(a)}\bbD\left(p_{M_1,M_2,U_0^n|\calC_n}\,p_{Z^n|M_1,M_2,U_0^n,\calC_n}\lVert p_{M_1,M_2U_0^n|\calC_n} p_{Z^n|U_0^n,\calC_n}\right)\label{eq:Step_a_Sem_Sec}\\
    &\mathop\le\limits^{(b)}\bbD\pr{p_{M_1,M_2,U_0^n|\calC_n}\,p_{Z^n|M_1,M_2,U_0^n,\calC_n}\lVert p_{M_1,M_2U_0^n|\calC_n}\,q_{Z|U_0}^\on}\nonumber\\
    &\mathop\le\limits^{(c)}\max_{(m_1,m_2)\in\calM_1\times\calM_2}\hspace{-1mm}\bbD\left(p_{U_0^n|\calC_n}\,p_{Z^n|M_1,M_2,U_0^n,\calC_n}\pr{\cdot|m_1,m_2,\cdot,\cdot}\lVert p_{U_0^n|\calC_n}\,q_{Z|U_0}^\on\right),\label{eq:Bounding_MIM1M2Z_1}
\end{align}
\end{subequations}
\end{figure*}
where $(a)$ and $(c)$ follow since $(M_1,M_2)$ is independent of $U_0^n$, and $(b)$ follows since
\begin{align*}
    &\bbD\pr{p_{M_1,M_2,U_0^n,Z^n|\calC_n}\lVert p_{M_1,M_2U_0^n|\calC_n}\,p_{Z^n|U_0^n,\calC_n}}\nonumber\\
    &=\bbD\pr{p_{M_1,M_2,U_0^n,Z^n|\calC_n}\lVert p_{M_1,M_2U_0^n|\calC_n}\,q_{Z|U_0}^\on}\nonumber\\
    &\qquad-\bbD\pr{p_{U_0^n,Z^n|\calC_n}\lVert p_{U_0^n|\calC_n}\,q_{Z|U_0}^\on},
\end{align*}and relative entropy is non-negative. Maximizing both sides of \eqref{eq:Bounding_MIM1M2Z_1} \ac{wrt} $p_{M_1,M_2}$ leads to
\begin{align}
    &\max_{p_{M_1,M_2}}\bbI\pr{M_1,M_2;Z^n\lvert\calC_n}\nonumber\\
    &\le\max_{(m_1,m_2)\in\calM_1\times\calM_2}\hspace{-1mm}\bbD\left(p_{U_0^n|\calC_n}\,p_{Z^n|M_1,M_2,U_0^n,\calC_n}\pr{\cdot|m_1,m_2,\cdot,\cdot}\right.\nonumber\\
    &\qquad\left.\lVert p_{U_0^n|\calC_n}\,q_{Z|U_0}^\on\right).\label{eq:Max_MI}
\end{align}Note that $p_{U_0^n|\calC_n}\,p_{Z^n|M_1,M_2,U_0^n,\calC_n}\!\pr{\cdot|m_1,m_2,\cdot,\cdot}\ll p_{U_0^n|\calC_n}\,q_{Z|U_0}^{\otimes n}$, where $\ll$ denotes the absolute continuity of the corresponding \acp{PMF}. Moreover, according to \cite[Eq.~(30)]{Cuff13}, for \acp{PMF} $p$ and $q$ with finite support such that $p\ll q$, the relative entropy $\bbD(p\|q)$ can be upper bounded in terms of their total variation distance $\bbV(p,q)$. Now applying \cite[Lemma~9]{Goldfeld18} leads to
\begin{align}
    &\max_{(m_1,m_2)\in\calM_1\times\calM_2}\bbD\left(p_{U_0^n|\calC_n}\,p_{Z^n|M_1,M_2,U_0^n,\calC_n}\pr{\cdot|m_1,m_2,\cdot,\cdot}\right.\nonumber\\
    &\qquad\left.\lVert p_{U_0^n|\calC_n}\,q_{Z|U_0}^\on\right)\nonumber\\
    &\le\max_{(m_1,m_2)\in\calM_1\times\calM_2}\delta\pr{m_1,m_2,\calC_n}\left(n\log\abs{\calZ}\right.\nonumber\\
    &\left.\qquad-\log\delta\pr{m_1,m_2,\calC_n}+n\log q_{Z|U_0}^{(\min)}\right),\label{eq:KLD_TV}
\end{align}where $q_{Z|U_0}^{(\min)}\triangleq\min\left\{q_{Z|U_0}(z|u_0),(z,u_0)\in\calZ\times\calU_0\right.$: $\left.q_{Z|U_0}(z|u_0)>0\right\}$ and
\begin{align*}
    &\delta\pr{m_1,m_2,\calC_n}\nonumber\\
    &\triangleq\bbV\pr{p_{U_0^n|\calC_n}\,p_{Z^n|M_1,M_2,U_0^n,\calC_n}\pr{\cdot|m_1,m_2,\cdot,\cdot}, p_{U_0^n|\calC_n}\,q_{Z|U_0}^\on}\nonumber\\
    &=\bbV\pr{p_{Z^n|M_1,M_2,U_0^n,\calC_n}\pr{\cdot|m_1,m_2,\cdot,\cdot}, q_{Z|U_0}^\on}.
\end{align*}
Therefore, from \eqref{eq:Max_MI} and \eqref{eq:KLD_TV}, to obtain
\begin{align}
    \max_{p_{M_1,M_2}}\bbI(M_1,M_2;Z^n)\le\exp\pr{-n\gamma_1},\quad\text{for some}\,\gamma_1>0,\label{eq:Semantic_Sec_to0}
\end{align}it suffices to show that there exist $\calC_n\in\bbC$ and $\gamma_1>0$ such that $\max\limits_{(m_1,m_2)\in\calM_1\times\calM_2}\delta\pr{m_1,m_2,\calC_n}\le\exp\pr{-n\gamma_1}$, for $n$ large enough. We have
\begin{align}
    &\bbP\left(\left\{\max\limits_{(m_1,m_2)\in\calM_1\times\calM_2}\delta\pr{m_1,m_2,\calC_n}\right.\right.\nonumber\\
    &\qquad\qquad\qquad\le\exp\pr{-n\gamma_1}\Big\}^c\Big)\nonumber\\
    &=\bbP\Big(\!\Big\{\forall (m_1,m_2)\in\calM_1\times\calM_2,\,\delta\pr{m_1,m_2,\calC_n}\nonumber\\
    &\qquad\qquad\qquad\le\exp\pr{-n\gamma_1}\Big\}^c\Big)\nonumber\\
    &=\bbP\Big(\exists (m_1,m_2)\in\calM_1\times\calM_2,\,\delta\pr{m_1,m_2,\calC_n}\nonumber\\
    &\qquad\qquad\qquad>\exp\pr{-n\gamma_1}\Big)\nonumber\\
    &=\bbP\pr{\bigcup\limits_{(m_1,m_2)\in\calM_1\times\calM_2}\br{\delta\pr{m_1,m_2,\calC_n}>\exp\pr{-n\gamma_1}}}\nonumber\\
    &\le\sum\limits_{(m_1,m_2)\in\calM_1\times\calM_2}\bbP\big(\delta\pr{m_1,m_2,\calC_n}>\exp\pr{-n\gamma_1}\big),\label{eq:Union_Bound_Sem_Sec}
\end{align}where the inequality follows from the union bound. 
Now from Lemma~\ref{lemma:Resolvability_1}, if
\begin{subequations}\label{eq:Sec_Cons_Two_Sided}
    \begin{align}
        R'_1&>\bbI(U_1;Z\lvert U_0),\label{eq:Sec_Cons_R1p}\\
        R'_2&>\bbI(U_2;Z\lvert U_0),\label{eq:Sec_Cons_R2p}\\
        R'_1+R'_2&>\bbI(U_1,U_2;Z\lvert U_0),\label{eq:Sec_Cons_R1pR2p}
    \end{align}
\end{subequations}
then, for any $m_1\in\calM_1$, $m_2\in\calM_2$, and sufficiently large $n$, there exist $\gamma_1, \gamma_2>0$ such that
\begin{align}
    \bbP\pr{\delta\pr{m_1,m_2,\calC_n}>\exp\pr{-\gamma_1n}}\le\exp\pr{-\exp(\gamma_2n)}.\label{eq:Prob_Delta_any_M}
\end{align}
\begin{figure*}[b!]
    \hrulefill
    \setcounter{equation}{33}
\begin{align}
    \bbI\pr{M_1,M_2;Z^n\lvert\calC_n}&\mathop=\limits^{(a)}\bbE_C\sbr{\bbD\pr{p_{Z^n\lvert M_1M_2C}\lVert p_{Z^n\lvert C}\lvert p_{M_1M_2}}}\nonumber\\
    &\mathop=\limits^{(a)}\bbD\pr{p_{M_1,M_2,U_0^n|\calC_n}\,p_{Z^n|M_1,M_2,U_0^n,\calC_n}\lVert p_{M_1,M_2U_0^n|\calC_n}\,p_{Z^n|U_0^n,\calC_n}}\nonumber\\
    &=\bbD\pr{p_{M_1,M_2,U_0^n|\calC_n}\,p_{Z^n|M_1,M_2,U_0^n,\calC_n}\lVert p_{M_1,M_2U_0^n|\calC_n}\,p_{Z^n|U_0^n,\calC_{2,n}}}-\bbD\pr{p_{Z^n\lvert U_0^n,\calC_n}\lVert p_{Z\lvert U_0^n,\calC_{2,n}}}\nonumber\\
    &\mathop\le\limits^{(b)}\bbD\pr{p_{M_1,M_2,U_0^n|\calC_n}\,p_{Z^n|M_1,M_2,U_0^n,\calC_n}\lVert p_{M_1,M_2,U_0^n|\calC_n}\,p_{Z^n|U_0^n,\calC_{2,n}}}\nonumber\\
    &\mathop\le\limits^{(c)}\max_{(m_1,m_2)\in\calM_1\times\calM_2}\bbD\pr{p_{U_0^n|\calC_n}\,p_{Z^n|M_1,M_2,U_0^n,\calC_n}\pr{\cdot|m_1,m_2,\cdot,\cdot}\lVert p_{U_0^n|\calC_n}\,p_{Z^n|U_0^n,\calC_{2,n}}},\label{eq:Bounding_MIM1M2Z_4}
\end{align}
    \setcounter{equation}{29}
\end{figure*}
Substituting \eqref{eq:Prob_Delta_any_M} into \eqref{eq:Union_Bound_Sem_Sec} leads to
\begin{align}
    &\bbP\pr{\br{\max\limits_{(m_1,m_2)\in\calM_1\times\calM_2}\delta\pr{m_1,m_2,\calC_n}\le\exp\pr{-n\gamma_1}}^c}\nonumber\\
    &\le\sum\limits_{(m_1,m_2)\in\calM_1\times\calM_2}\exp\pr{-\exp(\gamma_2n)}\nonumber\\
    &=2^{n(R_1+R_2)}\exp\pr{-\exp(\gamma_2n)}\nonumber\\
    &\triangleq\lambda_n\xrightarrow[n\to\infty]{}0,\nonumber
\end{align}which implies that 
\begin{align}
    &\bbP\pr{\max\limits_{(m_1,m_2)\in\calM_1\times\calM_2}\delta\pr{m_1,m_2,\calC_n}\le\exp\pr{-n\gamma_1}}\nonumber\\
    &\ge1-\lambda_n\xrightarrow[n\to\infty]{}1.\label{eq:P_Max_TV}
\end{align}Inequality~\eqref{eq:P_Max_TV}, together with~\eqref{eq:KLD_TV}, implies that if $R'_1$ and $R'_2$ satisfy~\eqref{eq:Sec_Cons_Two_Sided}, then the probability that a randomly generated sequence of codes satisfies the semantic security constraint is arbitrarily close to~$1$ for sufficiently large~$n$.

To establish the existence of a sequence of $\bigl(2^{nR_1},2^{nR_2},n\bigr)$ reliable and semantically secure codes, we employ the Selection Lemma \cite[Lemma~2.2]{BlochBarros}.

\begin{lemma}[Selection Lemma]
    \label{lemma:Selection}
    Let $\br{B_n}_{n\in\mathbb{N}}$, where $B_n\in\calB_n$, be a sequence of \acp{RV}.  
    Let $\br{f_1^{(n)}, f_2^{(n)},\ldots, f_M^{(n)}}_{n\in\bbN}$, with $M<\infty$, be a collection of $M$ sequences of bounded functions $f_i^{(n)}:\calB_n \to \bbR_+$, for $i\in[M]$.  
    If
    \begin{align*}
        \mathbb{E}\bigl[f_i^{(n)}(B_n)\bigr] \xrightarrow[n\to\infty]{} 0,\qquad \forall\, i\in[M],
    \end{align*}
    then there exists a sequence $\{b_n\}_{n\in\mathbb{N}}$, with $b_n\in\calB_n$ for every $n$, such that
    \begin{align*}
        f_i^{(n)}(b_n) \xrightarrow[n\to\infty]{} 0,\qquad \forall\, i\in[M].
    \end{align*}
\end{lemma}
Applying Lemma~\ref{lemma:Selection} to the \acp{RV} $\br{C_n}_{n\in\bbN}$ and the bounded functions $\mathsf{e}(C_n)$ and $\indic{1}_{\br{\max_{p_{M_1,M_2}}\bbI(M_1,M_2;Z^n) > \exp\pr{-n\gamma_1}}}$, and using \eqref{eq:Pe_to0} and \eqref{eq:Semantic_Sec_to0}, we conclude that there exists a sequence of codes $\br{C_n}_{n\in\bbN}$ such that
\begin{subequations}
    \begin{align}
        \mathsf{e}(C_n) &\xrightarrow[n\to\infty]{} 0, \label{eq:Pe_to0_Final}\\
        \indic{1}_{\br{\max_{p_{M_1,M_2}}\bbI(M_1,M_2;Z^n) > \exp\pr{-n\gamma_1}}} &\xrightarrow[n\to\infty]{} 0. \label{eq:Semantic_Sec_Final}
    \end{align}
\end{subequations}Since the indicator function in \eqref{eq:Semantic_Sec_Final} takes values only in $\br{0,1}$, its convergence to zero implies that there exists an $n_0\in\bbN$ such that
\begin{align*}
    \indic{1}_{\br{\max_{p_{M_1,M_2}}\bbI(M_1,M_2;Z^n) > \exp\pr{-n\gamma_1}}} = 0,
\quad \forall\, n > n_0.
\end{align*}
Consequently,
\begin{align}
\max_{p_{M_1,M_2}}\bbI(M_1,M_2;Z^n)
    \le \exp\pr{-n\gamma_1}, \qquad \forall\, n > n_0.\label{eq:Sem_sec_Final}
\end{align}Finally, applying the Fourier-Motzkin elimination procedure to \eqref{eq:Reliability_Constraints} and \eqref{eq:Sec_Cons_Two_Sided} to eliminate $R'_1$ and $R'_2$ leads to the rate region $\calB_1$ in Theorem~\ref{thm:Achievable}.

Now we show that the region $\calB_2$ in Theorem~\ref{thm:Achievable} is achievable. One such situation arises when we are interested in deriving an achievable rate region for secure communication over the \ac{MAWTC} with a single confidential message. Codebook generation, encoding, and decoding are similar to those we developed for deriving the achievable region $\calB_1$ in Theorem~\ref{thm:Achievable} with $R_2=0$. Assume that 
\begin{align}
    R'_1>\bbI(U_1;Z\lvert U_0,U_2).\label{eq:Sec_Cons_One_Sided}
\end{align}If the codebook for the second transmitter is large enough, i.e., $R'_2>\bbI(U_2;Z\lvert U_0)$, then the hypothesis of Lemma~\ref{lemma:Resolvability_1} is satisfied, and the security proof is similar to that in \eqref{eq:Sem_sec_Final}. Otherwise, if $R'_2\le\bbI(U_2;Z\lvert U_0)$, we use Lemma~\ref{lemma:one_Sided_MAC_Res} as in \eqref{eq:Bounding_MIM1M2Z_4} provided at the bottom of the page, 
where $(a)$ follows similarly to \eqref{eq:Step_a_Sem_Sec}, $(b)$ follows from the non-negativity of the relative entropy, and $(c)$ follows from the independence of $(M_1,M_2)$ and $U_0^n$. Now, following the same steps used to prove \eqref{eq:Pe_to0_Final} and \eqref{eq:Sem_sec_Final}, but applying Lemma~\ref{lemma:one_Sided_MAC_Res} in place of Lemma~\ref{lemma:Resolvability_1}, one can show that there exists a sequence of codes $\br{C_n}_{n\in\bbN}$ and constants $\gamma_1,n_0>0$ such that
\begin{align*}
&\mathsf{e}(C_n) \xrightarrow[n\to\infty]{} 0,\\
&\max_{p_{M_1,M_2}}\bbI(M_1,M_2;Z^n)
    \le \exp\pr{-n\gamma_1}, \qquad \forall\, n > n_0,
\end{align*}if \eqref{eq:Reliability_Constraints}, with $R_2=0$, and \eqref{eq:Sec_Cons_One_Sided} hold.

\section{Proof of Lemma~\ref{lemma:Pure_Shared_States}}
\label{proof:lemma:Pure_Shared_States}
In this section, we show that both the inner and outer regions can be achieved using pure shared states and auxiliary variables with bounded cardinalities.
\subsection{Inner Bound}
We begin by proving Lemma~\ref{lemma:Pure_Shared_States} and showing that the achievable region in Theorem~\ref{thm:Achievable} can be attained using pure shared states together with auxiliary variables of bounded cardinality.
\subsubsection{Purification}
Consider a given Hilbert space $\calH_{E_1}\otimes\calH_{E_2}$. We show that the union over the bipartite states $\Psi_{E_1E_2}$ is exhausted by pure states. Fix a distribution $\br{p_{U_1|U_0}(\cdot|u_0)\, p_{U_2|U_0}(\cdot|u_0)}_{u_0 \in \calU_0}$, a bipartite pure state $\Psi_{E_1E_2}$, and measurement families $\calL_1(u_0,u_1)$ and $\calL_2(u_0,u_2)$ acting on $E_1$ and $E_2$, respectively. We define
\begin{align*}
\mathfrak{R}\pr{\Psi,\calL_1,\calL_2}\triangleq\bigcup\limits_{p_{U_0}p_{U_1\lvert U_0}p_{U_2\lvert U_0},\Psi_{E_1E_2},\calL_1\otimes\calL_2}\br{\calB_1\cup\calB_2\cup\calB_3},
\end{align*}where $\calB_1,\calB_2$, and $\calB_3$ are defined in \eqref{eq:Regions}. Consider a spectral decomposition of $\Psi_{E_1E_2}$ of the form
\setcounter{equation}{34}
\begin{align}
    \Psi_{E_1E_2}
    = \sum_{v \in \calV} p_V(v)\,
        \den{\psi_v}{\psi_v}_{E_1E_2}, \label{eq:Spectral_Psi}
\end{align}
where $p_V(v)$ is a probability distribution on $\calV$, such that the pure states $\br{\ket{\psi_v}_{E_1E_2}}_{v \in \calV}$ form an orthonormal basis for 
$\calH_{E_1} \otimes \calH_{E_2}$. Therefore, $\Psi_{E_1E_2}$ has the following purification,
\begin{align*}
    \ket{\phi_{E_1E_2B_1B_2}}=\sum_{v\in\calV}\sqrt{p_V(v)}\ket{\psi_v}_{E_1E_2}\otimes\ket{v}_{B_1}\otimes\ket{v}_{B_2}.
\end{align*}
Now consider the rate region $\mathfrak{R}(\Psi, \calL'_1, \calL'_2)$ corresponding to the following choice of shared state, auxiliary distributions, and measurements.  
Retain the same family of conditional distributions 
$\br{p_{U_1|U_0}(\cdot|u_0)\, p_{U_2|U_0}(\cdot|u_0)}_{u_0 \in \calU_0}$.  
Let Transmitter~1 and Transmitter~2 share the state $\ket{\phi_{E_1E_2B_1B_2}}$.  
Suppose Transmitter~1 performs a measurement on the joint system $(E_1B_1)$ and Transmitter~2 performs a measurement on $(E_2B_2)$, using the \acp{POVM}
\begin{align*}
    L'_1(x_1, v | U_0 U_1)
        &= L_1(x_1 | U_0 U_1) \otimes \den{v}{v},\\[1mm]
    L'_2(x_2, v | U_0 U_2)
        &= L_2(x_2 | U_0 U_2) \otimes \den{v}{v},
\end{align*}
for all $(u_0, u_1, u_2, x_1, x_2, v) \in 
\calU_0 \times \calU_1 \times \calU_2 \times 
\calX_1 \times \calX_2 \times \calV$.
The corresponding input distribution $\tilde{p}_{X_1 X_2 | U_0 U_1 U_2}$ is given by
\begin{align*}
    &\tilde{p}_{X_1 X_2 | U_0 U_1 U_2}
    \nonumber\\
    &= \sum_{v \in \calV} p_V(v)\,
       \tra\!\left[\bigl(L_1(x_1 | u_0, u_1) \otimes L_2(x_2 | u_0, u_2)\bigr)\right.\nonumber\\
       &\qquad\left.\times
       \ket{\psi_v}\!\bra{\psi_v}_{E_1E_2}\right] \\
    &\mathop=\limits^{(a)} \tra\!\left[\bigl(L_1(x_1 | u_0, u_1) \otimes L_2(x_2 | u_0, u_2)\bigr)\right.\nonumber\\
    &\qquad\left.\times
       \pr{\sum_{v \in \calV} p_V(v)\, \ket{\psi_v}\!\bra{\psi_v}_{E_1E_2}}\right] \\
    &\mathop=\limits^{(b)}
       \tra\!\sbr{\bigl(L_1(x_1 | u_0, u_1) \otimes L_2(x_2 | u_0, u_2)\bigr)
       \Psi_{E_1E_2}} \\
    &\mathop=\limits^{(c)} p_{X_1 X_2 | U_0 U_1 U_2},
\end{align*}
where $(a)$ follows from the linearity of trace operation, $(b)$ follows from the spectral decomposition of $\Psi_{E_1E_2}$ in~\eqref{eq:Spectral_Psi}, and $(c)$ follows from the definition of the joint distribution in~\eqref{eq:Joint_Dist_Single_Letter_1}. Therefore, $\mathfrak{R}(\Psi, \calL_1, \calL_2)= \mathfrak{R}(\Psi, \calL'_1, \calL'_2)$, which implies that the entire region $\mathfrak{R}(\Psi, \calL_1, \calL_2)$
can be achieved using only pure shared states.
\subsubsection{Cardinality Bounds}
We consider a joint state $\Psi_{E_1E_2}$ and a joint \ac{PMF} $p_{U_0U_1U_2X_1X_2YZ}$ on $\calU_0\times\calU_1\times\calU_2\times\calX_1\times\calX_2\times\calY\times\calZ$ distributed according to \eqref{eq:Joint_Dist_Single_Letter_1}. By the Fenchel-Eggleston-Carath\'eodory theorem, we can preserve the value of $\bbI(U_1;Y\lvert U_0,U_2)-\bbI(U_1;Z\lvert U_0), \bbI(U_2;Y\lvert U_0,U_1)-\bbI(U_2;Z\lvert U_0)$, and $\bbI(U_1,U_2;Y\lvert U_0)-\bbI(U_1,U_2;Z\lvert U_0)$ by restricting the alphabet size $\card{\calU_0}\le3$. 

Note that since the region $\calB_1$ in Theorem~\ref{thm:Achievable} has more constraints, it is sufficient to bound the cardinalities of $(U_1, U_2)$ only for $\calB_1$. 
To bound the cardinalities of $(U_1, U_2)$, consider the rate region in \eqref{eq:Region_1}:
\begin{subequations}\label{eq:Cardinalities}
\begin{align}
    &\calB_1\triangleq \nonumber\\
&\left\{ \begin{array}{rl}
  (R_1,R_2):\;
    R_1&\hspace{-2.5mm}<\bbI(U_1;Y\lvert U_0,U_2)-\bbI(U_1;Z|U_0),\\
    R_2&\hspace{-2.5mm}<\bbI(U_2;Y\lvert U_0,U_1)-\bbI(U_2;Z|U_0),\\
    R_1+R_2&\hspace{-2.5mm}<\bbI(U_1,U_2;Y|U_0)-\bbI(U_1,U_2;Z|U_0),
	\end{array}
\hspace{-2mm}\right\},\label{eq:R1R2_Cardinality}
\end{align}
for some,
\begin{align}
    &p(u_0,u_1,u_2,x_1,x_2,y,z)\nonumber\\
    &=p_{U_0}(u_0)p_{U_1|U_0}(u_1|u_0)p_{U_2|U_0}(u_2|u_0)\nonumber\\
    &\quad\times\tra\br{\pr{L_1(x_1\lvert u_0,u_1)\otimes L_2(x_2\lvert u_0,u_2)}\Psi_{E_1E_2}}\nonumber\\
    &\quad\times W_{YZ\lvert X_1X_2}(y,z\lvert x_1,x_2).\label{eq:Joint_Dist_Single_Letter_Cardinality}
\end{align}
\end{subequations}
Fix a joint state $\Psi_{E_1E_2}$, a collection of \acp{POVM}, the value $U_0 = u_0$, and the conditional distribution $p_{U_2|U_0}(\cdot \mid u_0)$. For each fixed $p_{U_1}$, the corresponding region in~\eqref{eq:R1R2_Cardinality} has a pentagonal structure. Therefore, to characterize the boundary of the achievable region, it suffices to maximize a linear functional of the rates. In particular, we consider maximizing the following function:
\begin{align}
G_\lambda\pr{p_{U_1}}=&(1-\lambda)\!\pr{\bbI(U_1;Y|U_2)-\bbI(U_1;Z)}\nonumber\\
&
+\lambda\!\pr{\bbI(U_2;Y)-\bbI(U_2;Z|U_1)},\label{eq:Linear_Com_rates}
\end{align}where $\lambda\in\sbr{\tfrac{1}{2},1}$ and for simplicity, we omit the conditioning on the \ac{RV} $U_0=u_0$. 
For a given \ac{PMF} $p(u_1,u_2,x_1,x_2,y,z)$, distributed according to \eqref{eq:Joint_Dist_Single_Letter_Cardinality} define the perturbation as follows,
\begin{align*}
    p_\epsilon(u_1,u_2,x_1,x_2,y,z)&=p(u_1,u_2,x_1,x_2,y,z)\pr{1+\epsilon\phi(u_1)}.
\end{align*}We assume that $1+\epsilon\phi(u_1)>0$, for all $u_1$, and $\bbE[\phi(U_1)]=\sum_{u_1}p(u_1)\phi(u_1)=0$ so that the distribution $p_\epsilon(u_1,u_2,x_1,x_2,y,z)$ is a valid \ac{PMF}. Note that $\phi(u_1)$ may also depend on $u_0$, since we are working under the conditioning $U_0 = u_0$. We further impose the condition
\begin{align}
&\bbE[\phi(U_1)| X_1 = x_1, X_2 = x_2]\nonumber\\
&= \sum_{u_1} p(u_1 | x_1, x_2)\,\phi(u_1)
= 0,\,\, \forall (x_1,x_2)\in\calX_1\times\calX_2,
\label{eq:Second_Assumption_Purtub}
\end{align}
which ensures that $p_\epsilon(x_1,x_2) = p(x_1,x_2)$.  
In other words, the perturbation preserves the marginal \acp{PMF} of $X_1$ and $X_2$, and consequently those of $Y$ and $Z$ as well; that is, it leaves $p_{X_1X_2YZ}$ unchanged. 
Finally, note that, since there are $\abs{\calX_1}\abs{\calX_2} + 1$ linear constraints, a perturbation satisfying the linear equation system \eqref{eq:Second_Assumption_Purtub} together with the constraint $\bbE[\phi(U_1)]=0$ exists whenever $\abs{\calU_1} \ge \abs{\calX_1}\abs{\calX_2} + 1$. Now we prove that the perturbed distribution $p_\epsilon(\cdot)$ also preserves the \ac{PMF} structure in \eqref{eq:Joint_Dist_Single_Letter_Cardinality}. Consider
\begin{align*}
   &p_\epsilon(u_1,u_2,x_1,x_2,y,z)=p(u_1,u_2,x_1,x_2,y,z)\pr{1+\epsilon\phi(u_1)}\nonumber\\
    &=\pr{p_{U_1}(u_1)\pr{1+\epsilon\phi(u_1)}}p_{U_2}(u_2)\nonumber\\
    &\quad\times\tra\br{\pr{L_1(x_1\lvert u_1)\otimes L_2(x_2\lvert u_2)}\Psi_{E_1E_2}}\nonumber\\
    &\quad \times W_{YZ\lvert X_1X_2}(y,z\lvert x_1,x_2)\nonumber\\
    &=p_\epsilon(u_1)p_{U_2}(u_2)\nonumber\\
    &\quad\times\tra\br{\pr{L_1(x_1\lvert u_1)\otimes L_2(x_2\lvert u_2)}\Psi_{E_1E_2}}\nonumber\\
    &\quad\times W_{YZ\lvert X_1X_2}(y,z\lvert x_1,x_2),
\end{align*}where $p_\epsilon(u_1)\triangleq p_{U_1}(u_1)\pr{1+\epsilon\phi(u_1)}$. Now, define the \acp{RV} $\pr{\bar{U}_1,\bar{U}_2,\bar{X}_1,\bar{X}_2,\bar{Y},\bar{Z}}\sim p_\epsilon$. We have,
\begin{subequations}\label{eq:MIs_Cardinality}
\begin{align}
    &\bbI\pr{\bar{U}_1;\bar{Y}|\bar{U}_2}\nonumber\\
    &=\bbH\pr{\bar{U}_2,\bar{Y}}-\bbH\pr{\bar{U}_2}+\bbH\pr{\bar{U}_1,\bar{U}_2}-\bbH\pr{\bar{U}_1,\bar{U}_2,\bar{Y}}\nonumber\\
    &=\bbH\pr{\bar{U}_2,\bar{Y}}-\bbH\pr{U_2}+\bbH\pr{U_1,U_2}-\bbH\pr{U_1,U_2,Y}\nonumber\\
    &\quad+\epsilon\pr{\bbH_\phi\pr{U_1,U_2}-\bbH_\phi\pr{U_1,U_2,Y}}\label{eq:MI_U2YgU1}\\
    &\bbI\pr{\bar{U}_1;\bar{Z}}\nonumber\\
    &=\bbH\pr{\bar{Z}}+\bbH\pr{\bar{U}_1}-\bbH\pr{\bar{U}_1,\bar{Z}}\nonumber\\
    &=\bbH\pr{Z}+\bbH\pr{U_1}-\bbH\pr{U_1,Z}\nonumber\\
    &\quad+\epsilon\pr{\bbH_\phi\pr{U_1}-\bbH_\phi\pr{U_1,Z}}\label{eq:MI_U1Z}\\
    &\bbI\pr{\bar{U}_2;\bar{Y}}\nonumber\\
    &=\bbH\pr{\bar{Y}}+\bbH\pr{\bar{U}_2}-\bbH\pr{\bar{U}_2,\bar{Y}}\nonumber\\
    &=\bbH\pr{Y}+\bbH\pr{U_2}-\bbH\pr{\bar{U}_2,\bar{Y}}\label{eq:MI_U2Z}\\
    &\bbI\pr{\bar{U}_2;\bar{Z}|\bar{U}_1}\nonumber\\
    &=\bbH\pr{\bar{U}_1,\bar{Z}}-\bbH\pr{\bar{U}_1}+\bbH\pr{\bar{U}_1,\bar{U}_2}-\bbH\pr{\bar{U}_1,\bar{U}_2,\bar{Z}}\nonumber\\
    &=\bbH\pr{U_1,Z}-\bbH\pr{U_1}+\bbH\pr{U_1,U_2}-\bbH\pr{U_1,U_2,Z}\nonumber\\
    &\quad+\epsilon\big(\bbH_\phi\pr{U_1,Z}-\bbH_\phi\pr{U_1}\nonumber\\
    &\quad+\bbH_\phi\pr{U_1,U_2}-\bbH_\phi\pr{U_1,U_2,Z}\big),\label{eq:MI_U1YgU2}
\end{align}where
\begin{align}
    &\bbH_\phi(U_1,U_2)\triangleq-\sum_{u_1,u_2}p(u_1)p(u_2)\phi(u_1)\log p(u_1)p(u_2),\nonumber\\
    &\bbH_\phi(U_1,U_2,Y)\triangleq-\sum_{u_1,u_2,y}p(u_1)p(u_2)W(y|u_1,u_2)\phi(u_1)\nonumber\\
    &\quad\times\log p(u_1)p(u_2)W(y|u_1,u_2),\nonumber\\
    &\bbH_\phi(U_1)\triangleq-\sum_{u_1}p(u_1)\phi(u_1)\log p(u_1),\nonumber\\
    &\bbH_\phi(U_1,Z)\triangleq-\sum_{u_1,z}p(u_1)W(z|u_1)\phi(u_1)\log p(u_1)W(z|u_1),\nonumber\\
    &\bbH_\phi(U_1,U_2,Z)\triangleq-\sum_{u_1,u_2,z}p(u_1)p(u_2)W(z|u_1,u_2)\phi(u_1)\nonumber\\
    &\quad\times\log p(u_1)p(u_2)W(z|u_1,u_2).\nonumber
\end{align}
\end{subequations}
Suppose that $p\pr{u_1,u_2,x_1,x_2,y,z}$ attains the maximum of \eqref{eq:Linear_Com_rates}, therefore
\begin{subequations}
\begin{align}
 &\frac{\partial}{\partial\epsilon}\big[(1-\lambda)\!\pr{\bbI(\bar{U}_1;\bar{Y}|\bar{U}_2)-\bbI(\bar{U}_1;\bar{Z})}\nonumber\\
 &\quad
+\lambda\!\pr{\bbI(\bar{U}_2;\bar{Y})-\bbI(\bar{U}_2;\bar{Z}|\bar{U}_1)}\!\big]\Big|_{\epsilon=0}=0,\label{eq:First_Derivative2}\\
 &\frac{\partial^2}{\partial\epsilon^2}\big[(1-\lambda)\!\pr{\bbI(\bar{U}_1;\bar{Y}|\bar{U}_2)-\bbI(\bar{U}_1;\bar{Z})}\nonumber\\
 &\quad
+\lambda\!\pr{\bbI(\bar{U}_2;\bar{Y})-\bbI(\bar{U}_2;\bar{Z}|\bar{U}_1)}\!\big]\Big|_{\epsilon=0}\le0.\label{eq:Second_Derivative2}
\end{align}
Considering \eqref{eq:MIs_Cardinality}, \eqref{eq:First_Derivative2} reduces to
\begin{align}
    &(1-\lambda)((\bbH_\phi\pr{U_1,U_2}-\bbH_\phi\pr{U_1,U_2,Y}-\bbH_\phi\pr{U_1}\nonumber\\
    &\quad+\bbH_\phi\pr{U_1,Z}))+\lambda(-\bbH_\phi\pr{U_1,Z}+\bbH_\phi\pr{U_1}\nonumber\\
    &\quad-\bbH_\phi\pr{U_1,U_2}+\bbH_\phi\pr{U_1,U_2,Z})=0,\label{eq:First_Drevitive_Equality2}
\end{align}and \eqref{eq:Second_Derivative2} leads to
\begin{align}
    -(2\lambda-1)\frac{\partial^2}{\partial\epsilon^2}\bbH\pr{\bar{U}_2,\bar{Y}}\Big|_{\epsilon=0}\le0.\label{eq:Second_Derivative3}
\end{align}From \cite[Lemma~2,~Part~2]{Perurbation}, \eqref{eq:Second_Derivative3} is equivalent to $\bbE\sbr{\bbE\sbr{\phi(U_1)|U_2,Y}^2}\le0$, which holds with equality if $\bbE\sbr{\phi(U_1)|U_2=u_2,Y=y}$, for all $(u_2,y)$ in the support of $p(u_2,y)$. This implies that
\begin{align}
    p_\epsilon(u_2,y)=p(u_2,y).\label{eq:Perserving_Joint_Dist}
\end{align}
\end{subequations}Therefore, considering \eqref{eq:MIs_Cardinality}, \eqref{eq:First_Drevitive_Equality2}, and \eqref{eq:Perserving_Joint_Dist}, we have
\begin{align*}
&(1-\lambda)\!\pr{\bbI(\bar{U}_1;\bar{Y}|\bar{U}_2)-\bbI(\bar{U}_1;\bar{Z})}
+\lambda\!(\bbI(\bar{U}_2;\bar{Y})\nonumber\\
&-\bbI(\bar{U}_2;\bar{Z}|\bar{U}_1))=(1-\lambda)\!\pr{\bbI(U_1;Y|U_2)-\bbI(U_1;Z)}\nonumber\\
&+\lambda\!\pr{\bbI(U_2;Y)-\bbI(U_2;Z|U_1)}.
\end{align*}
Thus, if $p(u_1,u_2,x_1,x_2,y,z)$ attains the minimum of the weighted sum rate, then this minimum is preserved under any valid perturbation $p_\epsilon(u_1,u_2,x_1,x_2,y,z)$ that satisfies~\eqref{eq:Second_Assumption_Purtub}. We may choose $\epsilon$ such that $\min\limits_{u_1}\bigl(1+\epsilon\,\phi(u_1)\bigr)=0$. 
Let $u_1^\star$ denote the value achieving this minimum. Then it follows immediately that
$p_\epsilon(u_1^\star)=0$, and hence one symbol of $U_1$ can be removed without affecting
\eqref{eq:Linear_Com_rates}. In other words, the cardinality of the auxiliary \ac{RV}
$U_1$ decreases by one. This reduction can proceed until
$\abs{\calU_1}=\abs{\calX_1}\abs{\calX_2}$, at which point a nonzero perturbation
$\phi(u_1)$ satisfying \eqref{eq:Second_Assumption_Purtub} may no longer exist.
Consequently, the alphabet size of $U_1$ can be constrained to satisfy
$\abs{\calU_1}\le \abs{\calX_1}\abs{\calX_2}$.
When the time-sharing \ac{RV} $U_0$ is incorporated, this bound becomes
$\abs{\calU_1} \le 3\,\abs{\calX_1}\abs{\calX_2}$.
By symmetry, the same conclusion applies to the auxiliary \ac{RV} $U_2$.

\section{Proof of Theorem~\ref{thm:NLetter_Capacity}}
\label{proof:thm:NLetter_Capacity}
Consider any sequence of $\pr{2^{nR_1},2^{nR_2},n}$ codes for the \ac{MAWTC} $W_{YZ|X_1X_2}$, with a bipartite state $\Psi_{E_1E_2}^\on$ shared among the transmitters, that simultaneously satisfy  
(i) the reliability requirement $\mathsf{e}(C_n) \le \delta_n$ with $\delta_n \to 0$ as $n \to \infty$, and  
(ii) the weak secrecy condition $\tfrac{1}{n}\, \bbI(M_1,M_2; Z^n) \le \epsilon_n$ with $\epsilon_n \to 0$ as $n \to \infty$. Note that semantic security is strictly stronger than weak security, and therefore, any outer bound proven for the weak security setting also serves as an outer bound for the semantic security setting. By the Fano inequality, it follows that
\begin{subequations}\label{eq:Fanos}
    \begin{align}
        \bbH\pr{M_1,M_2|Y^n}&\le n\delta_n.\label{eq:Fano}
    \end{align}Therefore, we also have
    \begin{align}
        \bbH\pr{M_1|M_2,Y^n}&\le n\delta_n,\label{eq:Fano_1}\\
        \bbH\pr{M_2|M_1,Y^n}&\le n\delta_n.\label{eq:Fano_2}
    \end{align}
\end{subequations}
\begin{figure*}[b!]
    \hrulefill
    \setcounter{equation}{47}
\begin{align}
    &\calL_{E_1\to X_1}^{(t,m_1,s_1)}\otimes\calL_{E_2\to X_2}^{(t,m_2,s_2)}\pr{\Psi_{E_1E_2}}=\sum_{k,\ell}\calL_{E_1\to X_1}^{(t,m_1,s_1)}\otimes\calL_{E_2\to X_2}^{(t,m_2,s_2)}\pr{Q_{1,k}\otimes Q_{2,\ell}}\nonumber\\
    &=\sum_{k,\ell}\sum_{a_1\in\calX_1}\br{\sum_{x_1^n\in\calX_1^n:\,x_1(t)=a_1}\tra\pr{L_{x_1^n}^{\pr{m_1,s_1}}Q_{1,k}}}\den{a_1}{a_1}\otimes\sum_{a_2\in\calX_2}\br{\sum_{x_2^n\in\calX_2^n:\,x_2(t)=a_2}\tra\pr{L_{x_2^n}^{\pr{m_2,s_2}}Q_{2,\ell}}}\den{a_2}{a_2}\nonumber\\
    &=\sum_{\pr{a_1,a_2}\in\calX_1\times\calX_2}\sbr{\sum_{\substack{\pr{x_1^n,x_2^n}\in\calX_1^n\times\calX_2^n:\\x_1(t)=a_1,x_2(t)=a_2}}\sum_{k,\ell}\tra\pr{\pr{L_{x_1^n}^{\pr{m_1,s_1}}\otimes L_{x_2^n}^{\pr{m_2,s_2}}}\pr{Q_{1,k}\otimes Q_{2,\ell}}}}\den{a_1a_2}{a_1a_2}\nonumber\\
    &=\sum_{\pr{a_1,a_2}\in\calX_1\times\calX_2}\sbr{\sum_{\substack{\pr{x_1^n,x_2^n}\in\calX_1^n\times\calX_2^n:\\x_1(t)=a_1,x_2(t)=a_2}}\tra\pr{\pr{L_{x_1^n}^{\pr{m_1,s_1}}\otimes L_{x_2^n}^{\pr{m_2,s_2}}}\Psi_{E_1E_2}}}\den{a_1a_2}{a_1a_2}.\label{eq:additional_1}
\end{align}
    \setcounter{equation}{41}
\end{figure*}
Now we have,
\begin{align}
    nR_1&=\bbH(M_1)\nonumber\\
    &\mathop=\limits^{(a)}\bbH(M_1|M_2)\nonumber\\
    &\mathop\le\limits^{(b)}\bbI(M_1;Y^n|M_2)+n\delta_n\nonumber\\
    &\mathop\le\limits^{(c)}\bbI(M_1;Y^n|M_2)-\bbI(M_1;Z^n)+n(\delta_n+\epsilon_n),\label{eq:proof_Upper_R1}
\end{align}where
\begin{itemize}
    \item[$(a)$] follows since $M_1$ and $M_2$ are independent;
    \item[$(b)$] follows from Fano's inequality in \eqref{eq:Fano_1};
    \item[$(c)$] follows from the weak security constraint and since $\bbI(M_1;Z^n)\le\bbI(M_1,M_2;Z^n)\le n\epsilon_n$.
\end{itemize}Similarly, one can show that
\begin{subequations}
    \begin{align}
        &nR_2\le\bbI(M_2;Y^n|M_1)-\bbI(M_2;Z^n)+n(\delta_n+\epsilon_n),\label{eq:proof_Upper_R2}\\
        &n(R_1+R_2)\le\bbI(M_1,M_2;Y^n)-\bbI(M_1,M_2;Z^n)\nonumber\\
        &\qquad+n(\delta_n+\epsilon_n).\label{eq:proof_Upper_R1R2}
    \end{align}
\end{subequations}
Also, when $R_1$ or $R_2$ are not intended to be transmitted securely, similar to \eqref{eq:proof_Upper_R1}, one can show that
\begin{subequations}\label{eq:proof_Upper_R1R22}
    \begin{align}
        &nR_1\le\bbI(M_1;Y^n|M_2)-\bbI(M_1;Z^n|M_2)+n(\delta_n+\epsilon_n),\label{eq:proof_Upper_R12}\\
        &nR_2\le\bbI(M_2;Y^n|M_1)-\bbI(M_2;Z^n|M_1)+n(\delta_n+\epsilon_n),\label{eq:proof_Upper_R22}
    \end{align}where we use $\bbI(M_1;Z^n|M_2)\le\bbI(M_1,M_2;Z^n)\le n\epsilon_n$ and $\bbI(M_2;Z^n|M_1)\le\bbI(M_1,M_2;Z^n)\le n\epsilon_n$ instead of $\bbI(M_1;Z^n)\le\bbI(M_1,M_2;Z^n)\le n\epsilon_n$ and $\bbI(M_2;Z^n)\le\bbI(M_1,M_2;Z^n)\le n\epsilon_n$, respectively.
\end{subequations}
The proof follows by identifying $U_1^n$ and $U_2^n$ as $M_1$ and $M_2$, respectively.
\section{Proof of Theorem~\ref{thm:Single_Letter_Outter}}
\label{proof:thm:Single_Letter_Outter}
Consider any sequence of $\pr{2^{nR_1},2^{nR_2},n}$ codes for the \ac{MAWTC} $W_{YZ|X_1X_2}$, with a bipartite state $\Psi_{E_1E_2}^\on$ shared among the transmitters, that simultaneously satisfy (i) the reliability requirement $\mathsf{e}(C_n) \le \delta_n$ with $\delta_n \to 0$ as $n \to \infty$, and (ii) the weak secrecy condition $\tfrac{1}{n}\, \bbI(M_1,M_2; Z^n) \le \epsilon_n$ with $\epsilon_n \to 0$ as $n \to \infty$.
Given a uniformly distributed message $M_i$ and uniformly distributed local randomness $S_i$, for $i\in\{1,2\}$, the Transmitter~$i$ encodes its message by performing the measurement $\calL_i^{(M_i,S_i)}=\br{L_{x_i^n}^{(M_i,S_i)}}_{x_i^n\in\calX_i^n}$ on her share of the entangled system $E_i$, producing an outcome $X_i^n$, which is then sent over the channel.

Similar to the proof of Theorem~\ref{thm:NLetter_Capacity} in Appendix~\ref{proof:thm:NLetter_Capacity}, we establish Theorem~\ref{thm:Single_Letter_Outter} under the weak security criterion. Since weak security is implied by semantic security, the resulting region also serves as a valid outer bound for the semantic security setting.

\setcounter{equation}{44}
Now we have,
\begin{align}
    &nR_1=\bbH(M_1)\nonumber\\
    &\mathop\le\limits^{(a)}\bbI(M_1;Y^n)+n\delta_n\nonumber\\
    &\mathop\le\limits^{(b)}\bbI(M_1;Y^n)-\bbI(M_1;Z^n)+n(\delta_n+\epsilon_n)\nonumber\\
    &=\sum_{t=1}^n\sbr{\bbI(M_1;Y(t)|Y^{t-1})-\bbI(M_1;Z(t)|Z_{t+1}^n)}+n(\delta_n+\epsilon_n)\nonumber\\
    &\mathop=\limits^{(c)}\sum_{t=1}^n\sbr{\bbI(M_1;Y(t)|Y^{t-1},Z_{t+1}^n)-\bbI(M_1;Z(t)|Y^{t-1},Z_{t+1}^n)}\nonumber\\
    &\qquad+n(\delta_n+\epsilon_n),\label{eq:proof_Upper_SLetter_R11}
\end{align}where
\begin{itemize}
    \item[$(a)$] follows from Fano's inequality in \eqref{eq:Fano} and since $\bbH(M_1|Y^n)\le\bbH(M_1,M_2|Y^n)\le n\delta_n$;
    \item[$(b)$] follows from the weak security constraint and since $\bbI(M_1;Z^n)\le\bbI(M_1,M_2;Z^n)\le n\epsilon_n$;
    \item[$(c)$] follows from K\"orner-Marton-Csisz\'{a}r sum identity \cite[Lemma~17.12]{BCC:IT78}.
\end{itemize}
We also have,
\begin{align}
    &nR_1=\bbH(M_1)\nonumber\\
    &\mathop=\limits^{(a)}\bbH(M_1|M_2)\nonumber\\
    &\mathop\le\limits^{(b)}\bbI(M_1;Y^n|M_2)+n\delta_n\nonumber\\
    &\mathop\le\limits^{(c)}\bbI(M_1;Y^n|M_2)-\bbI(M_1;Z^n|M_2)+n(\delta_n+\epsilon_n)\nonumber\\
    &=\sum_{t=1}^n\big[\bbI(M_1;Y(t)|M_2,Y^{t-1})\nonumber\\
    &\qquad-\bbI(M_1;Z(t)|M_2,Z_{t+1}^n)\big]+n(\delta_n+\epsilon_n)\nonumber\\
    &\mathop=\limits^{(d)}\sum_{t=1}^n\big[\bbI(M_1;Y(t)|M_2,Y^{t-1},Z_{t+1}^n)\nonumber\\
    &\qquad-\bbI(M_1;Z(t)|M_2,Y^{t-1},Z_{t+1}^n)\big]+n(\delta_n+\epsilon_n),\label{eq:proof_Upper_SLetter_R11_2}
\end{align}where
\begin{itemize}
    \item[$(a)$] follows since $M_1$ and $M_2$ are independent;
    \item[$(b)$] follows from Fano's inequality in \eqref{eq:Fano_1} and since $\bbH(M_1|M_2,Y^n)\le\bbH(M_1,M_2|Y^n)\le n\delta_n$;
    \item[$(c)$] follows from the weak security constraint and since $\bbI(M_1;Z^n|M_2)\le\bbI(M_1,M_2;Z^n)\le n\epsilon_n$;
    \item[$(d)$] follows from K\"orner-Marton sum identity \cite[Lemma~17.12]{BCC:IT78}.
\end{itemize}
\begin{figure*}[b!]
    \hrulefill
    \setcounter{equation}{51}
\begin{subequations}\label{eq:Separating}
    \begin{align}
        p_{\text{atyp}}^{(1)}&\triangleq\frac{1}{2^{n(R_1+R_2)}}\sum_{z^n\in\calZ^n}\sum_{(s_1,s_2)\in\calS_1\times\calS_2}W_{Z|U_0U_1U_2}^\on\pr{z^n|U_0^n,U_1^n(s_1),U_2^n(s_2)}\indic{1}_{\br{\pr{U_0^n,U_1^n(s_1),z^n}\notin\calA_{1,\epsilon}^{(n)}}},\label{eq:ATypical_U1}\\
        p_{\text{atyp}}^{(2)}&\triangleq\frac{1}{2^{n(R_1+R_2)}}\sum_{z^n\in\calZ^n}\sum_{(s_1,s_2)\in\calS_1\times\calS_2}W_{Z|U_0U_1U_2}^\on\pr{z^n|U_0^n,U_1^n(s_1),U_2^n(s_2)}\indic{1}_{\br{\pr{U_0^n,U_1^n(s_1),U_2^n(s_2),z^n}\notin\calA_{2,\epsilon}^{(n)}}},\label{eq:ATypical_U1U2}\\
        p_{\text{typ}}(z^n)&\triangleq\frac{1}{2^{n(R_1+R_2)}}\sum_{(s_1,s_2)\in\calS_1\times\calS_2}\frac{W_{Z|U_0U_1U_2}^\on\pr{z^n|U_0^n,U_1^n(s_1),U_2^n(s_2)}}{q_{Z|U_0}^\on(z^n|U_0^n)}\nonumber\\
        &\quad\times\indic{1}_{\br{\pr{U_0^n,U_1^n(s_1),z^n}\in\calA_{1,\epsilon}^{(n)}}\bigcap\br{\pr{U_0^n,U_1^n(s_1),U_2^n(s_2),z^n}\in\calA_{2,\epsilon}^{(n)}}}.\label{eq:Typical_U1U2}
    \end{align}
    \end{subequations}
        \setcounter{equation}{46}
\end{figure*}

Next, in order to derive a single-letter upper bound, we further examine the encoding \acp{POVM} 
$\calL_1 \otimes \calL_2$ acting across $n$ consecutive channel uses. Owing to the use of the same coding technique, this analysis closely parallels the proof in \cite[Section~VIII]{Uzi_Entangled_MAC}. Since Transmitter~$i$, for $i\in\br{1,2}$, performs a measurement that depends on its message $M_i$ and local randomness $S_i$, the resulting channel inputs $X_1(t)$ and $X_2(t)$ can be viewed as the outcomes of the corresponding measurements, as formally defined below. Consider $\calL_{E_1\to X_1}^{(t,m_1,s_1)}\otimes\calL_{E_2\to X_2}^{(t,m_2,s_2)}$ such that
    \begin{align}
        \calL_{E_i\to X_i}^{(t,m_i,s_i)}&=\sum_{a_i\in\calX_i}\br{\sum_{x_i^n\in\calX_i^n:\,x_i(t)=a_i}\hspace{-3mm}\tra\pr{L_{x_i^n}^{\pr{m_i,s_i}}Q_i}}\den{a_i}{a_i},\label{eq:Measurment_t}
    \end{align}
for $i\in\br{1,2}$ and every pair of $Q_1$ and $Q_2$ on $\calH_{E_1}$ and $\calH_{E_2}$, respectively. Now, let $\Psi_{E_1E_2}=\sum_{k,\ell}Q_{1,k}\otimes Q_{2,\ell}$ be an arbitrary decomposition of the bipartite state. Therefore, by linearity, we have \eqref{eq:additional_1} at the bottom of the previous page. 
Therefore, the channel inputs $X_1(t)$ and $X_2(t)$ can be obtained as the outcomes of product measurements of the form $L_1\!\pr{x_1 | t, U_0(t), U_1(t)}$ and $L_2\!\pr{x_2 | t, U_0(t), U_2(t)}$, where we define 
$U_0(t)\triangleq\pr{Y^{t-1},Z_{t+1}^n}$, $U_1(t)\triangleq\pr{M_1,U_0(t)}$, and $U_2(t)\triangleq\pr{M_2,U_0(t)}$. Observe that $\pr{Y^{t-1},Z_{t+1}^n}$ is a function of $S_1$ and $S_2$, and therefore $L_1\!\pr{x_1 | t, U_0(t), U_1(t)}$ and $L_2\!\pr{x_2 | t, U_0(t), U_2(t)}$ also are functions of $S_1$ and $S_2$.

Now, from \eqref{eq:proof_Upper_SLetter_R11} we have,
\begin{align*}
     nR_1&\le\sum_{t=1}^n\sbr{\bbI(U_1(t);Y(t)|U_0(t))-\bbI(U_1(t);Z(t)|U_0(t))}\nonumber\\
     &\qquad+n(\delta_n+\epsilon_n)\nonumber\\
     &\mathop=\limits^{(a)}\bbI(U_1(T);Y(T)|U_0(T))-\bbI(U_1(T);Z(T)|U_0(T))\nonumber\\
     &\qquad+n(\delta_n+\epsilon_n)\nonumber\\
     &\mathop=\limits^{(b)}\bbI(U_1;Y|U_0)-\bbI(U_1;Z|U_0)+n(\delta_n+\epsilon_n),
\end{align*}where $(a)$ since the index $T$ is independent of all the involved \acp{RV} and is uniformly at random drawn over $[n]$, and $(b)$ follows by defining $U_0\triangleq\pr{T,U_0(T)}$, $U_1\triangleq U_1(T)$, $U_2\triangleq U_2(T)$, $Y\triangleq Y(T)$, and $Z\triangleq Z(T)$. Similarly, from \eqref{eq:proof_Upper_SLetter_R11_2} we have,
\begin{align*}
     nR_1&\le\bbI(U_1;Y|U_0,U_2)-\bbI(U_1;Z|U_0,U_2)+n(\delta_n+\epsilon_n).
\end{align*}
By an analogous argument, one can show that
    \begin{align*}
     &nR_2\le\bbI(U_2;Y|U_0)-\bbI(U_2;Z|U_0)+n(\delta_n+\epsilon_n),\\
      &nR_2\le\bbI(U_2;Y|U_0,U_1)-\bbI(U_2;Z|U_0,U_1)+n(\delta_n+\epsilon_n),\\
     &n(R_1+R_2)\le\bbI(U_1,U_2;Y|U_0)-\bbI(U_1,U_2;Z|U_0)\nonumber\\
     &\qquad+n(\delta_n+\epsilon_n).
\end{align*}
This completes the proof of Theorem~\ref{thm:Single_Letter_Outter}.

    \begin{figure*}[b!]
    \hrulefill
    \setcounter{equation}{55}
    \begin{align}
        &\bbP_{B_n}\left(\frac{1}{2^{n(R_1+R_2)}}\sum_{\pr{j_1,j_2}\in\calJ_1\times\calJ_2}\sum_{z^n\in\calZ^n}p_{Z|V_0V_1V_2}^\on\pr{z^n|V_0^n,V_1^n(j_1),V_2^n(j_2)}\right.\nonumber\\
        &\qquad\qquad\times\indic{1}_{\br{\pr{V_0^n,V_1^n(j_1),V_2^n(j_2),z^n}\in\calA}}>\mu(1+\delta)\Big)\le\exp\br{-2\delta^2\mu^22^{n\min\br{R_1,R_2}}}.\label{eq:Additional_2}
    \end{align}
    \hrulefill
    \begin{align}
            &\bbP_{B_n}\left(\frac{1}{2^{n(R_1+R_2)}}\sum_{\pr{j_1,j_2}\in\calJ_1\times\calJ_2}\sum_{z^n\in\calZ^n}p_{Z|V_0V_1V_2}^\on\pr{z^n|V_0^n,V_1^n(j_1),V_2^n(j_2)}\right.\nonumber\\
             &\qquad\qquad\times\indic{1}_{\br{\pr{V_0^n,V_1^n(j_1),V_2^n(j_2),z^n}\in\calA}}>\mu(1+\delta)\Big)\nonumber\\
             &=\bbP_{B_n}\left(\sum_{\pr{j_1,j_2}\in\calJ_1\times\calJ_2}\sum_{z^n\in\calZ^n}p_{Z|V_0V_1V_2}^\on\pr{z^n|V_0^n,V_1^n(j_1),V_2^n(j_2)}\right.\nonumber\\
             &\qquad\qquad\times\indic{1}_{\br{\pr{V_0^n,V_1^n(j_1),V_2^n(j_2),z^n}\in\calA}}>2^{n(R_1+R_2)}\mu(1+\delta)\Big)\nonumber\\
             &\le\bbP_{B_n}\left(\sum_{\pr{j_1,j_2}\in\calJ_1\times\calJ_2}\sum_{z^n\in\calZ^n}p_{Z|V_0V_1V_2}^\on\pr{z^n|V_0^n,V_1^n(j_1),V_2^n(j_2)}\right.\nonumber\\
             &\qquad\qquad\times\indic{1}_{\br{\pr{V_0^n,V_1^n(j_1),V_2^n(j_2),z^n}\in\calA}}>2^{n(R_1+R_2)}\pr{\bbP\pr{\pr{V_0^n,V_1^n,V_2^n,Z^n}\in\calA}+\mu\delta}\Big)\nonumber\\
             &=\sum_{v_0^n}\bbP(V_0^n=v_0^n)\bbP_{B_n|V_0^n=v_0^n}\left(\sum_{\pr{j_1,j_2}\in\calJ_1\times\calJ_2}\sum_{z^n\in\calZ^n}p_{Z|V_0V_1V_2}^\on\pr{z^n|V_0^n,V_1^n(j_1),V_2^n(j_2)}\right.\nonumber\\
             &\qquad\qquad\times\indic{1}_{\br{\pr{V_0^n,V_1^n(j_1),V_2^n(j_2),z^n}\in\calA}}>2^{n(R_1+R_2)}\pr{\bbP\pr{\pr{V_0^n,V_1^n,V_2^n,Z^n}\in\calA}+\mu\delta}|V_0^n=v_0^n\Big)\nonumber\\
             &\mathop\le^{(a)}\sum_{v_0^n}\bbP(V_0^n=v_0^n)\exp\pr{-2\frac{2^{2n(R_1+R_2)}\mu^2\delta^2}{2^{n\max\br{R_1,R_2}}2^{n(R_1+R_2)}}}\nonumber\\
             &\le\exp\pr{-2\mu^2\delta^22^{n\min\br{R_1,R_2}}}.\label{eq:Additional_3}
        \end{align}
\end{figure*}

\setcounter{equation}{48}
\section{Proof of Lemma~\ref{lemma:Resolvability_1}}
\label{proof:thm:Resolvability_1}
We begin by establishing a slightly weaker version of Lemma~\ref{lemma:Resolvability_1}, which is stated as follows.
\begin{lemma}
    \label{lemma:Resolvability_12}
For a discrete memoryless \ac{MAC} $W_{Z|X_1,X_2}$, with entangled transmitters if $(R_1,R_2)$ belongs to
\begin{subequations}\label{eq:Resolvability_12}
    \begin{align}
        \bigcup\limits_{p_{U_0}p_{U_1\lvert U_0}p_{U_2\lvert U_0},\,\calL_1\otimes\calL_2,\,\Psi_{E_1E_2}}\left\{ \begin{array}{l}
 (R_1,R_2)\in\bbR_+^2: \\ 
 R_1 > \bbI(U_1;Z\lvert U_0) \\ 
 R_2 > \bbI(U_2;Z\lvert U_0,U_1) \\
 \end{array} \right\}.\label{eq:Resolvability_121}
    \end{align}  
    Then, for any $u_0^n\in\calU_0^n$ generated \ac{iid} according to $P_{U_0}$, there exist $\gamma_1,\gamma_2>0$ such that for block length $n$ large enough we have
    \begin{align}
        &\bbP_\mu\pr{\bbV\pr{p_{Z^n|C_n},q_{Z|U_0}^\on(\cdot|u_0^n)}>\exp\pr{-\gamma_1n}}\nonumber\\
        &\qquad\le\exp\pr{-\exp(\gamma_2n)},\label{eq:Resolvability_122}
    \end{align}where $q_{Z|U_0}$ is defined in \eqref{eq:Target_Dist}.
    \end{subequations}
\end{lemma}Lemma~\ref{lemma:Resolvability_1} is then proved using a time-sharing argument. 
\begin{proof}
Before proceeding with the proof of Lemma~\ref{lemma:Resolvability_12}, we first present two lemmas on the concentration bounds introduced in \cite{Concentration}, which play a key role in our distribution approximation analysis.
\begin{lemma}[Chernoff-Hoeffding Bound {\cite[Example~1.1]{Concentration_Book}}]
\label{lemma:Hoeffding}
Let $A \triangleq \sum_{i=1}^n A_i$, where the \acp{RV} $\br{A_i}_{i=1}^n$ are independent and take values in $[0,1]$. Suppose that $\bbE[A] \le \mu$. Then, for any $0 < \delta < 1$,
    \begin{align*}
        \bbP\pr{A>\mu(1+\delta)}\le\exp\pr{-\frac{\mu\delta^2}{3}}.
    \end{align*}
\end{lemma}
\begin{lemma}[\hspace{-0.01mm}{\cite[Theorem~2.1]{Concentration_2}}]
\label{lemma:Janson}
    Let $A \triangleq \sum_{i=1}^n A_i$, where the \acp{RV} $\br{A_i}_{i=1}^n$ take values in $[0,1]$ and can be partitioned into $\chi$ subsets such that the \acp{RV} within each subset are mutually independent. Then, for any $\delta > 0$,
        \begin{align*}
        \bbP\pr{A>(\bbE[A]+\delta)}\le\exp\pr{-\frac{2\delta^2}{n\cdot\chi}}.
    \end{align*}
\end{lemma}
To describe the typical sets, we first define the information density $i_{p_{U_1,Z|U_0}}:\calU_0\times\calU_1\times\calZ\to\bbR_+$, and $i_{p_{U_2,Z|U_0,U_1}}:\calU_0\times\calU_1\times\calU_2\times\calZ\to\bbR_+$ as
\begin{subequations}
    \begin{align}
        &i_{p_{U_1,Z|U_0}}\pr{u_1;z|u_0}\triangleq\log\pr{\frac{p_{Z|U_0=u_0,U_1=u_1}}{p_{Z|U_0=u_0}}(z)},\label{eq:Info_Density_XtZ}\\
        &i_{p_{U_2,Z|U_0,U_1}}\pr{u_2;z|u_0,u_1}\triangleq\log\pr{\frac{p_{Z|U_0=u_0,U_1=u_1,U_2=u_2}}{p_{Z|U_0=u_0,U_1=u_1}}(z)\!}\!.\label{eq:Info_Density_X1X2Z}
    \end{align}
    \end{subequations}Now, let $\epsilon\ge0$ be arbitrary, to be defined later, and define 
    \begin{subequations}\label{eq:Typical}
        \begin{align}
        &\calA_{1,\epsilon}^{(n)}\triangleq\Bigg\{(u_0^n,u_1^n,z^n):\frac{1}{n}i_{p_{U_1,Z|U_0}^\on}\pr{u_1^n;z^n|u_0^n}\nonumber\\
        &\qquad\quad<\bbI(U_1;Z|U_0)+\epsilon\Bigg\},\label{eq:Typical_1}\\
        &\calA_{2,\epsilon}^{(n)}\triangleq\Bigg\{(u_0^n,u_1^n,u_2^n,z^n):\frac{1}{n}i_{p_{U_2,Z|U_0,U_1}^\on}\pr{u_2^n;z^n|u_0^n,u_1^n}\nonumber\\
        &\qquad\quad<\bbI(U_2;Z|U_0,U_1)+\epsilon\Bigg\}.\label{eq:Typical_2}
    \end{align}
    \end{subequations}
By using the indicator functions, we split the $p_{Z^n|\calC_n}$ into atypical and typical parts as provided in \eqref{eq:Separating} at the bottom of the previous page. Using \eqref{eq:Separating} and the triangle inequality, we now split the total variation distance into atypical and typical parts as follows
    \setcounter{equation}{52}
    \begin{align}
        &\bbV\pr{p_{Z^n|\calC_n},q_{Z|U_0}^\on\pr{\cdot|U_0^n}}\nonumber\\
        &=\sum_{z^n\in\calZ^n}q_{Z|U_0}^\on(z^n|U_0^n)\abs{\frac{p_{Z^n|\calC_n}(z^n)}{q_{Z|U_0}^\on(z^n|U_0^n)}-1}\nonumber\\
        &\le p_{\text{atyp}}^{(1)}+p_{\text{atyp}}^{(2)}+\sum_{z^n\in\calZ^n}q_{Z|U_0}^\on(z^n|U_0^n)\sbr{p_{\text{typ}}(z^n)-1}^+.\label{eq:Total_TV_Bound}
    \end{align}
    Note that we assume that $p_{Z^n|\calC_n}$ is absolutely continuous \ac{wrt} $q_{Z|U_0}^\on(\cdot|U_0^n)$, therefore the summation in \eqref{eq:Total_TV_Bound} is bounded. The following two lemmas are essential in our proof.
    \begin{lemma}[Bound on Typical Sets]
    \label{lemma:Bounding_Typical}
        Let $\epsilon>0$ and $0<\delta<1$, $V_0\in\calV_0$, $V_1\in\calV_1$, $V_2\in\calV_2$, and $Z\in\calZ$ be four finite \acp{RV} with joint \ac{PMF} $p_{V_0V_1V_2Z}$. Also, for any $v_0^n\in\calV_0^n$, let  $B_n=\br{V_1^n(j)}_{j\in\calJ}$, where $\calJ\triangleq[2^{nR}]$, be a set of codewords of length $n$ generated \ac{iid} according to $p_{V_1|V_0}(\cdot|v_0^n)$. For any $v_2^n\in\calV_2^n$ and $z^n\in\calZ^n$, we have
        \begin{align}
            &\bbP_{B_n}\Bigg(\frac{1}{2^{nR}}\sum_{j\in\calJ}\frac{p_{Z|V_0V_1V_2}^\on\pr{z^n|v_0^n,V_1^n(j),v_2^n}}{p_{Z|V_0V_2}^\on(z^n|v_0^n,v_2^n)}\nonumber\\
            &\qquad\times\indic{1}_{\br{\pr{v_0^n,V_1^n(j),v_2^n,z^n}\in\calA_\epsilon^n}}>1+\delta\Bigg)\nonumber\\
            &\le\exp\pr{-\frac{\delta^2}{3}2^{n\pr{R-\bbI(V_1;Z|V_0,V_2)-\epsilon}}},\nonumber
        \end{align}where
        \begin{align}
            \calA_{\epsilon,n}&\triangleq\big\{(v_0^n,v_1^n,v_2^n,z^n):i(v_1^n;z^n|v_0^n,v_2^n)\nonumber\\
            &\qquad\le n\pr{\bbI(V_1;Z|V_0,V_2)+\epsilon}\big\}.\label{eq:Tipical_Set_First_Lemma}
        \end{align}
    \end{lemma}
    \begin{proof}
        We have
        \begin{align}
            &\bbP_{B_n}\Bigg(\frac{1}{2^{nR}}\sum_{j\in\calJ}\frac{p_{Z|V_0V_1V_2}^\on\pr{z^n|v_0^n,V_1^n(j),v_2^n}}{p_{Z|V_0V_2}^\on(z^n|v_0^n,v_2^n)}\nonumber\\
            &\quad\times\indic{1}_{\br{\pr{v_0^n,V_1^n(j),v_2^n,z^n}\in\calA_{\epsilon,n}}}>1+\delta\Bigg)\nonumber\\
            &=\bbP_{B_n}\Bigg(\sum_{j\in\calJ}2^{-n(\bbI(V_1;Z|V_0,V_2)+\epsilon)}\nonumber\\
            &\quad\times\frac{p_{Z|V_0V_1V_2}^\on\pr{z^n|v_0^n,V_1^n(j),v_2^n}}{p_{Z|V_0V_2}^\on(z^n|v_0^n,v_2^n)}\indic{1}_{\br{\pr{v_0^n,V_1^n(j),v_2^n,z^n}\in\calA_{\epsilon,n}}}\nonumber\\
            &\quad>2^{-n(\bbI(V_1;Z|V_0,V_2)+\epsilon-R)}(1+\delta)\Bigg).\label{eq:P_hat}
        \end{align}Note that, from \eqref{eq:Tipical_Set_First_Lemma}, each summand on the \ac{RHS} of \eqref{eq:P_hat} is bounded above by 1, and the expectation of the sum can be bounded as follows,
        \begin{align}
            &\bbE_{B_n}\Bigg[\sum_{j\in\calJ}2^{-n(\bbI(V_1;Z|V_0,V_2)+\epsilon)}\frac{p_{Z|V_0V_1V_2}^\on\pr{z^n|v_0^n,V_1^n(j),v_2^n}}{p_{Z|V_0V_2}^\on(z^n|v_0^n,v_2^n)}\nonumber\\
            &\qquad\times\indic{1}_{\br{\pr{v_0^n,V_1^n(j),v_2^n,z^n}\in\calA_\epsilon^n}}\Bigg]\nonumber\\
            &\le2^{-n(\bbI(V_1;Z|V_0,V_2)+\epsilon)}\nonumber\\
            &\qquad\times\sum_{j\in\calJ}\bbE_{B_n}\sbr{\frac{p_{Z|V_0V_1V_2}^\on\pr{z^n|v_0^n,V_1^n(j),v_2^n}}{p_{Z|V_0V_2}^\on(z^n|v_0^n,v_2^n)}}\nonumber\\
            &=2^{-n(\bbI(V_1;Z|V_0,V_2)+\epsilon-R)}.\nonumber
        \end{align}Now applying Lemma~\ref{lemma:Hoeffding} to \eqref{eq:P_hat} completes the proof of Lemma~\ref{lemma:Bounding_Typical}.
    \end{proof}
\begin{figure*}[b!]
    \hrulefill
    \setcounter{equation}{61}
\begin{align}
    &\bbP_\mu\pr{p_{\text{typ}}(z^n)>1+3\cdot2^{-n\beta}\big|\calC_{2,n}\in\mathfrak{C}_{z^n}}\nonumber\\
    &\mathop=\limits^{(a)}\sum_{u_0^n}\sum_{\calC_{1,n}}\bbP_\mu\pr{U_0^n=u_0^n}\bbP_\mu\pr{C_{1,n}=\calC_{1,n}|U_0^n=u_0^n}\bbP_\mu\left(\frac{1}{2^{nR_1}}\sum_{s_1\in\calS_1}\frac{W_{Z|U_0U_1}^\on\pr{z^n|U_0^n,U_1^n(s_1)}}{q_{Z|U_0}^\on\pr{z^n|U_0^n}}\right.\nonumber\\
    &\quad\left.\times \indic{1}_{\br{\pr{U_0^n,U_1^n(s_1),z^n}\in\calA_{1,\epsilon}^{(n)}}}\,p_{\text{typ}}^{(1)}(U_0^n,U_1^n(s_1),z^n)>1+3\cdot2^{-n\beta}\big|
    U_0^n=u_0^n,C_{1,n}=\calC_{1,n},C_{2,n}\in\mathfrak{C}_{z^n}\right)\nonumber\\
    &\mathop\le\limits^{(b)}\sum_{u_0^n}\sum_{\calC_{1,n}}\bbP_\mu\pr{U_0^n=u_0^n}\bbP_\mu\pr{C_{1,n}=\calC_{1,n}|U_0^n=u_0^n}\bbP_\mu\left(\frac{1}{2^{nR_1}}\sum_{s_1\in\calS_1}\frac{W_{Z|U_0U_1}^\on\pr{z^n|U_0^n,U_1^n(s_1)}}{q_{Z|U_0}^\on\pr{z^n|U_0^n}}\right.\nonumber\\
    &\quad\left.\times \indic{1}_{\br{\pr{U_0^n,U_1^n(s_1),z^n}\in\calA_{1,\epsilon}^{(n)}}}>\frac{1+3\cdot2^{-n\beta}}{1+2^{-n\beta}}\Bigg|
    U_0^n=u_0^n,C_{1,n}=\calC_{1,n},C_{2,n}\in\mathfrak{C}_{z^n}\right)\nonumber\\
    &\mathop\le\limits^{(c)}\sum_{u_0^n}\bbP_\mu\pr{U_0^n=u_0^n}\bbP_\mu\left(\frac{1}{2^{nR_1}}\sum_{s_1\in\calS_1}\frac{W_{Z|U_0U_1}^\on\pr{z^n|U_0^n,U_1^n(s_1)}}{q_{Z|U_0}^\on\pr{z^n|U_0^n}}\indic{1}_{\br{\pr{U_0^n,U_1^n(s_1),z^n}\in\calA_{1,\epsilon}^{(n)}}}\right.\nonumber\\
    &\qquad\qquad>1+2^{-n\beta}|U_0^n=u_0^n\Big)\nonumber\\
    &\mathop\le\limits^{(d)}\sum_{u_0^n}\bbP_\mu\pr{U_0^n=u_0^n}\exp\pr{-\frac{1}{3}2^{-n\pr{\bbI(U_1;Z|U_0)+\epsilon+2\beta-R_1}}}\nonumber\\
    &\le\exp\pr{-\frac{1}{3}2^{-n\pr{\bbI(U_1;Z|U_0)+\epsilon+2\beta-R_1}}},\label{eq:Ptilde}
\end{align}
    \setcounter{equation}{57}
\end{figure*}
\begin{lemma}[Bound on Atypical Sets]
\label{lemma:Bounding_Atypical}
    Let $0<\delta$, $V_0\in\calV_0$, $V_1\in\calV_1$, $V_2\in\calV_2$, and $Z\in\calZ$ be four finite \acp{RV} with joint \ac{PMF} $p_{V_0V_1V_2Z}$. Also, let $V_0^n\in\calV_0^n$ be a random sequence generated \ac{iid} according to $P_{V_0}$, and for any realization $V_0^n=v_0^n$, let $B_{1,n}=\br{V_1^n(j_1)}_{j_1\in\calJ_1}$, where $\calJ_1\triangleq[2^{nR_1}]$, and $B_{2,n}=\br{V_2^n(j_2)}_{j_2\in\calJ_2}$, where $\calJ_2\triangleq[2^{nR_2}]$, be sets of codewords of length $n$ generated \ac{iid} according to $p_{V_1|V_0}(\cdot|v_0^n)$ and $p_{V_2|V_0}(\cdot|v_0^n)$, respectively. We define $B_n\triangleq\br{V_0^n,B_{1,n},B_{2,n}}$. Now, consider a set $\calA\subseteq\calV_0^n\times\calV_1^n\times\calV_2^n\times\calZ^n$, with $\bbP\pr{\pr{V_0^n,V_1^n,V_2^n,Z^n}\in\calA}\le\mu$. We have \eqref{eq:Additional_2} at the bottom of the previous page.
    \end{lemma}
    \begin{proof}
        We have \eqref{eq:Additional_3} at the bottom of the previous page, 
        where $(a)$ follows from Lemma~\ref{lemma:Janson}, noting that $(i)$ the innermost sum is nonnegative and bounded by $1$; $(ii)$ the summation over $\pr{j_1,j_2}$ contains $2^{n(R_1 + R_2)}$ terms, which can be partitioned into $2^{n\max\br{R_1, R_2}}$ subsets, each consisting of $2^{n\min\br{R_1, R_2}}$ independently distributed terms; and $(iii)$ the overall expectation of the two inner summations equals $2^{n(R_1 + R_2)} \mathbb{P}\!\left((V_0^n,V_1^n, V_2^n, Z^n) \in \calA|V_0^n=v_0^n\right)$.
    \end{proof}
We now proceed to the proof of Lemma~\ref{lemma:Resolvability_12}. To bound $p_{\text{atyp}}^{(1)}$ in \eqref{eq:Total_TV_Bound}, for any $\alpha>1$, we have
\begin{align*}
    &\bbP_{U_0^n,U_1^n,Z^n}\pr{\pr{U_0^n,U_1^n,Z^n}\notin\calA_{1,\epsilon}^{(n)}}\nonumber\\
    &=\bbP_{U_0^n,U_1^n,Z^n}\Bigg(\frac{W_{Z|U_0U_1}^\on\pr{Z^n|U_0^n,U_1^n(s_1)}}{W_{Z|U_0}^\on\pr{Z^n|U_0^n}}\nonumber\\
    &\qquad>2^{n(\bbI(U_1;Z|U_0)+\epsilon)}\Bigg)\nonumber\\
    &=\bbP_{U_0^n,U_1^n,Z^n}\Bigg(\pr{\frac{W_{Z|U_0U_1}^\on\pr{Z^n|U_0^n,U_1^n(s_1)}}{W_{Z|U_0}^\on\pr{Z^n|U_0^n}}}^{\alpha-1}\nonumber\\
    &\qquad>2^{n(\alpha-1)(\bbI(U_1;Z|U_0)+\epsilon)}\Bigg)\nonumber\\
    &\mathop\le\limits^{(a)}2^{-n(\alpha-1)(\bbI(U_1;Z|U_0)+\epsilon)}\nonumber\\
    &\qquad\times\bbE_{U_0^n,U_1^n,Z^n}\pr{\pr{\frac{W_{Z|U_0U_1}^\on\pr{Z^n|U_0^n,U_1^n(s_1)}}{W_{Z|U_0}^\on\pr{Z^n|U_0^n}}}^{\alpha-1}}\nonumber\\
    &\mathop=\limits^{(b)}2^{-n(\alpha-1)\pr{\bbI(U_1;Z|U_0)+\epsilon-\bbD_\alpha\pr{P_{U_0U_1Z}\rVert P_{U_0}P_{U_1|U_0}P_{Z|U_0}}}}\nonumber\\
    &\mathop\le\limits^{(c)}2^{-n\beta_1},
\end{align*}where
\begin{itemize}
    \item[$(a)$] follows from Markov's inequality;
    \item[$(b)$] follows from the definition of  R\'{e}nyi divergence of order $\alpha$;
    \item[$(c)$] \sloppy follows by choosing $\beta_1<(\alpha-~1)\big(\bbI(U_1;Z|U_0)+\epsilon-\bbD_\alpha\pr{P_{U_0U_1Z}\rVert P_{U_0}P_{U_1|U_0}P_{Z|U_0}}\big)$.
\end{itemize}
Note that since any joint distribution is absolutely continuous with respect to the product of its marginals, the R\'{e}nyi divergence is finite. Moreover, it is continuous in $\alpha$, and as $\alpha \to 1$, the R\'{e}nyi divergence converges to the mutual information and
\begin{align*}
    &\bbD_\alpha\!\left(P_{U_0U_1Z}\lVert P_{U_0} P_{U_1|U_0} P_{Z|U_0}\right)\nonumber\\
    &\ge \bbD_1\!\left(P_{U_0U_1Z}\lVert P_{U_0} P_{U_1|U_0} P_{Z|U_0}\right)\nonumber\\
    &=\bbI(U_1; Z | U_0).
\end{align*}
Therefore, by selecting $\epsilon > 0$ such that
\begin{align*}
\epsilon > \bbD_\alpha\!\left(P_{U_0U_1Z}\lVert P_{U_0} P_{U_1|U_0} P_{Z|U_0}\right) - \bbI(U_1; Z | U_0),
\end{align*}we guarantee the existence of an $\alpha>1$ for which $\beta_1>0$ while keeping $\beta_1$ sufficiently small to satisfy the desired condition. Similarly, one can show that
\begin{align}
    &\bbP_{U_0^n,U_1^n,U_2^n,Z^n}\pr{\pr{U_0^n,U_1^n,U_2^n,Z^n}\notin\calA_{2,\epsilon}^{(n)}}\le2^{-n\beta_2},\label{eq:Beta_2}
\end{align}\sloppy where $\beta_2<(\alpha-1)\big(\bbI(U_2;Z|U_0,U_1)+\epsilon-\bbD_\alpha\pr{P_{U_0U_1U_2Z}\rVert P_{U_0U_1}P_{U_2|U_0U_1}P_{Z|U_0U_1}}\big)$. Now we use an argument similar to that we used for $\beta_1$ to show that there exists an $\alpha>1$ such that $\beta_2>0$ and is sufficiently small. Then we choose $\beta\triangleq\min\br{\beta_1,\beta_2}$. 
Now applying Lemma~\ref{lemma:Bounding_Atypical} first with $A=\calU_2^n\times\pr{\pr{\calU_0^n\times\calU_1^n\times\calZ^n}\backslash\calA_{1,\epsilon}^{(n)}}$ and $\delta=1$, and then with $A=\pr{\pr{\calU_0^n\times\calU_1^n\times\calU_2^n\times\calZ^n}\backslash\calA_{2,\epsilon}^{(n)}}$ and $\delta=1$, respectively, leads to
\begin{subequations}\label{eq:p_Atyp_final}
    \begin{align}
    \bbP_\mu\pr{p_{\text{atyp}}^{(1)}>2\cdot2^{-n\beta}}&\le\exp\pr{-2\cdot2^{n\pr{\min\br{R_1,R_2}-2\beta}}},\label{eq:p_Atyp_1_final}\\
    \bbP_\mu\pr{p_{\text{atyp}}^{(2)}>2\cdot2^{-n\beta}}&\le\exp\pr{-2\cdot2^{n\pr{\min\br{R_1,R_2}-2\beta}}}.\label{eq:p_Atyp_2_final}
    \end{align}
\end{subequations}
To bound the probability of the typical term in \eqref{eq:Total_TV_Bound}, for any fixed $u_0^n$, $u_1^n$, and $z^n$, we first apply Lemma~\ref{lemma:Bounding_Typical} with $V_0=U_0$, $V_1=U_2$, $V_2=U_1$, and $\delta=2^{-n\beta}$, which leads to
\begin{align}
    \bbP_\mu\pr{p_{\text{typ}}^{(1)}(u_0^n,u_1^n,z^n)>1+2^{-n\beta}}\nonumber\\
    \le\exp\pr{-\frac{1}{3}2^{-n\pr{\bbI\pr{U_2;Z|U_0,U_1}+\epsilon+2\beta-R_2}}},\label{eq:pTyp1}
\end{align}where
\begin{align*}
    &p_{\text{typ}}^{(1)}(u_0^n,u_1^n,z^n)\triangleq\nonumber\\
    &\frac{1}{2^{nR_2}}\sum_{s_2\in\calS_2}\frac{W_{Z|U_0U_1U_2}^\on\pr{z^n|u_0^n,u_1^n,U_2^n(s_2)}}{W_{Z|U_0U_1}^\on\pr{z^n|u_0^n,u_1^n}}\nonumber\\
    &\quad\times\indic{1}_{\br{\pr{u_0^n,u_1^n,U_2^n(s_2),z^n}\in\calA_{2,\epsilon}^{(n)}}}.
\end{align*}Now let,
\begin{align}
    \mathfrak{C}_{z^n}\triangleq\hspace{-1mm}\bigcap_{(u_0^n,u_1^n)\in\calU_0^n\times\calU_1^n}\br{\calC_{2,n}: p_{\text{typ}}^{(1)}(u_0^n,u_1^n,z^n)\le1+2^{-n\beta}}.\label{eq:CZ}
\end{align}For any arbitrary but fixed $z^n$, we have \eqref{eq:Ptilde} at the bottom of the page, 
where
\begin{itemize}
    \item[$(a)$] follows from law of total probability;
    \item[$(b)$] follows since we have $C_{2,n}\in\mathfrak{C}_{z^n}$ in the conditioning;
    \item[$(c)$] follows since for sufficiently large $n$ we have $2^{-n\beta}\le1$;
    \item[$(d)$] follows from Lemma~\ref{lemma:Bounding_Typical} with $V_0=U_0$, $V_1=U_1$, $V_2=\emptyset$, and $\delta=2^{-n\beta}$.
\end{itemize}
\begin{figure*}[b!]
    \hrulefill
    \setcounter{equation}{67}
\begin{align}
    \bbE_\mu\bbV\pr{p_{Z^n\lvert\calC_n},p_{Z^n\lvert\calC_2}}&=\bbE_\mu \bbV\left(\frac{1}{2^{n(R_1+R_2)}}\sum_{s_1,s_2}W_{Z\lvert U_0U_1U_2}^\on\pr{z^n\lvert U_0^n,U_1^n(s_1),U_2^n(s_2)},\frac{1}{2^{nR_2}}\sum_{s_2}W_{Z\lvert U_0U_2}^\on\pr{z^n\lvert U_0^n,U_2^n(s_2)}\right)\nonumber\\
    &\mathop\le\limits^{(a)}\frac{1}{2^{nR_2}}\sum_{s_2}\bbE_\mu \bbV\left(\frac{1}{2^{nR_1}}\sum_{s_1}W_{Z\lvert U_0U_1U_2}^\on\pr{z^n\lvert U_0^n,U_1^n(s_1),U_2^n(s_2)},W_{Z\lvert U_0U_2}^\on\pr{z^n\lvert U_0^n,U_2^n(s_2)}\right)\nonumber\\
    &\mathop=\limits^{(b)} \bbE_{U_0^n,U_2^n(1),C_{1,n}}\bbV\left(\frac{1}{2^{nR_1}}\sum_{s_1}W_{Z\lvert U_0U_1U_2}^\on\pr{z^n\lvert U_0^n,U_1^n(s_1),U_2^n(1)},W_{Z\lvert U_0U_2}^\on\pr{z^n\lvert U_0^n,U_2^n(1)}\right)\nonumber\\
    &\mathop=\limits^{(c)} \bbE_{U_0^n,U_2^n,C_{1,n}}\bbV\left(\bar{p}_{Z^n|C_{1,n}}(z^n),W_{Z\lvert U_0U_2}^\on\pr{z^n\lvert U_0^n,U_2^n}\right),\label{eq:Total_Vari_One_Sided}
\end{align}
    \setcounter{equation}{62}
\end{figure*}
Now we have,
\begin{align*}
    &\bbP_\mu\pr{\bbV\pr{p_{Z^n|\calC_n},q_{Z|U_0}^\on(\cdot|U_0^n)}>7\cdot2^{-n\beta}}\nonumber\\
    &\mathop\le\limits^{(a)}\bbP_\mu\pr{p_{\text{atyp}}^{(1)}>2\cdot2^{-n\beta}}+\bbP_\mu\pr{p_{\text{atyp}}^{(2)}>2\cdot2^{-n\beta}}\nonumber\\
    &\quad+\sum_{z^n\in\calZ^n}\bbP_\mu\pr{p_{\text{typ}}(z^n)>1+3\cdot2^{-n\beta}}\nonumber\\
    &=\bbP_\mu\pr{p_{\text{atyp}}^{(1)}>2\cdot2^{-n\beta}}+\bbP_\mu\pr{p_{\text{atyp}}^{(2)}>2\cdot2^{-n\beta}}\nonumber\\
    &\quad+\sum_{z^n\in\calZ^n}\Big(\bbP_\mu\pr{C_{2,n}\in\mathfrak{C}_{z^n},p_{\text{typ}}(z^n)>1+3\cdot2^{-n\beta}}\nonumber\\
    &\quad+\bbP_\mu\pr{C_{2,n}\notin\mathfrak{C}_{z^n},p_{\text{typ}}(z^n)>1+3\cdot2^{-n\beta}}\Big)\nonumber\\
    &\le\bbP_\mu\pr{p_{\text{atyp}}^{(1)}>2\cdot2^{-n\beta}}+\bbP_\mu\pr{p_{\text{atyp}}^{(2)}>2\cdot2^{-n\beta}}\nonumber\\
    &\quad+\sum_{z^n\in\calZ^n}\Big(\bbP_{C_{2,n}}\pr{C_{2,n}\in\mathfrak{C}_{z^n}}\nonumber\\
    &\quad+\bbP_\mu\pr{p_{\text{typ}}(z^n)>1+3\cdot2^{-n\beta}\big|C_{2,n}\in\mathfrak{C}_{z^n}}\Big)\nonumber\\
    &\mathop\le\limits^{(b)}2\exp\pr{-2\times2^{n\pr{\min\br{R_1,R_2}-2\beta}}}\nonumber\\
    &\quad+\abs{\calZ^n}\abs{\calU_0^n}\abs{\calU_1^n}\exp\pr{-\frac{1}{3}2^{-n\pr{\bbI(U_2;Z|U_0,U_1)+\epsilon+2\beta-R_2}}}\nonumber\\
    &\quad+\abs{\calZ^n}\exp\pr{-\frac{1}{3}2^{-n\pr{\bbI(U_1;Z|U_0)+\epsilon+2\beta-R_1}}},
\end{align*}where $(a)$ follows \eqref{eq:Total_TV_Bound} and the union bound; and $(b)$ follows from \eqref{eq:p_Atyp_final}, \eqref{eq:pTyp1}, and \eqref{eq:Ptilde}.

Now we need to choose $\gamma_1$ and $\gamma_2$ such that \eqref{eq:Resolvability_122} holds. First, we choose $\epsilon$ and $\beta$ small enough such that the following quantity is positive,
\begin{align*}
    \hat{\delta}&\triangleq\min\big\{\min\br{R_1,R_2}-2\beta,R_1-2\beta-\epsilon\nonumber\\
    &\quad-\bbI(U_1;Z|U_0),R_2-2\beta-\epsilon-\bbI(U_2;Z|U_0,U_1)\big\}.
\end{align*}
Since there is no constraint on $\beta$ and $\epsilon$ other than being positive and sufficiently small, such a choice is possible if $R_1>\bbI(U_1;Z|U_0)$ and $R_2>\bbI(U_2;Z|U_0,U_1)$. Then, choosing $\gamma_2<\hat{\delta}$ and $\gamma_1<\beta$ completes the proof of Lemma~\ref{lemma:Resolvability_12}. 
\end{proof}

Similarly, by interchanging the roles of $U_1$ and $U_2$ in Lemma~\ref{lemma:Resolvability_12}, and using the time-sharing between these two corner points, we can establish Lemma~\ref{lemma:Resolvability_1}. However, to obtain the entire achievable region stated in Lemma~\ref{lemma:Resolvability_1}, a time-sharing argument between these two corner points is required, as detailed below. We need to show that for any rate pair $(R_1,R_2)$ satisfying \eqref{eq:Resolvability_11}, there exists a parameter $\lambda^\star\in[0,1]$ such that
\begin{align*}
    R_1&>(1-\lambda^\star)\bbI(U_1;Z|U_0)+\lambda^\star\bbI(U_1;Z|U_0,U_2),\\
    R_2&>\lambda^\star\bbI(U_2;Z|U_0)+(1-\lambda^\star)\bbI(U_2;Z|U_0,U_1).
\end{align*}Then, Lemma~\ref{lemma:Resolvability_1} can be established by applying the result twice: first, over a block of length $(1-\lambda^\star)n$, and then, with the roles of $U_1$ and $U_2$ interchanged, over a block of length $\lambda^\star n$. First, we define $R_1(\lambda):[0,1]\to\bbR$ and $R_2(\lambda):[0,1]\to\bbR$ as
\begin{align*}
    R_1(\lambda)&=(1-\lambda)\bbI(U_1;Z|U_0)+\lambda\bbI(U_1;Z|U_0,U_2),\\
    R_2(\lambda)&=\lambda\bbI(U_2;Z|U_0)+(1-\lambda)\bbI(U_2;Z|U_0,U_1).
\end{align*}Note that $R_1(\lambda)$ and $R_2(\lambda)$ are continuous functions of $\lambda$. Moreover, since $R_1(0)=\bbI(U_1;Z|U_0)$, $R_1(1)=\bbI(U_1;Z|U_0,U_2)$, $R_2(0)=\bbI(U_2;Z|U_0,U_1)$, and $R_2(1)=\bbI(U_2;Z|U_0)$, and given that $\bbI(U_1;Z|U_0,U_2)\ge\bbI(U_1;Z|U_0)$ and $\bbI(U_2;Z|U_0,U_1)\ge\bbI(U_2;Z|U_0)$, it follows that $R_1(\lambda)$ is increasing and $R_2(\lambda)$ is decreasing in $\lambda$. If either $R_1$ is already greater than $\bbI(U_1;Z|U_0,U_2)$ or $R_2$ is greater than $\bbI(U_2;Z|U_0,U_1)$, then our rate pair $(R_1,R_2)$ lies in one of the cases that Lemma~\ref{lemma:Resolvability_12} (or its symmetric version) already covers. Therefore, the claim is immediately true, and we do not need to construct a time-sharing parameter $\lambda^\star$; the direct theorem applies as it is. Since Lemma~\ref{lemma:Resolvability_12} already covers the cases where 
$R_1 > \bbI(U_1;Z|U_0,U_2)$ or $R_2 > \bbI(U_2;Z|U_0,U_1)$, we now focus on the remaining situation where 
\begin{align*}
    \bbI(U_1;Z|U_0) &\le R_1 \le \bbI(U_1;Z|U_0,U_2)\\
\bbI(U_2;Z|U_0) &\le R_2 \le \bbI(U_2;Z|U_0,U_1).
\end{align*}
In this case, the functions
\begin{align*}
R_1(\lambda) &= (1-\lambda )\bbI(U_1;Z|U_0) + \lambda\bbI(U_1;Z|U_0,U_2),\\
R_2(\lambda) &= \lambda \bbI(U_2;Z|U_0) + (1-\lambda) \bbI(U_2;Z|U_0,U_1)
\end{align*}
are \emph{strictly monotonic} in $\lambda$: 
$R_1(\lambda)$ increases continuously from $\bbI(U_1;Z|U_0)$ to $\bbI(U_1;Z|U_0,U_2)$, 
while $R_2(\lambda)$ decreases continuously from $\bbI(U_2;Z|U_0,U_1)$ to $\bbI(U_2;Z|U_0)$. 
By the intermediate value theorem, there exist unique parameters 
$\lambda_1, \lambda_2 \in [0,1]$ such that $R_1(\lambda_1) = R_1$ and $R_2(\lambda_2) = R_2$. 
Next, we observe that $\lambda_1 > \lambda_2$. 
Indeed, suppose for contradiction that $\lambda_1 \le \lambda_2$. 
Since $R_1(\lambda)$ is strictly increasing and $R_2(\lambda)$ is strictly decreasing, we have that,
\begin{align*}
R_1 + R_2 
&= R_1(\lambda_1) + R_2(\lambda_2) \\
&\le R_1(\lambda_1) + R_2(\lambda_1) \\
&= \bbI(U_1,U_2;Z|U_0),
\end{align*}
which contradicts the assumption $R_1 + R_2 > \bbI(U_1,U_2;Z|U_0)$.
Therefore, we must have $\lambda_1 > \lambda_2$. 

By continuity, we can choose a parameter $\lambda^\star$ satisfying 
$\lambda_2 < \lambda^\star < \lambda_1$. 
Because $R_1(\lambda)$ increases with $\lambda$ and $R_2(\lambda)$ decreases with $\lambda$, 
it then follows that $R_1 > R_1(\lambda^\star)$ and $R_2 > R_2(\lambda^\star)$, as required. 
    \begin{figure*}[b!]
    \hrulefill
    \setcounter{equation}{71}
    \begin{align}
        &\bbP_{B_n}\left(\frac{1}{2^{nR}}\sum_{j\in\calJ}\sum_{z^n\in\calZ^n}p_{Z|V_0V_1V_2}^\on\pr{z^n|V_0^n,V_1^n(j),V_2^n}\indic{1}_{\br{\pr{V_0^n,V_1^n(j),V_2^n,z^n}\in\calA}}>\mu(1+\delta)\right)\le\exp\br{-2\delta^2\mu^22^{nR}}.\label{eq:Additinal_4}
    \end{align}
\hrulefill
\begin{align}
            &\bbP_{B_n}\left(\frac{1}{2^{nR}}\sum_{j\in\calJ}\sum_{z^n\in\calZ^n}p_{Z|V_0V_1V_2}^\on\pr{z^n|V_0^n,V_1^n(j),V_2^n}\indic{1}_{\br{\pr{V_0^n,V_1^n(j),V_2^n,z^n}\in\calA}}>\mu(1+\delta)\right)\nonumber\\
             &=\bbP_{B_n}\left(\sum_{j\in\calJ}\sum_{z^n\in\calZ^n}p_{Z|V_0V_1V_2}^\on\pr{z^n|V_0^n,V_1^n(j),V_2^n}\indic{1}_{\br{\pr{V_0^n,V_1^n(j),V_2^n,z^n}\in\calA}}>2^{nR}\mu(1+\delta)\right)\nonumber\\
             &\le\bbP_{B_n}\left(\sum_{j\in\calJ}\sum_{z^n\in\calZ^n}p_{Z|V_0V_1V_2}^\on\pr{z^n|V_0^n,V_1^n(j),V_2^n}\indic{1}_{\br{\pr{V_0^n,V_1^n(j),V_2^n,z^n}\in\calA}}>2^{nR}\pr{\bbP\pr{\pr{V_0^n,V_1^n,V_2^n,Z^n}\in\calA}+\mu\delta}\right)\nonumber\\
             &=\sum_{v_0^n}\sum_{v_2^n}\bbP(V_0^n=v_0^n,V_2^n=v_2^n)\bbP_{B_n|V_0^n=v_0^n,V_2^n=v_2^n}\left(\sum_{j\in\calJ}\sum_{z^n\in\calZ^n}p_{Z|V_0V_1V_2}^\on\pr{z^n|V_0^n,V_1^n(j),V_2^n}\right.\nonumber\\
             &\qquad\qquad\times\indic{1}_{\br{\pr{V_0^n,V_1^n(j),V_2^n,z^n}\in\calA}}>2^{nR}\pr{\bbP\pr{\pr{V_0^n,V_1^n,V_2^n,Z^n}\in\calA}+\mu\delta}\big|V_0^n=v_0^n,V_2^n=v_2^n\Big)\nonumber\\
             &\mathop\le\limits^{(a)}\sum_{v_0^n}\sum_{v_2^n}\bbP(V_0^n=v_0^n,V_2^n=v_2^n)\exp\pr{-2\frac{2^{2nR}\mu^2\delta^2}{2^{nR}}}\nonumber\\
             &\le\exp\pr{-2\mu^2\delta^22^{nR}},\label{eq:Additinal_5}
        \end{align}
            \setcounter{equation}{62}
\end{figure*}
\section{Proof of Lemma~\ref{lemma:one_Sided_MAC_Res}}
\label{proof:lemma:one_Sided_MAC_Res}
We prove Lemma~\ref{lemma:one_Sided_MAC_Res} by considering two cases. First, when $R_2 > \bbI(U_2;Z|U_0)$, the triangle inequality yields
\begin{align}
    \bbV\pr{p_{Z^n\lvert\bbC_n},p_{Z^n\lvert\bbC_2}}&\le\bbV\pr{p_{Z^n\lvert\bbC_n},q_{Z|U_0}^\on(\cdot|u_0^n)}\nonumber\\
    &\quad+\bbV\pr{p_{Z^n\lvert\bbC_2},q_{Z|U_0}^\on(\cdot|u_0^n)}.\label{eq:Triangle}
\end{align}
From Lemma~\ref{lemma:Resolvability_1}, it follows that for sufficiently large $n$, there exist $\gamma'_1,\gamma'_2>0$ such that
\begin{align}
    \bbP_\mu\pr{\bbV\pr{p_{Z^n|\bbC_n},q_{Z|U_0}^\on(\cdot|u_0^n)}>\exp\pr{-\gamma'_1n}}\nonumber\\
    \le\exp\pr{-\exp(\gamma'_2n)},\label{eq:First_Case_First_Term}
\end{align}if
\begin{subequations}\label{eq:Conditions_1}
\begin{align}
R_1 &> \bbI(U_1;Z\lvert U_0), \\ 
 R_2 &> \bbI(U_2;Z\lvert U_0), \label{eq:Conditions_10}\\ 
 R_1 + R_2 &> \bbI(U_1,U_2;Z\lvert U_0).
\end{align}
\end{subequations}
Moreover, since $p_{Z^n\lvert\bbC_2}$ corresponds to the scenario in which the first transmitter sends an \ac{iid} sequence and the second transmitter transmits a codeword from $\bbC_2$, it follows from Lemma~\ref{lemma:Resolvability_1} that, for sufficiently large $n$, there exist $\gamma''_1,\gamma''_2>0$ such that
\begin{align}
    \bbP_\mu\pr{\bbV\pr{p_{Z^n|\bbC_2},q_{Z|U_0}^\on(\cdot|u_0^n)}>\exp\pr{-\gamma''_1n}}\nonumber\\
    \le\exp\pr{-\exp(\gamma''_2n)},\label{eq:First_Case_Second_Term}
\end{align}
if \eqref{eq:Conditions_10} holds. 
Now, we have
\begin{align}
    &\bbP_\mu\pr{\bbV\pr{p_{Z^n|\bbC_n},p_{Z^n\lvert\bbC_2}}>\exp\pr{-\gamma_1n}}\nonumber\\
    &\mathop\le\limits^{(a)}\bbP_\mu\Bigg(\bbV\pr{p_{Z^n\lvert\bbC_n},q_{Z|U_0}^\on(\cdot|u_0^n)}\nonumber\\
    &\quad+\bbV\pr{p_{Z^n\lvert\bbC_2},q_{Z|U_0}^\on(\cdot|u_0^n)}>\exp\pr{-\gamma_1n}\Bigg)\nonumber\\
    &\mathop\le\limits^{(b)}\bbP_\mu\pr{\bbV\pr{p_{Z^n\lvert\bbC_n},q_{Z|U_0}^\on(\cdot|u_0^n)}>\frac{\exp\pr{-\gamma_1n}}{2}}\nonumber\\
    &\quad+\bbP_\mu\pr{\bbV\pr{p_{Z^n\lvert\bbC_2},q_{Z|U_0}^\on(\cdot|u_0^n)}>\frac{\exp\pr{-\gamma_1n}}{2}},\nonumber\\
    &\mathop\le\limits^{(c)}\exp\pr{-\exp(\gamma_2n)},
\end{align}where
\begin{itemize}
    \item[$(a)$] follows from \eqref{eq:Triangle}; 
    \item[$(b)$] follows from the union bound and the fact that, for non-negative \acp{RV} $A_1$ and $A_2$ and a constant $c$, we have $\br{A_1 + A_2 > c} \subseteq \br{A_1 > \tfrac{c}{2}} \cup \br{A_2 > \tfrac{c}{2}}$, together with the choice $\gamma_1 < \min\{\gamma'_1,\gamma''_1\}$, which ensures that for sufficiently large $n$,
    \begin{align*}
        \frac{\exp\pr{-\gamma_1n}}{2}\ge\exp\pr{-\gamma'_1n},\exp\pr{-\gamma''_1n};
    \end{align*}
    \item[$(c)$] follows from \eqref{eq:First_Case_First_Term} and \eqref{eq:First_Case_Second_Term} by choosing $\gamma_2=\min\br{\gamma'_2,\gamma''_2}$.
\end{itemize}
This results in the region $\calR_1$ of Lemma~\ref{lemma:one_Sided_MAC_Res}.

To establish the region $\calR_2$ in Lemma~\ref{lemma:one_Sided_MAC_Res}, consider the case where $R_2\le\bbI(U_2;Z|U_0)$. We have, \eqref{eq:Total_Vari_One_Sided} at the bottom of the previous page, 
where
\begin{itemize}
    \item[$(a)$] follows from the triangle inequality;
    \item[$(b)$] follows from the symmetry of codebook construction with respect to $s_2$;
    \item[$(c)$] follows by defining
    \setcounter{equation}{68}
\begin{align}
    \bar{p}_{Z^n|C_{1,n}}\triangleq\frac{1}{2^{nR_1}}\sum_{s_1}W_{Z\lvert U_0U_1U_2}^\on\pr{z^n\lvert U_0^n,U_1^n(s_1),U_2^n}.\label{eq:Induced_Dist_One_Sided}
\end{align}
\end{itemize}
\begin{figure*}[b!]
    \hrulefill
    \setcounter{equation}{73}
\begin{align}
    &\bbP_{U_0^n,U_1^n,U_2^n,Z^n}\pr{\pr{U_0^n,U_1^n,U_2^n,Z^n}\notin\calA_\epsilon^{(n)}}\nonumber\\
    &\quad=\bbP_{U_0^n,U_1^n,U_2^n,Z^n}\pr{\frac{W_{Z|U_0U_1U_2}^\on\pr{Z^n|U_0^n,U_1^n(s_1),U_2^n}}{W_{Z|U_0U_2}^\on\pr{Z^n|U_0^n,U_2^n}}>2^{n(\bbI(U_1;Z|U_0,U_2)+\epsilon)}}\nonumber\\
    &\quad=\bbP_{U_0^n,U_1^n,U_2^n,Z^n}\pr{\pr{\frac{W_{Z|U_0U_1U_2}^\on\pr{Z^n|U_0^n,U_1^n(s_1),U_2^n}}{W_{Z|U_0U_2}^\on\pr{Z^n|U_0^n,U_2^n}}}^{\alpha-1}>2^{n(\alpha-1)(\bbI(U_1;Z|U_0,U_2)+\epsilon)}}\nonumber\\
    &\quad\mathop\le\limits^{(a)}2^{-n(\alpha-1)(\bbI(U_1;Z|U_0,U_2)+\epsilon)}\bbE_{U_0^n,U_1^n,U_2^n,Z^n}\pr{\pr{\frac{W_{Z|U_0U_1U_2}^\on\pr{Z^n|U_0^n,U_1^n(s_1),U_2^n}}{W_{Z|U_0U_2}^\on\pr{Z^n|U_0^n,U_2^n}}}^{\alpha-1}}\nonumber\\
    &\quad\mathop=\limits^{(b)}2^{-n(\alpha-1)\pr{\bbI(U_1;Z|U_0,U_2)+\epsilon-\bbD_\alpha\!\pr{P_{U_0U_1U_2Z}\rVert P_{U_0U_2}P_{U_1|U_0U_2}P_{Z|U_0U_2}}}}\nonumber\\
    &\quad\mathop\le\limits^{(c)}2^{-n\beta},\label{eq:P_atyp_1_One_Sided}
\end{align}
    \setcounter{equation}{69}
\end{figure*}
Now let $\epsilon>0$ be arbitrary, to be defined later, and define
\begin{align*}
    \calA_\epsilon^{(n)}&\triangleq\Bigg\{\pr{u_0^n,u_1^n,u_2^n,z^n}:\frac{1}{n}i_{p_{U_1,Z|U_0,U_2}}\pr{u_1^n;z^n|u_0^n,u_2^n}\nonumber\\
    &\qquad<\bbI(U_1;Z|U_0,U_2)\Bigg\}.
\end{align*}Now, by using the indicator function, we split $\bar{p}_{Z^n|\bbC_1}$ in \eqref{eq:Induced_Dist_One_Sided}
into atypical and typical parts as follows,
\begin{subequations}\label{eq:Atyp_Typ_One_Sided}
\begin{align}
    \bar{p}_{\text{atyp}}&\triangleq\sum_{z^n\in\calZ^n}\frac{1}{2^{nR_1}}\sum_{s_1}W_{Z\lvert U_0U_1U_2}^\on\pr{z^n\lvert U_0^n,U_1^n(s_1),U_2^n}\nonumber\\
    &\quad\times\indic{1}_{\br{\pr{U_0^n,U_1^n(s_1),U_2^n,z^n}\notin\calA_\epsilon^{(n)}}},\label{eq:Atyp_One_Sided}\\
    \bar{p}_{\text{typ}}&\triangleq\frac{1}{2^{nR_1}}\sum_{s_1}\frac{W_{Z\lvert U_0U_1U_2}^\on\pr{z^n\lvert U_0^n,U_1^n(s_1),U_2^n}}{W_{Z\lvert U_0U_2}^\on\pr{z^n\lvert U_0^n,U_2^n}}\nonumber\\
    &\quad\times\indic{1}_{\br{\pr{U_0^n,U_1^n(s_1),U_2^n,z^n}\in\calA_\epsilon^{(n)}}}.\label{eq:Typ_One_Sided}
\end{align}
\end{subequations}Using \eqref{eq:Atyp_Typ_One_Sided} and the triangle inequality, we split the total variation distance on the \ac{RHS} of \eqref{eq:Total_Vari_One_Sided} into atypical and typical parts as follows,
\begin{align}
    &\bbV\pr{\bar{p}_{Z^n|C_{1,n}},W_{Z\lvert U_0U_2}^\on\pr{\cdot|U_0^n,U_2^n}}\nonumber\\
    &=\sum_{z^n\in\calZ^n}W_{Z\lvert U_0U_2}^\on\pr{z^n|U_0^n,U_2^n}\left|\frac{\bar{p}_{Z^n|\bbC_1}(z^n)}{W_{Z\lvert U_0U_2}^\on\pr{z^n|U_0^n,U_2^n}}-1\right|\nonumber\\
    &\le \bar{p}_{\text{atyp}}+\sum_{z^n\in\calZ^n}W_{Z\lvert U_0U_2}^\on\pr{z^n|U_0^n,U_2^n}\sbr{\bar{p}_{\text{typ}}(z^n)-1}^+.\label{eq:Atyp_Typ_Trian_One_Sided}
\end{align}Note that we assume $\bar{p}_{Z^n|C_{1,n}}$ is absolutely continuous \ac{wrt} $W_{Z\lvert U_0U_2}^\on\pr{\cdot|U_0^n,U_2^n}$, therefore the summation in \eqref{eq:Atyp_Typ_Trian_One_Sided} is bounded. The following lemma plays a key role in our proof.
\begin{lemma}[Bound on Atypical Sets]
\label{lemma:Bounding_Atypical_One_Sided}
    Let $0<\delta$, $V_0\in\calV_0$, $V_1\in\calV_1$, $V_2\in\calV_2$, and $Z\in\calZ$ be four finite \acp{RV} with joint \ac{PMF} $p_{V_0V_1V_2Z}$. Also, let $V_0^n\in\calV_0^n$ be a random sequence generated \ac{iid} according to $P_{V_0}$, and for any realization $V_0^n=v_0^n$, let $B_{1,n}=\br{V_1^n(j)}_{j\in\calJ}$, where $\calJ\triangleq[2^{nR}]$, be a set of codewords of length $n$ generated \ac{iid} according to $p_{V_1|V_0}(\cdot|v_0^n)$, and $V_2^n\in\calV_2^n$ be a random sequence generated \ac{iid} according to $P_{V_2|V_0}(\cdot|v_0^n)$. We define $B_n\triangleq\br{V_0^n,B_{1,n},V_2^n}$. Now, consider a set $\calA\subseteq\calV_0^n\times\calV_1^n\times\calV_2^n\times\calZ^n$, with $\bbP\pr{\pr{V_0^n,V_1^n,V_2^n,Z^n}\in\calA}\le\mu$. Then, we have \eqref{eq:Additinal_4} at the bottom of the previous page.
    \end{lemma}
    \begin{proof}
        We have \eqref{eq:Additinal_5} at the bottom of the previous page, 
        where $(a)$ follows from Lemma~\ref{lemma:Janson}, by noting that: $(i)$ the innermost sum is nonnegative and bounded above by $1$; $(ii)$ the summation over $j$ contains $2^{nR}$ independent terms, which implies $\chi=1$ in Lemma~\ref{lemma:Janson}; and $(iii)$ the overall expectation of the two inner summations equals $2^{nR}\mathbb{P}\!\left((V_0^n,V_1^n,V_2^n,Z^n)\in\calA|V_0^n=v_0^n,V_2^n=v_2^n\right)$.
    \end{proof}
    \setcounter{equation}{74}
We now proceed to the proof of Lemma~\ref{lemma:one_Sided_MAC_Res}. To bound $\bar{p}_{\text{atyp}}$ in \eqref{eq:Atyp_Typ_Trian_One_Sided}, for any $\alpha>1$, we have \eqref{eq:P_atyp_1_One_Sided} at the bottom of the page, 
where
\begin{itemize}
    \item[$(a)$] follows from Markov's inequality;
    \item[$(b)$] follows from the definition of  R\'{e}nyi divergence of order $\alpha$;
    \item[$(c)$] follows by choosing $\beta<(\alpha-1)\big(\bbI(U_1;Z|U_0,U_2)+\epsilon-\bbD_\alpha\!\pr{P_{U_0U_1U_2Z}\rVert P_{U_0U_2}P_{U_1|U_0U_2}P_{Z|U_0U_2}}\big)$.
\end{itemize}
Note that since any joint distribution is absolutely continuous with respect to the product of its marginals, the R\'{e}nyi divergence is finite. Moreover, it is continuous in $\alpha$, and as $\alpha \to 1$, the R\'{e}nyi divergence converges to the mutual information and
\begin{align*}
    &\bbD_\alpha\!\pr{P_{U_0U_1U_2Z}\rVert P_{U_0U_2}P_{U_1|U_0U_2}P_{Z|U_0U_2}}\nonumber\\
    &\ge \bbD_1\!\pr{P_{U_0U_1U_2Z}\rVert P_{U_0U_2}P_{U_1|U_0U_2}P_{Z|U_0U_2}}\nonumber\\
    &=\bbI(U_1; Z | U_0,U_2).
\end{align*}
Therefore, by selecting $\epsilon > 0$ such that
\begin{align*}
\epsilon &> \bbD_\alpha\!\pr{P_{U_0U_1U_2Z}\rVert P_{U_0U_2}P_{U_1|U_0U_2}P_{Z|U_0U_2}} \nonumber\\
&\quad- \bbI(U_1; Z|U_0,U_2),
\end{align*}we guarantee the existence of an $\alpha>1$ for which $\beta>0$ while keeping $\beta$ sufficiently small to satisfy the desired condition.  
Now applying Lemma~\ref{lemma:Bounding_Atypical_One_Sided} with 
$A=\pr{\pr{\calU_0^n\times\calU_1^n\times\calU_2^n\times\calZ^n}\backslash\calA_\epsilon^{(n)}}$ and $\delta=1$, leads to
    \begin{align}
    \bbP_\mu\pr{\bar{p}_{\text{atyp}}>2\cdot2^{-n\beta}}&\le\exp\pr{-2\cdot2^{n\pr{R_1-2\beta}}}.\label{eq:p_Atyp_1_final_One_Sided}
    \end{align}
To bound the probability of the typical term in \eqref{eq:Atyp_Typ_Trian_One_Sided}, for any fixed $u_0^n$, $u_2^n$, and $z^n$, we apply Lemma~\ref{lemma:Bounding_Typical} with $V_0=U_0$, $V_1=U_1$, $V_2=U_2$, and $\delta=2^{-n\beta}$, which leads to
\begin{align}
    &\bbP_{C_{1,n}}\pr{\bar{p}_{\text{typ}}(u_0^n,u_2^n,z^n)>1+2^{-n\beta}}\nonumber\\
    &\quad\le\exp\pr{-\frac{1}{3}2^{-n\pr{\bbI\pr{U_1;Z|U_0,U_2}+\epsilon+2\beta-R_1}}}.\label{eq:pTyp1_One_Sided}
\end{align}
Now we have,
\begin{align*}
    &\bbP_\mu\pr{\bbV\pr{\bar{p}_{Z^n|C_n},W_{Z|U_0U_2}^\on(\cdot|U_0^n,U_2^n)}>3\cdot2^{-n\beta}}\nonumber\\
    &\mathop\le\limits^{(a)}\bbP_\mu\pr{\bar{p}_{\text{atyp}}>2\cdot2^{-n\beta}}+\sum_{z^n\in\calZ^n}\bbP_\mu\pr{\bar{p}_{\text{typ}}(z^n)>1+2^{-n\beta}}\nonumber\\
    &=\bbP_\mu\pr{\bar{p}_{\text{atyp}}>2\cdot2^{-n\beta}}\nonumber\\
    &+\sum_{z^n\in\calZ^n}\bbP_{U_0^n,C_{1,n},U_2^n}\pr{\bar{p}_{\text{typ}}(U_0^n,U_2^n,z^n)>1+2^{-n\beta}}\nonumber\\
    &=\bbP_\mu\pr{\bar{p}_{\text{atyp}}>2\cdot2^{-n\beta}}\nonumber\\
    &+\sum_{z^n\in\calZ^n}\sum_{(u_0^n,u_2^n)\in\calU_0^n\times\calU_2^n}\bbP\pr{U_0^n=u_0^n,U_2^n=u_2^n}\nonumber\\
    &\times\bbP_{C_{1,n}}\pr{\bar{p}_{\text{typ}}(U_0^n,U_2^n,z^n)>1+2^{-n\beta}|U_0^n=u_0^n,U_2^n=u_2^n}\nonumber\\
    &\le\bbP_\mu\pr{\bar{p}_{\text{atyp}}>2\cdot2^{-n\beta}}\nonumber\\
    &+\sum_{z^n\in\calZ^n}\sum_{(u_0^n,u_2^n)\in\calU_0^n\times\calU_2^n}\hspace{-2mm}\bbP_{C_{1,n}}\pr{\bar{p}_{\text{typ}}(u_0^n,u_2^n,z^n)>1+2^{-n\beta}}\nonumber\\
    &\mathop\le\limits^{(b)}\exp\pr{-2\times2^{n\pr{R_1-2\beta}}}\nonumber\\
    &+\abs{\calZ^n}\abs{\calU_0^n}\abs{\calU_2^n}\exp\pr{-\frac{1}{3}2^{-n\pr{\bbI(U_1;Z|U_0,U_2)+\epsilon+2\beta-R_1}}},
\end{align*}where
\begin{itemize}
    \item[$(a)$] follows from \eqref{eq:Atyp_Typ_Trian_One_Sided}, the union bound, and since assuming
\begin{align}
    &\calB\triangleq\nonumber\\
    &\br{\sum_{z^n\in\calZ^n}\hspace{-1mm}W_{Z\lvert U_0U_2}^\on\!\pr{z^n|U_0^n,U_2^n}\sbr{\bar{p}_{\text{typ}}(z^n)-1}^+>\!2^{-n\beta}},\nonumber
\end{align}the set $\calB$ ensures that there exists at least one $z^n$ such that $\sbr{\bar{p}_{\text{typ}}(z^n)-1}^+>2^{-n\beta}$ and therefore
\begin{align}
    \calB\subseteq\bigcup_{z^n}\br{\bar{p}_{\text{typ}}(z^n)>1+2^{-n\beta}};\nonumber
\end{align}
\item[$(b)$] follows from \eqref{eq:p_Atyp_1_final_One_Sided} and \eqref{eq:pTyp1_One_Sided}.
\end{itemize}

Now we need to choose $\gamma_1$ and $\gamma_2$ such that \eqref{eq:Doble_Expo_One_Sided} holds. First, we choose $\epsilon$ and $\beta$ small enough such that the following quantity is positive, $\hat{\delta}\triangleq\min\br{R_1-2\beta,R_1-2\beta-\epsilon-\bbI(U_1;Z|U_0,U_2)}$. 
Since there is no constraint on $\beta$ and $\epsilon$ other than being positive and sufficiently small, such a choice is possible if $R_1>\bbI(U_1;Z|U_0,U_2)$. By selecting parameters $\gamma_2 < \hat{\delta}$ and $\gamma_1 < \beta$, and noting that our analysis holds in full generality, covering both cases $R_2 \le \bbI(U_2;Z|U_0)$ and $R_2 > \bbI(U_2;Z|U_0)$, we complete the proof of the region $\calR_2$ in Lemma~\ref{lemma:one_Sided_MAC_Res}. 

\end{appendices}

\bibliographystyle{IEEEtran}
\bibliography{IEEEabrv,bibfile}

@STRING{IEEE_J_IFS        = "{IEEE} Trans. Inf. Forensics Security"}

@STRING{IEEE_J_IT         = "{IEEE} Trans. Inf. Theory"}

@STRING{IEEE_J_SAIT       = "{IEEE} J. Sel. Areas Inf. Theory"}

@BOOK{BlochBarros,

AUTHOR    = {Bloch, M. R and Barros, J},
TITLE     = {Physical-Layer Security: From Information Theory to Security	Engineering},
PUBLISHER = {Cambridge University Press},
YEAR      = 2011,
EDITION   = {1rd},
ADDRESS   = {\hspace{-0.2cm}Cambridge, U.K},
}

@ARTICLE{BCC:IT78, 
author={Csisz\'ar, I and K\"{o}rner, J}, 
journal=IEEE_J_IT, 
title={Broadcast channels with confidential messages}, 
year={1978}, 
volume={24}, 
number={3}, 
pages={339-348}, 
month=may,
}

@ARTICLE{Slofstra20, 
author={Slofstra, W.}, 
journal={J. Amer. Math. Soc.}, 
title={Tsirelson’s problem and an embedding theorem for groups arising from non-local games}, 
year={2020}, 
volume={33},
number  = {1},
month   = Jan,
pages={1-56}, 
}

@article{QMAC,

author  = {Winter, A},
journal = IEEE_J_IT,
title   = {The capacity of the quantum multiple access channel},
volume  = {47},
number  = {7},
month   = nov,
year    = {2001},
pages   = {3059 - 3065},
}

@article{Boche14,

author  = {Boche, H and N\"{o}tzel, J},
journal = {Quantum Inf. Process.},
title   = {The classical-quantum multiple access channel with conferencing encoders and with common messages},
volume  = {13},
number  = {12},
month   = dec,
year    = {2014},
pages   = {2595 - 2617},
}

@article{Goldfeld18,

author  = {Goldfeld, Z and Cuff, P and Permuter, H. H},
title   = {Wiretap Channels with Random States Non-Causally Available at the Encoder},
journal = {available at \url{https://arxiv.org/abs/1608.00743v2}},
month   = Jan,
year    = {2018},
}

@article{Frey18,

author  = {Frey, M and Bjelakovi\'{c}, I and Sta\'{n}czak, S},
title   = {The {MAC} Resolvability Region, Semantic Security
and Its Operational Implications},
journal = {available at \url{https://arxiv.org/abs/1710.02342v2}},
month   = Jan,
year    = {2018},
}

@article{Hsieh08,

author = {Hsieh, M.-H and Devetak, I and Winter, A},
title  = {Entanglement-assisted capacity of quantum multiple-access channels},
journal = IEEE_J_IT,
volume  = {54},
number  = {7},
month   = Jul,
year    = {2008},
pages   = {3078-3090},
}

@article{Shi21,

author = {Shi, H and Hsieh, M.-H and Guha, S and Zhang, Z},
title  = {Entanglement-assisted capacity regions and protocol designs for quantum multiple-access channels},
journal = {npj Quantum Inf.},
volume  = {7},
number  = {1},
month   = May,
year    = {2021},
pages   = {1-9},
}

@article{Wyner,

author  = {Wyner, A. D},
title   = {The Wire-Tap Channel},
journal = {Bell System Technical Journal},
volume  = {57},
number  = {8},
month   = Oct,
year    = {1975},
pages   = {1355-1367},
}

@inproceedings{Aghaee25,

author    = {Aghaee, H and Deppe, C},
title     = {On the Interference Channel with Entangled Transmitters},
booktitle = {Proc. {IEEE} Info. Theory Workshop (ITW)},
address   = {Sydney, Australia},
month     = Sep,
year      = {2025},
}

@inproceedings{Cuff16,

author    = {Cuff, P},
title     = {Soft Covering with High Probability},
booktitle = {Proc. {IEEE} Int.  Symp. on Info. Theory (ISIT)},
address   = {Barcelona, Spain},
month     = Jul,
year      = {2016},
pages     = {2963-2967}
}

@inproceedings{Cuff15,

author    = {Cuff, P},
title     = {A Stronger Soft Covering Lemma and Applications},
booktitle = {Proc. {IEEE} Commun. Netw. Secur. (CNS)},
address   = {Florence, Italy},
month     = Sep,
year      = {2015},
pages     = {40-43}
}

@article{Hayashi06,

author  = {Hayashi, M},
title   = {General nonasymptotic and asymptotic formulas in channel resolvability and identification capacity and their application to the wiretap channel},
journal = IEEE_J_IT,
volume  = {52},
number  = {4},
month   = Dec,
year    = {2006},
pages   = {1562-1575},
}

@article{Bloch2013,

author  = {Bloch, M. R and Laneman, J. N},
title   = {Strong Secrecy From Channel Resolvability},
journal = IEEE_J_IT,
volume  = {59},
number  = {12},
month   = Dec,
year    = {2013},
pages   = {8077-8098},
}

@inproceedings{ISIT21,

author    = {ZivariFard, H and Bloch, M. R and Nosratinia, A},
title     = {Covert Communication via Non-Causal Cribbing from a Cooperative Jammer},
booktitle = {Proc. {IEEE} Int.  Symp. on Info. Theory (ISIT)},
address   = {Melbourne, Australia},
month     = Jul,
year      = {2021},
pages     = {202-207},
}

@article{MAC_GMS,

author  = {ZivariFard, H  and Wang, X},
title   = {Keyless Covert Communication Over Quantum {MAC}s with General Message Sets},
journal = IEEE_J_IT,
volume  = {71},
number  = {9},
month   = Sep,
year    = {2026},
pages   = {6752 - 6780},
}

@inproceedings{YassaeeMAWC,

author    = {Yassaee, M. H and Aref, M. R},
title     = {Multiple Access Wiretap channels with strong secrecy},
booktitle = {Proc. {IEEE} Info. Theory Workshop (ITW)},
address   = {Dublin, Ireland},
month     = Sep,
year      = {2010},
pages     = {1-5},
}

@inproceedings{Wiese2012,

author    = {Wiese, M and Boche, H},
title     = {An achievable region for the Wiretap multiple-access channel with common message},
booktitle = {{IEEE} Int.  Symp. on Info. Theory (ISIT)},
address   = {Cambridge, MA},
month     = Jul,
year      = {2012},
pages     = {249-253}
}

@article{Kramer_Book,

author  = {Kramer, G},
title   = {Topics in multi-user information theory},
journal = {Found. Trends Comm. Inf. Theory},
volume  = {4},
number  = {4-5},
year    = {2008},
pages   = {265-444},
}

@article{Concentration,

author  = {Hoeffding, W},
journal = {J. Am. Stat. Ass.},
title   = {Probability inequalities for sums of bounded random variables},
volume  = {58},
number  = {301},
year    = {1963},
pages   = {13-30},
}

@article{Concentration_2,

author  = {Janson, S},
journal = {Random Structures \& Algorithms},
title   = {Large deviations for sums of partly dependent random variables},
volume  = {24},
number  = {3},
year    = {2004},
pages   = {234-248},
}

@BOOK{ElGamalKim,

AUTHOR    = {El Gamal, A and Kim, Y.-H},
TITLE     = {Network Information Theory},
PUBLISHER = {Cambridge University Press},
YEAR      = {2012},
EDITION   = {First},
ADDRESS   = {\hspace{-0.2cm}Cambridge, U.K},
}

@BOOK{Concentration_Book,

AUTHOR    = {Dubhashi, D. P and Panconesi, A},
TITLE     = {Concentration of measure for the analysis of randomized algorithms},
PUBLISHER = {Cambridge University Press},
YEAR      = {2009},
EDITION   = {1rd},
ADDRESS   = {Cambridge, U.K},
}

@BOOK{Eggleston,

AUTHOR    = {Eggleston, H. G},
TITLE     = {Convexity},
PUBLISHER = {Cambridge University Press},
YEAR      = {1958},
EDITION   = {6rd},
ADDRESS   = {Cambridge, U.K},
}

@article{Cuff13,

author  = {Cuff, P},
title   = {Distributed Channel Synthesis},
journal = IEEE_J_IT,
volume  = {59},
number  = {11},
month   = Nov,
year    = {2013},
pages   = {7071-7096},
}

@article{Perurbation,

author  = {Gohari, A. A and Anantharam, V},
title   = {Evaluation of Marton’s inner bound for the general broadcast channel},
journal = IEEE_J_IT,
volume  = {58},
number  = {2},
month   = Feb,
year    = {2012},
pages   = {608-619},
}

@article{WillemsConferencing,

author  = {Willems, F. M. J and van der Meulen, E. C},
title   = {The discrete memoryless multiple access channel with partially cooperating encoders},
journal = IEEE_J_IT,
volume  = {29},
number  = {3},
month   = May,
year    = {1983},
pages   = {441-445},
}

@article{UziCribbing,

author  = {Pereg, U and Deppe, C and Boche, H},
title   = {The Quantum Multiple-Access Channel
With Cribbing Encoders},
journal = IEEE_J_IT,
volume  = {68},
number  = {6},
month   = Jun,
year    = {2022},
pages   = {3965 - 3988},
}

@article{Uzi_Entangled_MAC,

author  = {Pereg, U and Deppe, C and Boche, H},
title   = {The Multiple-Access Channel With Entangled Transmitters},
journal = IEEE_J_IT,
volume  = {71},
number  = {2},
month   = Feb,
year    = {2025},
pages   = {1096 - 1120},
}

@article{Uzi_Relay,

author  = {Pereg, U},
title   = {Quantum Relay Channels},
journal = IEEE_J_IT,
volume  = {71},
number  = {11},
month   = Nov,
year    = {2025},
pages   = {8595 - 8612},
}

@article{QMAC_Security,

author  = {Chou, R},
title   = {Private Classical Communication over Quantum Multiple-Access Channels},
journal = IEEE_J_IT,
volume  = {68},
number  = {3},
month   = Mar,
year    = {2022},
pages   = {1782 - 1794},
}

@article{GMAWCJamming,

author  = {Tekin, E and Yener, A},
title   = {The General {G}aussian Multiple-Access and Two-Way Wiretap Channels: Achievable Rates and Cooperative Jamming},
journal = IEEE_J_IT,
volume  = {54},
number  = {6},
month   = Jun,
year    = {2008},
pages   = {2735-2751},
}

@article{ZivSemantic2016,

author  = {Goldfeld, Z and Cuff, P and Permuter, H. H},
title   = {Semantic-Security Capacity for Wiretap Channels of Type {II}},
journal = IEEE_J_IT,
volume  = {62},
number  = {7},
month   = Jul,
year    = {2016},
pages   = {3863-3879},
}

@article{Notzel201,

author  = {N\"{o}tzel, J},
journal = IEEE_J_SAIT,
title   = {Entanglement-enabled communication},
volume  = {1},
number  = {2},
month   = Aug,
year    = {2020},
pages   = {401-415},
}

@inproceedings{Notzel202,

author  = {N\"{o}tzel, J and DiAdamo, S},
title     = {Entanglement-enabled communication for the Internet of Things},
booktitle = {Proc. Int. Conf. Comput. Inf. Telecomm. Syst. (CITS)},
address   = {Hangzhou, China},
month     = Oct,
year      = {2020},
pages     = {1-6}
}

@article{Yun23,

author  = {Yun, J and Rai, A and Bae, J},
title   = {Non-local and quantum advantages in network coding for multiple access channels},
journal = {available at \url{https://arxiv.org/abs/2304.10792}},
month   = Jan,
year    = {2023},
}

@article{Mermin90,
  title = {Simple unified form for the major no-hidden-variables theorems},
  author  = {Mermin, N. D},
  journal = {Phys. Rev. Lett.},
  volume = {65},
  pages   = {3373-3376},
  year = {1990},
}

@article{Peres90,
  title = {Incompatible results of quantum measurements},
  author  = {Peres, A},
  journal = {Phys. Lett. A},
  volume = {151},
  pages   = {107-108},
  year = {1990},
}

@article{Aravind03,

author  = {Aravind, P. K},
title   = { A simple demonstration of bell’s theorem involving two observers and no probabilities or inequalities},
journal = {available at \url{https://arxiv.org/abs/quant-ph/0206070}},
month   = Jan,
year    = {2003},
}

@inproceedings{Doolittle22,

author    = {Doolittle, B and Chitambar, E and Leditzky, F},
title     = {Nonclassical behaviors of two-sender multiple access channels},
booktitle = {Proc. APS March Meeting Abstr.},
address   = {available at \url{https://meetings.aps.org/Meeting/MAR22/Session/G00.182}},
year      = {2022},
}

@inproceedings{Fawzi22,

author    = {Fawzi, O and Ferm\'e, P},
title     = {Beating the sum-rate capacity of the binary adder channel with non-signaling correlations},
booktitle = {Proc. {IEEE} Int. Symp. Inf. Theory (ISIT)},
address   = {Melbourne, Australia},
month     = Jun,
year      = {2022},
pages     = {2750-2755},
}

@article{Brassard05,
author  = {Brassard, G and Broadbent, A and Tapp, A},
title   = {Quantum pseudo-telepathy},
journal = {Found. Phys.},
volume  = {35},
number  = {11},
month   = Nov,
year    = {2005},
pages   = {1877-1907},
}

@inproceedings{Interference_Entangled,

author    = {Hawellek, J and  Mohan, A and Aghaee, H and Deppe, C},
title     = {The Interference Channel with Entangled Transmitters},
booktitle = {Proc. {IEEE} Int.  Symp. on Info. Theory (ISIT)},
address   = {Ann Arbor, MI},
month     = Jun,
year      = {2025},
pages     = {2945-2949},
}

@article{Seshadri23,
author  = {Seshadri, A and  Leditzky, F and Siddhu, V and Smith, G},
title   = {On the Separation of Correlation-Assisted Sum Capacities of Multiple Access Channels},
journal = IEEE_J_IT,
volume  = {69},
number  = {9},
month   = Sep,
year    = {2023},
pages   = {5805-5844},
}

@article{Leditzky_20,

author  = {Leditzky, F and Alhejji, M.~A and Levin, J and Smith, G},
journal = {Nature Commun.},
title   = {Playing games with multiple access channels},
volume  = {11},
number  = {1},
month   = Mar,
year    = {2020},
pages   = {1-5},
}

@article{Loock20,

author = {van Loock, P and Alt, W and Becher, C and Benson, O and Boche, H and Deppe, C and Eschner, J and H\"ofling, S and Meschede, D and Michler, P and Schmidt, F and Weinfurter, H},
journal = {Adv. Quantum Technol.},
title = {Extending Quantum Links: Modules for Fiber- and Memory-Based Quantum Repeaters},
volume = {3},
number = {11},
pages = {1900141},
year = {2020},
}

@BOOK{Bassoli21,

AUTHOR    = {Bassoli, R and Boche, H and Deppe, C and Ferrara, R and Fitzek, F. H. P and Janssen, G and Saeedinaeeni, S},
TITLE     = {Quantum Communication Networks (Foundations in Signal Processing, Communications and Networking)},
PUBLISHER = {Springer},
YEAR      = {2021},
ADDRESS   = {\hspace{-0.2cm}Cham, Switzerland},
}

@article{Quek17,
  title = {Quantum and superquantum enhancements to two-sender, two-receiver channels},
  author  = {Quek, Y and Shor, P. W},
  journal = {Phys. Rev. A, Gen. Phys.},
  volume = {95},
  issue = {5},
  pages = {052329},
  numpages = {11},
  year = {2017},
  month = May,
}

@article{Clauser69,
  title = {Proposed experiment to test local hidden-variable theories},
  author  = {Clauser, J. F and Horne, M. A and Shimony, A and Holt, R. A},
  journal = {Phys. Rev. Lett.},
  volume = {23},
  issue = {15},
  pages = {880-884},
  year = {1969},
  month = Oct,
}

@article{Popescu94,
  title = {Quantum nonlocality as an axiom},
  author  = {Popescu, S and Rohrlich, D},
  journal = {Found. Phys.},
  volume = {24},
  issue = {3},
  pages = {379-385},
  year = {1994},
  month = Mar,
}

@inproceedings{Perturbation_Quantum,

author    = {Grace, M. R and Guha, S},
title     = {Perturbation theory for quantum information},
booktitle = {Proc. {IEEE} Info. Theory Workshop (ITW)},
address   = {Mumbai, India},
month     = Nov,
year      = {2022},
pages     = {500-505},
}

@article{Ziv_WTC_with_NonCausaul_CSI,

author  = {Goldfeld, Z and Cuff, P and Permuter, H. H},
journal = IEEE_J_IT,
title   = {Wiretap Channels With Random States Non-Causally Available at the Encoder},
volume  = {66},
number  = {3},
month   = Mar,
year    = {2020},
pages   = {1497-1519},
}

@article{ForenPaper,

author    = {Zivari-Fard, H and  Akhbari, B and Ahmadian-Attari, M and Aref, M. R},
journal = IEEE_J_IFS,
title   = {Imperfect and Perfect Secrecy in Compound Multiple Access Channel With Confidential Message},
volume  = {11},
number  = {6},
month   = Jun,
year    = {2016},
pages   = {1239 -1251},
}

@article{MultiCase_Paper,

author    = {ZivariFard, H and  Bloch, M.~R and Nosratinia, A},
journal = IEEE_J_IFS,
title   = {Two-Multicast Channel With Confidential Messages
},
volume  = {16},
month   = Jan,
year    = {2021},
pages   = {2743 -2758},
}

@article{XU24,

author  = {Xu, H and Wong, K.-K and Caire, G},
title   = {A New Achievable Region of the ${K}$-User {MAC} Wiretap Channel With Confidential and Open Messages Under Strong Secrecy},
journal = IEEE_J_IT,
volume  = {70},
number  = {12},
month   = Dec,
year    = {2024},
pages   = {9123-9151},
}

@inproceedings{effectivesecrec,

author    = {Hou, J and  Kramer, G},
title     = {Effective secrecy: Reliability, confusion and stealth},
booktitle = {Proc. {IEEE} Int.  Symp. on Info. Theory (ISIT)},
address   = {HI, USA},
month     = Jul,
year      = {2014},
pages     = {601-605}
}

@inproceedings{ICC2014,

author    = {Zivari-Fard, H and  Akhbari, B and Ahmadian-Attari, M and Aref, M. R},
title     = {Compound Multiple Access Channel with Confidential Messages},
booktitle = {Proc. {IEEE} Int. Conf. on Comm. (ICC)},
address   = {Sydney, Australia},
month     = Jun,
year      = {2014},
pages     = {1922-1927}
}

@article{IETpaper,

author    = {Zivari-Fard, H and  Akhbari, B and Ahmadian-Attari, M and Aref, M. R},
title     = {Multiple Access Channel with Common Message and Secrecy Constraint},
journal = {IET Commun.},
volume  = {10},
number  = {1},
month   = Feb,
year    = {2016},
pages   = {98-110},
}

@inproceedings{WieseBoche,

author    = {Wiese, M and  Boche, H},
title     = {Strong Secrecy for Multiple Access Channels},
booktitle = {Information Theory, Combinatorics, and Search Theory},
address   = {Springer},
year      = {2013},
pages     = {71-122}
}

\end{document}